\documentclass[11pt]{article}
\usepackage[utf8]{inputenc}

\usepackage{graphicx}
\usepackage{xcolor}
\usepackage{amsmath,amsthm,amssymb}
\usepackage[margin=1in]{geometry}
\usepackage[pdftex,colorlinks=true,linkcolor=red,citecolor=red,urlcolo
r=black]{hyperref}
\usepackage[]{hyperref}
\usepackage{comment}
\usepackage[normalem]{ulem}
\usepackage{cite}
\usepackage{stmaryrd}

\newcommand{\R}{\mathbb{R}}
\newcommand{\N}{\mathbb{N}}
\newcommand{\C}{\mathbb{C}}

\def\A{{\bf A}}
\def\B{{\bf B}}
\def\C{{\bf C}}

\def\H{{\bf H}}
\def\I{{\bf I}}

\def\L{{\bf L}}
\def\M{{\bf M}}
\def\N{{\bf N}}

\def\R{{\bf R}}
\def\S{{\bf S}}
\def\T{{\bf T}}
\def\U{{\bf U}}
\def\V{{\bf V}}

\def\X{{\bf X}}
\def\Y{{\bf Y}}
\def\Z{{\bf Z}}

\def\a{{\bf a}}
\def\b{{\bf b}}
\def\c{{\bf c}}

\def\e{{\bf e}}

\def\i{{\bf i}}
\def\j{{\bf j}}
\def\k{{\bf k}}
\def\l{{\bf l}}
\def\m{{\bf m}}

\def\p{{\bf p}}
\def\q{{\bf q}}
\def\r{{\bf r}}
\def\s{{\bf s}}
\def\t{{\bf t}}

\def\x{{\bf x}}
\def\y{{\bf y}}
\def\z{{\bf z}}

\def\0{{\bf 0}}
\def\1{{\bf 1}}
\def\2{{\bf 2}}
\def\3{{\bf 3}}
\def\4{{\bf 4}}
\def\5{{\bf 5}}
\def\6{{\bf 6}}
\def\7{{\bf 7}}
\def\8{{\bf 8}}
\def\9{{\bf 9}}

\def\bq{\bar{\bf q}}

\def\ba{\begin{array}}
\def\ea{\end{array}}
\def\be{\begin{equation}}
\def\ee{\end{equation}}
\def\bi{\begin{itemize}}
\def\ei{\end{itemize}}

\newcommand{\bea}{\begin{eqnarray}}
\newcommand{\eea}{\end{eqnarray}}
\newcommand{\bes}{\begin{equation*}}
\newcommand{\ees}{\end{equation*}}
\newcommand{\beas}{\begin{eqnarray*}}
\newcommand{\eeas}{\end{eqnarray*}}
\def\bn{\begin{eqnarray}}
\def\en{\end{eqnarray}}
\def\bu{\begin{enumerate}}
\def\eu{\end{enumerate}}
\def\bp{\begin{prop}}
\def\ep{\end{prop}}
\def\bl{\begin{lem}}
\def\el{\end{lem}}
\def\bc{\begin{cor}}
\def\ec{\end{cor}}
\def\bt{\begin{thm}}
\def\et{\end{thm}}
\def\bdf{\begin{dfn}}
\def\edf{\end{dfn}}

\makeatletter
\newtheorem*{rep@theorem}{\rep@title}
\newcommand{\newreptheorem}[2]{%
\newenvironment{rep#1}[1]{%
 \def\rep@title{#2 \ref{##1} (restated)}%
 \begin{rep@theorem}}%
 {\end{rep@theorem}}}
\makeatother

\allowdisplaybreaks

\newtheorem{thm}{Theorem}[section]
\newtheorem*{thm*}{Theorem}
\newtheorem{cor}[thm]{Corollary}

\newtheorem{lem}[thm]{Lemma}
\newtheorem*{lem*}{Lemma}

\newtheorem{prop}[thm]{Proposition}

\newtheorem{dfn}{Definition}[section] 
\newreptheorem{cor}{Corollary}

\usepackage{pgfplots}
\usepackage{authblk}

\newbox\bigstrutbox
\setbox\bigstrutbox=\hbox{\vrule height12.5pt depth5pt width0pt}
\def\bigstrut{\relax\ifmmode\copy\bigstrutbox\else\unhcopy\bigstrutbox
\fi}

\newbox\Bigstrutbox
\setbox\Bigstrutbox=\hbox{\vrule height17.5pt depth5pt width0pt}
\def\Bigstrut{\relax\ifmmode\copy\Bigstrutbox\else\unhcopy\Bigstrutbox
\fi}

\numberwithin{equation}{section}

\title{Normalization of 
Quaternionic Polynomials in Coordinate-Free Quaternionic Variables
in Conjugate-Alternating Order}

\author[1]{Hongbo Li\thanks{hli@mmrc.iss.ac.cn}}
\author[1]{Zhengyang Wang\thanks{wangzhengyang22@mails.ucas.ac.cn}}
\author[1]{Yue Liu\thanks{liuyue10@mails.ucas.ac.cn}}
\author[1]{Lei Huang\thanks{lhuang@mmrc.iss.ac.cn}}
\author[1]{Changpeng Shao\thanks{changpeng.shao@amss.ac.cn}}

\affil[1]{State Key Laboratory of Mathematical Sciences,
AMSS, UCAS, Chinese Academy of Sciences, Beijing 100190, China}

\begin{document}

\maketitle

\begin{abstract}
Quaternionic polynomials occur naturally in applications of quaternions
in science and engineering, and normalization of quaternionic polynomials 
is a basic manipulation. Once a Gr\"obner basis is certified for the
defining ideal $\cal I$ of the quaternionic polynomial algebra, the normal
form of a quaternionic polynomial can be computed by routine top reduction
with respect to the Gr\"obner basis.
In the literature, a Gr\"obner basis under the conjugate-alternating order
of quaternionic variables was conjectured
for $\cal I$ in 2013 \cite{li-huang-liu2013normalization}, 
but no readable and convincing proof was found. 

In this paper, we present the first readable certification of the conjectured
Gr\"obner basis. The certification is based on several novel techniques for reduction
in free associative algebras, which enables to not only make reduction to S-polynomials
more efficiently, but also reduce the number of S-polynomials needed for the certification.
\vspace{.2cm}

\end{abstract}

\section{Introduction}
\setcounter{equation}{0}
\label{sect:intro}

Quaternions are an important mathematical tool for science and engineering. For one thing,
in history, Gibb's vector algebra grows out of quaternions by separating the scalar part
from the pure imaginary part (vector) of the quaternionic product of vectors.
Unit quaternions form the group Spin(3), and
generate rotations in three-dimensional Euclidean space and 
four-dimensional Euclidean space \cite{altmann2005rotations}. 
Unit quaternions also generate one-qubit gates via the
isomorphism between $SU(2)$ and $Spin(3)$. This makes quaternions very useful not only
in geometry related application fields, such as
in computer graphics, robotics, geometric reasoning \cite{li2015symbolic}, and
color image processing \cite{moxey2003, ell2006, zou2016, lan2016, meng2020}, but also in
quantum computation. When a quaternion is presented as a
linear combination of some quaternionic variables, the quaternionic product of this 
quaternion with other quaternions yields a quaternionic polynomial in quaternionic variables.
Such polynomials occur naturally in applications of quaternions.

A {\it quaternionic variable} is defined to be a linear combination of the basis $1, \i, \j, \k$
of quaternions: $\q=u+x\i+y\j+z\k$, 
with the coefficient variables $u,x,y,z$ taking values in a field $\mathbb K$.
The coefficients are called the {\it coordinate variables} of the quaternionic variable. A 
quaternionic variable usually occurs in {\it coordinate-free form}, namely, letter $\q$. 
It is easy to deduce that each coordinate variable can be represented as a
non-commutative polynomial in the letters $\q, \bar{\q}, \i, \j, \k$, where $\bar{\q}$ is the
(quaternionic) conjugate of $\q$; say it is $v=v_c(\q, \bar{\q}, \i, \j, \k)$ for $v=u, x,y,z$.
The polynomials $u_c, x_c, y_c, y_c$ are called the {\it coordinate polynomials} of 
quaternionic variable $\q$. A {\it coordinate-free quaternionic polynomial} is a
$\mathbb K$-linear combination of finitely many quaternionic products of the
letters in alphabet ${\cal A}=\{\q_l, \bq_l, \i, \j, \k\,|\, l=1,\ldots,n\}$.

In polynomial algebra ${\mathbb K}[x_1, \ldots, x_n]$, by commutativity a monomial in
normal form is of the form $\lambda x_{i_1}^{r_1} x_{i_2}^{r_2} \ldots x_{i_m}^{r_m}$, where the $r_i>0$,
under the order of variables
$x_1\prec x_2\prec \ldots \prec x_n$. We say the normal form is {\it strictly increasing},
because $i_1<i_2<\ldots< i_m$. Polynomial normalization is a basic task in manipulating
polynomials. In coordinate-free quaternionic polynomial algebra, 
normalization is just as important as in the commutative case.

The following is a rigorous definition of the {\it coordinate-free quaternionic polynomial algebra}
with alphabet $\cal A$. Let ${\mathbb K}\langle {\cal A} \rangle$ be the free associative
$\mathbb K$-algebra over alphabet $\cal A$, and let $\cal I$ be the two-sided ideal generated by 
the following elements: (1) multiplication table of the basis elements $\i, \j, \k$, e.g.,
$\i^2-1, \i\j-\k$, etc.; (2) coordinatization of the quaternionic variables and their 
conjugates, e.g., for $\q, \bq$, $\q-(u_c+x_c \i+y_c \j+z_c \k)$ and
$\bq-(u_c-x_c \i-y_c \j-z_c \k)$, where $u_c, x_c, y_c, z_c$ are 
the coordinate polynomials in letters $\q, \bq, \i, \j, \k$; (3) commutativity between
any coordinate polynomial and a letter of $\cal A$. The {\it normal form} of
a quaternionic polynomial is the unique (non-commutative) polynomial 
in the free associative algebra ${\mathbb K}\langle {\cal A} \rangle$ that equals
the quaternionic polynomial modulo $\cal I$. 

Computation of the normal form requires a Gr\"obner basis of the {\it defining ideal} 
$\cal I$. Computing Gr\"obner basis in a non-commutative ring has a long history in the symbolic
computation society. Following Buchberger's original theory
\cite{buchberger1965algorithm} and
Bergman's diamond lemma \cite{bergman1978diamond}, 
the true concept of Gr\"obner basis for free associative algebras and
free monoid rings was introduced by F. Mora \cite{mora1985grobner}.
Subsequently, T. Mora \cite{mora1988,mora1994introduction} systematically established 
the theory of non-commutative Gr\"obner basis. Later on, 
volumes of algorithms
were developed \cite{rody1990,
green1998, cojocaru1999non, green2000multiplicative, giesbrecht2002},
and collected into
monographs \cite{mora1994introduction, ufnarovski1998introduction,
li2002, bueso2003}, and computer algebra systems such as 
Plural \cite{levandovskyy2003}, ApCoCoA \cite{apcocoa} and Letterplace \cite{scala09}.

As to the Gr\"obner basis computation involving quaternionic polynomials, 
in \cite{DGS10} the following ring of one-variable quaternionic polynomials was studied:
for two quaternionic monomials $\q^i \a_i$ and $\q^j \a_j$, where $\q$ is the 
quaternionic variable and $\a_i, \a_j$ are quaternionic coefficients, their product
is defined by
\be
(\q^i \a_i)*(\q^j \a_j):=\q^{i+j} \a_i\a_j.
\label{DGS10:product}
\ee
For two quaternionic polynomials of the form 
$\sum_i \lambda_i \q^i \a_i$ and $\sum_j \lambda_j \q^j \a_j$ respectively, where
the $\lambda$'s are scalars in $\mathbb K$, their product is the $\mathbb K$-linear
expansion of the products between the quaternionic monomials. Obviously this product
is not the quaternionic product, although the participating monomials are 
quaternionic ones. 

In \cite{H00}, another product was defined for multivariate quaternionic polynomials
and the corresponding Gr\"obner basis computation problem was studied: for two 
quaternionic monomials $\a_\alpha \q_1^{\alpha_1}\cdots \q_r^{\alpha_r}$ and
$\a_\beta \q_1^{\beta_1}\cdots \q_r^{\beta_r}$, where $\alpha=(\alpha_1, \ldots, \alpha_r)$
is a multi-index, the $\q_l$ are quaternionic variables, and the $\a$'s are quaternionic
coefficients, their product is defined as follows:
\be
(\a_\alpha \q_1^{\alpha_1}\cdots \q_r^{\alpha_r})*
(\a_\beta \q_1^{\beta_1}\cdots \q_r^{\beta_r}):=
\a_\alpha \a_\beta \q_1^{\alpha_1+\beta_1}\cdots \q_r^{\alpha_r+\beta_r}.
\ee
For two quaternionic polynomials of the form 
$\sum_{\alpha} \lambda_\alpha
\a_\alpha \q_1^{\alpha_1}\cdots \q_r^{\alpha_r}$ and 
$\sum_{\beta} \lambda_\beta \a_\beta \q_1^{\beta_1}\cdots \q_r^{\beta_r}$,
where the $\lambda$'s are scalars in $\mathbb K$, their product is the $\mathbb K$-linear
expansion of the products between the quaternionic monomials. Nor is this product
the quaternionic product. 

The first investigation of the Gr\"obner basis of the defining ideal $\cal I$ of
coordinate-free quaternionic polynomial algebra seems to be 
\cite{li-huang-liu2013normalization}, where for small number of quaternionic variables
$n=1, \ldots, 6$, the corresponding Gr\"obner basis elements of small degree ($\leq 8$)
under the order of variables
\be
\q_1\prec\bq_1\prec \q_2\prec\bq_2\prec \ldots \q_n\prec\bq_n\prec \i\prec \j\prec \k,
\label{order:alternate}
\ee
called the {\it conjugate-alternating order}, were computed, and then after simplification
and observation, the Gr\"obner basis elements for arbitrary number of quaternionic variables
and of arbitrary degree were conjectured. 

In the PhD dissertation \cite{liu2015}, a
proof based on segmenting quaternionic monomials into submonomials and then introducing
new letters to represent the submonomials was proposed. Let the alphabet of new letters
be $\cal V$. The number of letters in $\cal V$ does not exceed 30, no matter how big the number
of old variables is.
In \cite{liu2015} a method of delicate segmentation of quaternionic monomials
and an ordering of the new letters was proposed, by which the certification of the 
conjectured Gr\"obner basis is converted to the verification that every element in the
conjectured Gr\"obner basis of the new letters can be reduced to zero by the
conjectured Gr\"obner basis of the old letters. Extensive computer aided verifications
are needed in the proof, which is composed of over 1000 case-by-case reductions. 
In \cite{liu2024}, the verification of the reduction-to-zero property
of the conjectured Gr\"obner basis of the new letters was done with extensive computer help.
It remains to find a short and readable certification of the conjectured Gr\"obner basis
in \cite{li-huang-liu2013normalization}. Indeed, verifying the conversion
in \cite{liu2015} took one-month computer time, and verifying the reduction-to-zero property
in \cite{liu2024} took several hours. 

In this paper, we propose the first short and readable proof of the conjecture made in 
\cite{li-huang-liu2013normalization}. The proof is less than 40 pages, and the 
reduction procedures are readable.
Although computer is needed in the final stage of reduction in many cases,
this equipment reliance does not influence 
the readability. The key idea behind the improvement is that we completely discard
the idea of monomial segmentation proposed in \cite{liu2015}. Instead, we seek to
enlarge the conjectured Gr\"obner basis from a reduced one to a non-reduced one,
by allowing more powerful commutativity tools to be used in non-commutative 
polynomial reduction, and in particular, for quaternionic polynomials.

Similar to the generators of the defining ideal $\cal I$ of a
coordinate-free quaternionic polynomial algebra, the certified reduced Gr\"obner basis 
consists of three kinds of elements: 
\bu
\item multiplication table of the basis elements;

\item conjugate-defining equations from the 
coordinatization of letters $\q,\bq$, cf. (\ref{def:qqbar:1}); 

\item commutativity between the {\it bracket} of any  
monomial $\M$ and any letter $\b$, where for monomial $\M=\a_1\ldots \a_m$ with the $\a_i$
being letters, bracket $[\M]=\M+\bar{\M}
=\a_1\ldots \a_m+\overline{\a_m}\ldots \overline{\a_1}$
is twice the scalar part of $\M$. 
\eu
Normalizing a quaternionic
polynomial can be done by top reduction with respect the Gr\"obner basis elements 
in arbitrary way. A more effective approach is to enlarge the third kind of elements in the
Gr\"obner basis by including the commutativity between the bracket of any
monomial and any other monomial. Any quaternionic monomial in normal form is a submonomial of the
monomial depicted in Fig. \ref{normal:2} of this paper, where the letters of the monomial
form a pair of parallel and alternating non-decreasing subsequences, instead of only a single
increasing sequence in the commutative case.  

The technical contributions of this paper include the following:
\bi
\item Four new reduction techniques for free associative algebras:
subsequence-led reduction, subsequence-dominated reduction, 
bottom-letter controlled reduction, and substitutional reduction. These techniques allow
more flexible reduction of a non-commutative polynomial by locally increasing the order
of the leading term.

\item Before making reduction to the S-polynomials of a Gr\"obner basis candidate, first
extend the Gr\"obner basis candidate to a larger one with more general symmetries that allow
more powerful reduction techniques. Every newly included element should be verified to be
reduced to zero by the old ones. The enlarged Gr\"obner basis candidate is to be used
to make reduction to the S-polynomials of the original Gr\"obner basis candidate.

\item Reduce the number of S-polynomials used to certify 
the Gr\"obner basis candidate by selecting only the {\it clear S-polynomials} proposed
in this paper. This technique is valid for all free associative algebras where a reduced
Gr\"obner basis candidate has the property that all its elements have degree $>1$. For 
quaternionic polynomial algebra, the number of S-polynomials used to certify the conjectured
Gr\"obner basis is reduced by about half with the clear S-polynomial technique.
\ei

The content of this paper is arranged as follows. Section \ref{sect:prem} introduces
the basics of quaternionic polynomial algebra, and presents the Main Theorem of this paper,
which claims that the conjectured Gr\"obner basis {\bf BG} is 
a reduced Gr\"obner basis of the defining ideal
$\cal I$ of quaternionic polynomial algebra. 
Section \ref{sect:reduct} proposes several novel reduction techniques for free-associative
algebras, together with a preliminary extension $\overline{\bf BG}$
of the conjectured Gr\"obner basis {\bf BG}. 
Section \ref{sect:ext} made further extension of $\overline{\bf BG}$ by including all elements
of the form $\a[\T]-[\T]\a$, where $\a$ is an arbitrary letter and $\T$ is an arbitrary monomial.
Section \ref{sect:gb} proposes a simplified certification technique for 
reduced Gr\"obner basis in free associative algebras, called clear S-polynomial reduction.
With the techniques provided by all these sections, Section \ref{sect:proof} finishes the proof 
of the Main theorem. Section \ref{sect:conc} concludes this paper with a short discussion
of the Gr\"obner basis in conjugate-separating order.

In this paper, we use lower-case letters to represent (quaternionic) polynomials and scalars,
use bold-faced numbers and bold-faced lower-case letters to represent letters 
(quaternionic variables and basis elements), and use bold-faced upper-case letters to represent
monomials and sequences.

\section{From free associative algebra to quaternionic polynomial algebra}
\label{sect:prem}

Let $\cal A$ be a finite set of letters, and let the juxtaposition of letters denote the
{\it free associative product} among the letters, namely, the product is associative, and two
products of letters are equal if and only if as sequences of letters they are identical.
Every product of letters is called a (monic) {\it monomial}, with unit 1 denoting the
empty sequence. Often we simply call a monomial
in free associative algebra {\it a sequence of letters}.
A monomial is said to be {\it non-decreasing} if its sequence of letters
is non-decreasing. 
All monomials form a monoid $\langle {\cal A}\rangle$. A {\it monoid ideal} of 
$\langle {\cal A}\rangle$ is a subset that is closed under the left product and right product
by elements of $\langle {\cal A}\rangle$. For a subset $M$ of $\langle {\cal A}\rangle$, the
monoid ideal it generates is denoted by $\langle M \rangle$.

Let $\mathbb K$ be a field of characteristic $\neq 2$. The 
{\it free associative ${\mathbb K}$-algebra} ${\mathbb K}\langle {\cal A}\rangle$ over $\cal A$
is composed of polynomials equipped with the multilinear extension of the 
free associative product, where each {\it polynomial} is a ${\mathbb K}$-linear combination of
monomials. Each term of a polynomial is called a {\it monomial}, and
is a ${\mathbb K}$-scaling of 
a monic monomial. In particular, the unit 1 of monomial monoid
can be chosen to be the unit of field ${\mathbb K}$. 
The {\it degree}, or {\it length}, of a monomial $\A$, refers to the number of letters
(including multiplicity) in the monic monomial after removing the coefficient, denoted by $|\A|$.
The {\it degree} of a polynomial is the maximal degree of its terms.

An {\it ideal} of ${\mathbb K}\langle {\cal A}\rangle$ is a set that is closed under the
left multiplication and the right multiplication by elements of 
${\mathbb K}\langle {\cal A}\rangle$. Given a set $G\subseteq {\mathbb K}\langle {\cal A}\rangle$,
the {\it ideal generated by $G$}, denoted by $\langle G\rangle$, 
is composed of elements of the form
$\sum_i \A_i g_i \B_i$, where the summation is finite, 
$\A_i, \B_i$ are monomials, and $g_i\in G$.

In this paper we only consider the degree lexicographic order ``deglex" among monomials.
Let there be a total order ``$\prec$" among the letters of $\cal A$. 
For two monomials $\A, \B$, under the {\it degree lexicographic order}, $\A\prec \B$ if and only if
either $|\A|<|\B|$, or $|\A|=|\B|$ but lexicographically $\A\prec \B$
as sequences of letters. The {\it leading term} of a polynomial is the term with 
maximal order. In the leading term, the coefficient is called the {\it leading coefficient}, and
the monic monomial is also said to be {\it leading}. The order ``deglex" can be naturally
extended to polynomials. It is always preserved by the left or right 
multiplication with a polynomial.

There is another partial order ``$\prec_M$" among monomials, called
{\it monoid order}. For two monic monomials
$\A, \B$, $\A\prec_M \B$ if and only if every letter of $\A$ is $\prec$ every letter of $\B$.

Given set $G\subseteq {\mathbb K}\langle {\cal A}\rangle$, the (top) {\it reduction} of
a polynomial $f\in {\mathbb K}\langle {\cal A}\rangle$ with respect to $G$, 
is a procedure of obtaining another polynomial $h$ together with
finitely many polynomials of the form
$\A_i g_i \B_i$, where each $g_i\in G$, the
$\A_i, \B_i$ are monomials, such that $h\preceq f$ and every $\A_i g_i \B_i\preceq f$, and 
\be
f=h+\sum_i \A_i g_i \B_i. \label{def:reduce}
\ee
$f$ is said to be (top) {\it reduced} to 
$h$ with respect to $G$. 

Given two monomials $\A, \B$, $\A$ is said to be reducible with respect to $\B$, denoted by
$\A\,|\,\B$, if $\A$ is a {\it submonomial} of $\B$, namely, 
there exist monomials $\L,\R$ such that
$\B=\L \A \R$. A polynomial $f$ is said to be {\it reducible} with respect to a set 
$G$ of polynomials, if there exists an element $g\in G$, such that
the leading term ${\rm T}(f)$ of $f$ is reducible with respect to 
the leading term ${\rm T}(g)$ of $g$. In terms of (\ref{def:reduce}), $f$ is
reducible with respect to $G$ if and only if there exists a polynomial $h\prec {\rm T}(f)$ such that
$f$ is reduced to $h$ with respect to $G$. A submonomial is often called a {\it subsequence
of letters} in free associative algebra.

Given set $G\subset {\mathbb K}\langle {\cal A}\rangle$, the {\it normal form} of
a polynomial $f\in {\mathbb K}\langle {\cal A}\rangle$ with respect to $G$, is the unique
polynomial $h\in {\mathbb K}\langle {\cal A}\rangle$ such that $f$ is reduced to 
$h$ with respect to $G$, and $h=0$ if and only if $f\in \langle G\rangle$.

A {\it Gr\"obner basis} of ideal $\langle G\rangle$ is a subset $\bf GB$ of $\langle G\rangle$, 
such that $\bf GB$ generates $\langle G\rangle$, and the monoid ideal generated by
the leading monic monomials of the elements of $\langle G\rangle$ is generated by the
leading monic monomials of the elements of $\bf GB$. Gr\"obner basis $\bf GB$ is 
said to be {\it reduced}, if any two elements of $\bf GB$ are not reducible with
respect to each other. If a Gr\"obner basis of $\langle G\rangle$ is found, then for any
polynomial $f$, its normal form with respect to $G$ can be computed by making
reduction {\it arbitrarily} with respect to the elements of the Gr\"obner basis.

Now consider quaternionic polynomials in quaternionic variables.
as follows:
\be
\q_l=u_l+x_l\i+y_l\j+z_l\k, \hskip .2cm
\bq_l=u_l-x_l\i-y_l\j-z_l\k,
\label{def:qqbar}
\ee
where $\i, \j, \k$ are the standard pure quaternionic basis, and
the $u_l, x_l, y_l, z_l$ are coordinate variables taking values in $\mathbb{K}$.
(\ref{def:qqbar}) is the {\it coordinatization} of quaternionic variables.

Usually we maintain the integral form of $\q_l, \bq_l$ by expressing the 
coordinate variables as polynomials in $\q_l, \bq_l$ and the basis vectors, which can be
easily verified to be the following:
\be\ba{lllll}
  u_l &=& (\q_l+\bar{\q}_l)/2 &=& [\q]/2, \\

  x_l &=& -(\i\q_l-\bar{\q}_l\i)/2 &=& -[\i\q_l]/2,\\

  y_l &=& -(\j\q_l-\bar{\q}_l\j)/2 &=& -[\j\q_l]/2,\\

  z_l &=& -(\k\q_l-\bar{\q}_l\k)/2 &=& -[\k\q_l]/2,
\ea
\label{coordinate-free}
\ee
where the {\it bracket} of a quaternionic polynomial $f$ is twice its scalar part, 
denoted by $[f]$.
 
The expressions in (\ref{coordinate-free}) 
are called the {\it coordinate polynomials} of the quaternionic
variables $\q_l, \bq_l$.
Substituting (\ref{coordinate-free}) into (\ref{def:qqbar}), and making
simplification by the multiplication table of the basis letters, we get that
(\ref{def:qqbar}) is equivalent to the following single equation: 
\be
2\bq_l=-\q_l-\i\q_l\i-\j\q_l\j-\k\q_l\k.
\label{def:qqbar:1}
\ee
It is called the {\it conjugate-defining equation} of $\bq_l$ by $\q_l$ and the basis letters.

Let there be $n$ quaternionic variables $\q_1, \q_2, \ldots, \q_n$, 
with conjugates $\bq_1, \bq_2, \ldots, \bq_n$ respectively. Denote
\be\ba{lll}
{\cal A} &=& {\cal Q}\cup \overline{\cal Q}\cup {\cal E}, \hbox{ where }\\
{\cal Q} &=& \{\q_1, \q_2, \ldots, \q_n\}, \\
\overline{\cal Q} &=& \{\bq_1, \bq_2, \ldots, \bq_n\}, \\
{\cal E} &=& \{\i,\j,\k\}.
\ea
\ee

\bdf
The (coordinate-free) {\it quaternionic polynomial $\mathbb K$-algebra} over alphabet ${\cal A}$,
denoted by $Q_{\mathbb K}\langle {\cal A} \rangle$, 
is the quotient of the free associative $\mathbb K$-algebra ${\mathbb K}\langle {\cal A}\rangle$
modulo the ideal $\cal I$ generated by the following three groups of elements:
\bi
\item relations from the multiplication table of the basis letters:
\be\ba{llllll}
\i^2+1, & \j^2+1, & \k^2+1, \\
\i\j-\k, & \j\i+\k, & \i\k+\j, & \k\i-\j, &\k\j+\i, & \j\k-\i;
\ea
\label{gen:table}
\ee

\item $n$ relations from the conjugate-defining equations of $\bq_l$
for $l=1,\ldots,n$, see (\ref{def:qqbar:1}): 
\be
2\bq_l+\q_l+\i\q_l\i+\j\q_l\j+\k\q_l\k;
\label{gen:coord}
\ee

\item commutativity between the coordinates of the $\q_i$ and the letters of $\cal A$,
see (\ref{coordinate-free}):
for any letter $\a\in \cal A$, any $l\in \{1, \ldots, n\}$, 
\be
\a[\q_l]-[\q_l]\a, \ \ \,
\a[\i\q_l]-[\i\q_l]\a, \ \ \,
\a[\j\q_l]-[\j\q_l]\a, \ \ \,
\a[\k\q_l]-[\k\q_l]\a.
\label{def:quat}
\ee
\ei
\label{def:quat:all}
\edf

Henceforth we use $\cal I$ to denote specifically the above ideal, called the {\it generating ideal}
of quaternionic polynomial algebra.
The {\it basis-free quaternionic polynomial $\mathbb K$-algebra} over alphabet 
${\cal Q}$, denoted by $Q_{\mathbb K}\langle {\cal Q} \rangle$, is the intersection of 
$Q_{\mathbb K}\langle {\cal A} \rangle$ and 
${\mathbb K}\langle {\cal Q}\cup \overline{\cal Q}\rangle$.

One may argue that in the above definition, the coefficients should also 
be quaternions ${\mathbb K}\langle{\cal E}\rangle$ instead of scalars in $\mathbb K$.
In a quaternionic monomial, there can be more than one quaternionic parameter in symbolic form.
For example, in $\p_1\q_1\p_2\q_2\cdots \p_k\q_k\p_{k+1}$, the $\p_i$ are 
quaternionic parameters and the $\q_j$ are quaternionic variables. 
In normalizing a quaternionic monomial where all quaternionic parameters are given in
coordinate-free form, the quaternionic parameters are treated just as
quaternionic variables. If instead, some quaternionic parameters are given in coordinate form,
say $\p_l=u'_l+x'_l\i+y'_l\j+z'_l\k$, then by multilinearity, for any quaternionic monomials
$\A, \B$, 
\be
\A\p_l\B=u'_l\A\B+x'_l\A\i\B+y'_l\A\j\B+z'_l\A\k\B,
\ee 
so the normal form of $\A\p_l\B$ can be obtained from the normal forms of
$\A\B, \A\i\B, \A\j\B, \A\k\B$. In this way, the above definition of
quaternionic polynomial $\mathbb K$-algebra indeed takes into account quaternionic parameters.

In this paper, we always use the following order in alphabet ${\cal A}$, called the
{\it conjugate-alternating order}:
\begin{equation}
\q_1 \prec \bq_1 \prec \q_2 \prec \bq_2 \prec \ldots \prec \q_n \prec
\bq_n \prec \i \prec \j \prec \k.
\label{orderq}
\end{equation}

\bp
For any quaternionic monomials $\X,\Y$, $\X[\Y]-[\Y]\X\in {\cal I}$.
\label{ideal:commute}
\ep

\begin{proof}
By the coordinatization (\ref{def:qqbar}), $[\Y]$ can be written as
a ${\mathbb K}$-polynomial $f$ whose variables are the coordinate polynomials of the letters
in $\cal A$. By the commutativity of
these coordinate polynomials with the letters in $\cal A$, we get the 
commutativity of $[\Y]$ with monomial $\X$, namely, 
$\X[\Y]-[\Y]\X\in {\cal I}$.
\end{proof}

For example, the first line of (\ref{def:quat}) is 
on the commutativity of $[\q_l]$ with any letter $\a\in \cal A$. In particular,
\be
[\q_l]\q_l-\q_l[\q_l]=\bq_l\q_l-\q_l\bq_l.
\ee
Since both sides have the same leading term $\bq_l\q_l$, the relation
$\bq_l\q_l=\q_l\bq_l$ can be taken as the commutativity of $[\q_l]$ with $\q_l$. Alternatively, 
by $[\q_l\bq_l]=2\q_l\bq_l$, the relation $\bq_l\q_l=\q_l\bq_l$ can also be taken
as a special case of the following shift symmetry within a bracket.

\bc
For any quaternionic monomials $\X, \Y$, $[\X\Y]-[\Y\X]\in {\cal I}$.
\ec

\begin{proof} By $[\X]\Y=\Y[\X]$ and $\bar{\X}[\Y]=[\Y]\bar{\X}$,
\be\ba{lcl}
[\X\Y]-[\Y\X] &=& \X\Y+\bar{\Y}\bar{\X}-\Y\X-\bar{\X}\bar{\Y} \\

&=& [\X]\Y-\bar{\X}[\Y]+\bar{\Y}\bar{\X}-\Y\X \\

&=& (\Y[\X]-\Y\X)+(\bar{\Y}\bar{\X}-[\Y]\bar{\X}) \\

&=& 0.
\ea
\ee
\end{proof} 

In summary, the bracket operator has the following two symmetries: for any 
quaternionic monomials $\X,\Y$,
\bi
\item conjugate reversal symmetry: $[\X\Y]=[\bar{\Y}\bar{\X}]$; 
\item shift/rotation symmetry: $[\X\Y]=[\Y\X]$.
\ei

The following is the main theorem of this paper on the normal forms of quaternionic
polynomials.

\bt [Normal form]
\label{thm:normal}
A quaternionic monomial is in normal form if and only if it is a submonomial of 
a monomial in the following form, see Figure \ref{normal:2}:
\be
\A_1\q_{p1}\q_{f1}\A_2\q_{p2}\q_{f2}\cdots \A_r\q_{pr}\q_{fr}\A_{r+1}
\e_1 \q_{f(r+1)} \e_2 \q_{f(r+2)} \cdots \e_s \q_{f(r+s)} \e_{s+1},
\ee
where $r,s\geq 0$, the $\q_j\in \cal Q$, the $\e_j\in \cal E$, letter $\k$ occurs at most once
in the $\e_j$, and
\bu
\item each $\A_i$ is either empty or composed of a non-decreasing sequence of letters in
${\cal Q}\cup \overline{\cal Q}$;

\item sequence $\q_{p1}\q_{p2}\cdots \q_{pr}\e_1\e_2\cdots \e_{s+1}$ 
is non-decreasing, 
called the {\bf peak sequence}; 

\item sequence $\q_{f1}\q_{f2}\cdots \q_{f(r+s)}$ is non-decreasing, 
called the {\bf floor sequence}, or {\bf lower sequence};

\item sequence $\A_1\q_{p1}\A_2\q_{p2}\cdots \A_{r+1}\e_1\e_2\cdots \e_{s+1}$
is non-decreasing, called the {\bf ceiling sequence}, or {\bf upper sequence};

\item for any $1\leq l\leq r$, $\q_{pl}\succ \q_{fl}$, and 
$\q_{pl}\succ_M \A_{l}$ if the latter is non-empty.
\eu
To obtain the normal form of a quaternionic polynomial, make reduction with respect to
the following set of relations arbitrarily:
\bi
\item multiplication table of the basis elements;

\item conjugate-defining equation (\ref{def:qqbar:1});

\item commutativity of any bracket $[\Y]$ with any monomial $\X$, i.e.,
$\X[\Y]-[\Y]\X$;

\item shift invariance within any bracket: $[\X\Y]-[\Y\X]$ for any
monomials $\X,\Y$.
\ei
\et

\begin{figure}[htbp]
\centering\includegraphics[width=5.5in]{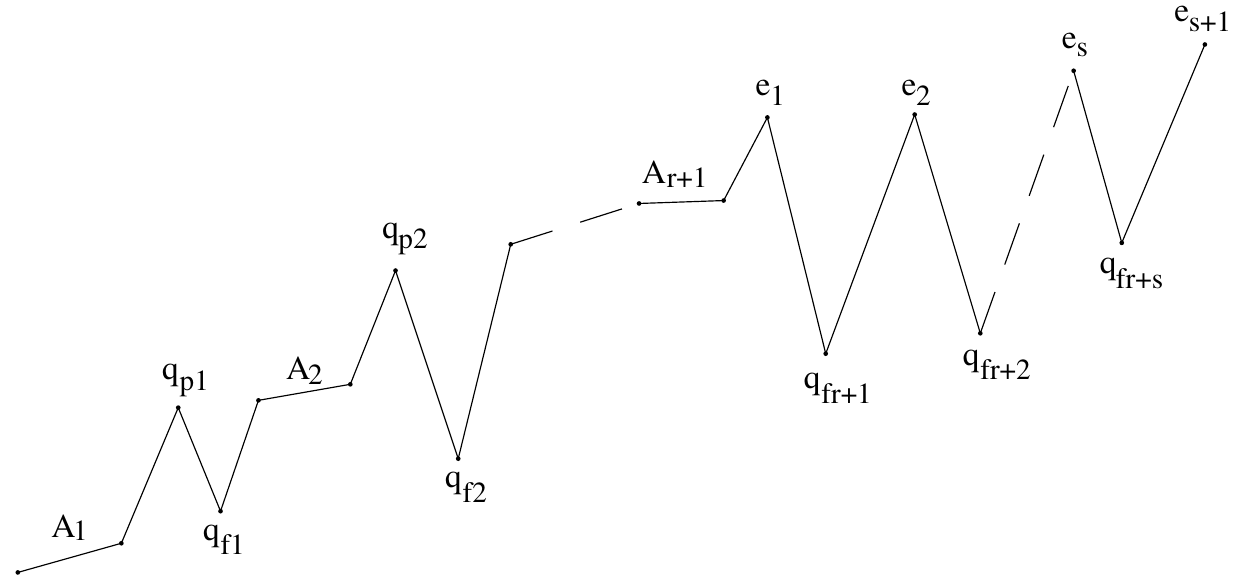}
\caption{Double non-decreasing structure of quaternionic monomials in normal form, with
peak sequence $\q_{p1}\q_{p2}\cdots \q_{pr}\e_1\e_2\cdots \e_{s+1}$, floor 
sequence $\q_{f1}\q_{f2}\cdots \q_{f(r+s)}$, and ceiling
sequence $\A_1\q_{p1}\A_2\q_{p2}\cdots \A_{r+1}\e_1\e_2\cdots \e_{s+1}$.
}
\label{normal:2}
\end{figure}

Theorem \ref{thm:normal} is a direct corollary of the following Main Theorem on
a reduced Gr\"obner basis {\bf BG}
of $\cal I$. The set of relations given in Theorem \ref{thm:normal} for making reduction
strictly contain {\bf BG} as a subset, so it is also a Gr\"obner basis of $\cal I$.

\begin{thm} [Main Theorem]
\label{mainthm}
In the following sets,
$\1\prec \2 \prec \cdots \prec \m$ is an arbitrary increasing sequence of
letters in ${\cal Q}$, where $m\geq 5$, and letters
$\e'\preceq \e$ are arbitrary in $\cal E$:

\begin{enumerate}
\item [{\bf BG2:}]
     \begin{enumerate}
     \item [{\rm (a)}] $\i^2+1$, $\j^2+1$, $\k^2+1$,\\
$\i\j-\k$, $\j\i+\k$, $\j\k-\i$, $\k\j+\i$,
$\k\i-\j$, $\i\k+\j$;

     \item [{\rm (b)}]
$\bar{\1}\1-\1\bar{\1}$,\
$[\2]\bar{\1}-\bar{\1}[\2]$,\ $[\2]\1-\1[\2]$;

      \item [{\rm (c)}]
$\2[\1]-[\1]\2$,\
$\e[\1]-[\1]\e$;
     \end{enumerate}

\item [{\bf BG3:}]
     \begin{enumerate}
      \item [{\rm (a)}]
$\k\1\k+\j\1\j+\i\1\i+2\times \bar{\1}+\1$; {\rm (notice: here is twice the $\bar{\1}$)}

     \item [{\rm (b)}]
$[\3\2]\1-\1[\3\2]$, \
$[\3\1]\bar{\2}-\bar{\2}[\3\1]$,\ \,
$[\3\1]\2-\2[\3\1]$, \
$[\2\1]\bar{\1}-\bar{\1}[\2\1]$, \
$[\2\1]\1-\1[\2\1]$,\\
$[\e\2]\1-\1[\e\2]$,\ \ 
$[\e\1]\2-\2[\e\1]$,\ \
$[\e\1]\bar{\2}-\bar{\2}[\e\1]$, \ \
$[\e\1]\1-\1[\e\1]$,\ \
$[\e\1]\bar{\1}-\bar{\1}[\e\1]$,\\
$[\j\1]\i-\i[\j\1]$,\ \ \ \
$[\k\1]\i-\i[\k\1]$,\ \ \ 
$[\k\1]\j-\j[\k\1]$;

     \item [{\rm (c)}]
$\2[\2\1]-[\2\1]\2$;
     \end{enumerate}

\item [{\bf BG4:}]
$\3[\2\3\1]-[\2\3\1]\3$, \
$\3[\2\4\1]-[\2\4\1]\3$,\
$\3[\2\e\1]-[\2\e\1]\3$,\
$\e'[\2\e\1]-[\2\e\1]\e'$,\
where $(\e',\e)\neq(\k,\k)$;

\item [{\bf BGm:}] $(m\geq 5)$ $\3[\2\A\m\1]-[\2\A\m\1]\3$, \
$\3[\2\A\e\1]-[\2\A\e\1]\3$, where 
\bi
\item $\A=\x_4\x_5\cdots \x_{m-1}$, with each $\x_i\in {\mathcal Q}\cup \overline{\mathcal Q}$,

\item $\3\preceq
\x_4\preceq \x_5\preceq \cdots \preceq \x_{m-1}\prec \m$.
\ei
\end{enumerate}
Then ${\bf BG}:=\cup_{t\geq 2}{\bf BGt}$ is a
reduced Gr\"obner basis of $\cal I$. Sometimes $\bf BG$ is written as
${\bf BG}[{\cal Q}]$ to denote the dependence on letters of $\cal Q$, or
equivalently, written as ${\bf BG}[{\cal A}]$ where 
${\cal A}={\cal Q}\cup \overline{\cal Q}\cup {\cal E}$.
\end{thm}

To prove the Main Theorem, notice that
the pairwise irreducibility of any two elements of {\bf BG} is obvious. 
That {\bf BG} generates ideal $\cal I$ is also easy to prove.

\begin{prop}
\label{lem:generators}
$\langle {\bf BG}\rangle = {\cal I}$.
\end{prop}

\begin{proof}
In {\bf BG}, {\bf BG2}(a) is exactly the relations 
from the multiplication table of basis letters, 
{\bf BG3}(a) is the conjugate-defining relation (\ref{def:qqbar:1}), and 
all other elements of {\bf BG} are on the commutativity between a bracket 
of monomial and a letter. By this observation and Proposition \ref{ideal:commute}, 
$\langle {\bf BG}\rangle\subseteq {\cal I}$.

Conversely, in the three groups of
generators of ${\cal I}$ given in Definition \ref{def:quat:all},
(\ref{gen:table}) and (\ref{gen:coord}) are already in {\bf BG}, and (\ref{def:quat})
is on the commutativity between letter $\a$ and a bracket of length-2 monomial, which is
easy to be reduced to zero by {\bf BG}. 
\end{proof}

So to prove that {\bf BG} is a Gr\"obner basis of $\cal I$, we only need to prove that
any S-polynomial of any two elements of {\bf BG} can be reduced to zero by {\bf BG}. It turns out
that this is a highly complicated procedure, due to the harsh monotonous 
requirement on the symbolically many letters 
in the brackets of {\bf BGm}, where $m$ itself is a symbolic number. To finish the proof
we need to work out new techniques in three directions: 
\bu
\item New reduction techniques for free associative algebras. In the next section,
we develop four new techniques for making reduction in free associative algebras:
subsequence-led reduction, subsequence-dominated reduction, 
bottom-letter controlled reduction, and substitutional reduction.

\item Extend {\bf BG} to a larger Gr\"obner basis candidate. We first extend {\bf BG} to
$\overline{\bf BG}$, where every letter in {\bf BG} can be replaced by its conjugate. Then
we further extend $\overline{\bf BG}$ by adding the following elements:
$\a[\T]-[\T]\a$, for all letter $\a\in {\cal A}$ and monomial $\T$.
In each extension, we verify that the newly added elements can 
be reduced to zero by {\bf BG}. These contents are in Section \ref{sect:reduct} and
Section \ref{sect:ext} respectively.

\item Reduce the number of S-polynomials by which the Gr\"obner basis can be certified. 
We propose a concept called {\it clear S-polynomial} in Section \ref{sect:gb}
for free associative algebras, and prove
that only such S-polynomials are needed in certifying a Gr\"obner basis whose elements
all have degree $>1$.
\eu

\section{New reduction techniques and preliminary extension of BG}
\setcounter{equation}{0}
\label{sect:reduct}

Recall the definition of reduction in free associative algebra.
Let $f$ be the input polynomial to be reduced by a set of polynomials $G$, 
and let $\T$ be the leading monic monomial of $f$. 
$f$ is said to be reduced to polynomial $h$ with respect to $G$, if
$h\preceq \T$, and there exist a finite sequence of triplets 
$(\A_i, g_i, \B_i)$, where the $g_i\in G$, 
the $\A_i, \B_i$ are monomials, such that each $\A_i g_i \B_i\preceq \T$, 
and $f=h+\sum_{i=1}^k \A_ig_i\B_i$. 

Suppose that in a reduction procedure, the order of applying the $g_i$
is sequentially from $i=1$ to $k$, then it is allowed that at some step
$j>1$, $\A_jg_j\B_j\succ$ the leading term of $f'=f-\sum_{i<j} \A_ig_i\B_i$,
but still $\A_jg_j\B_j\prec \T$. When such phenomenon occurs, locally 
the order of the leading term increases, but is still lower than 
that of the input polynomial. We call such a leading term a 
{\it local jumping head}. In the following, we 
present three typical techniques with local jumping head control of the
reduction procedure.

\vskip .2cm
{\bf Technique 1.} Subsequence-led reduction.

Let $\L$ be a monic monomial whose length $|\L|=l<|\T|$, such that
$\L\prec $ first $l$ letters of $\T$. Then 
for any element $g\in \langle G\rangle$
whose leading term ${\rm T}(g)$ has length $\leq |\T|-|\L|$, 
we have $\L g \prec \T$. Such a term is
always allowed in the polynomials $\A_ig_i\B_i$. When such a term does occur
in some step $j$, we say at this step, the reduction is {\it led by subsequence}
$\L$, denoted by $\L \langle G\rangle$, as 
$\A_jg_j\B_j=\L g \in \L\langle G\rangle$.
We can also replace $\L$ with a polynomial with $\L$ 
as the leading monomial.

For example, in ${\bf BGm}$-relation $\3[\2\A\m\1]-[\2\A\m\1]\3$, where $\A\succeq_M \3$
is a monomial, by replacing letter $\1$ by $\bar{\1}$, we get
$f=\3[\2\A\m\bar{\1}]-[\2\A\m\bar{\1}]\3$. In Proposition \ref{prop:overline:BG} below,
it is proved that $f$ is reduced to zero by ${\bf BG}$, where after some reductions,
$f$ is reduced to $f'=[\1](\3[\2\A\m]-[\2\A\m]\3)$. Since the input leading term
is $\T=\3\2\A\m\bar{\1}$, any monomial led by $\1$ or $\bar{\1}$ is $\prec \T$, as long
as the monomial is not longer that $\T$. So we can use the relation
$g=\3[\2\A\m]-[\2\A\m]\3\in {\cal I}=\langle {\bf BG}\rangle$ to make 
$[\1]$-led reduction to $f'$ to get zero.

\vskip .2cm
{\bf Technique 2.} Subsequence-dominated reduction.

Let $g=\sum_i \L_i \R_i\in G$, such that for some monomials $\A, \B$, polynomial
$\A g\B\prec $ the leading monomial $\T$ of $f$. Here $\A$ is allowed to be of degree 0.
Assume
$\T=\A\C$ for some monomial $\C$. Then the reduction $f-\A g\B$ is not led by
submonomial $\A$. Let $\L_1 \R_1$ be the leading term of $g$. 
When $\A\L_1\prec $ first $|\A\L_1|$ letters of $\T$, we call  
$f-\A g\B$ an {\it $\A\L_1$-dominated reduction} of $f$, denoted by 
``$(\A\L_1):$", and if $g$ belongs to a class with label $X$, we label the reduction
as $(\A\L_1):X$. If the leading term of $g$ is among the terms $\sum_{i<j} \L_i \R_i$,
and $\A\L_i\prec $ first $|\A\L_i|$ letters of $\T$ for each $i<j$, we also denote the
reduction as ``$(\A\sum_{i<j}\L_i):$", and if there is a label $\X$ for $g$, 
we label the reduction as $(\A\sum_{i<j}\L_i):X$.

For example, in ${\bf BGm}$-relation $\3[\2\A\m\1]-[\2\A\m\1]\3$, where $\A\succeq_M \3$
is a monomial, 
by replacing letter $\2$ by $\bar{\2}$, we get
$f=\3[\bar{\2}\A\m\1]-[\bar{\2}\A\m\1]\3$. In Proposition \ref{prop:overline:BG} below,
it is proved that $f$ is reduced to zero by ${\bf BG}$, where after some reductions,
$f$ is reduced to a polynomial with leading term in $\3\bar{\1}[\2]\bar{\m}\bar{\A}$. 
By relation $\bar{\1}[\2]-[\2]\bar{\1}$, the leading term increases its order from
$\3\bar{\1}\bar{\2}\bar{\m}\bar{\A}$ to $\3\bar{\2}\bar{\1}\bar{\m}\bar{\A}$. 
This is a reduction dominated by subsequence $\3\bar{\2}\bar{\1}$, because 
$\3\bar{\2}\bar{\1}\prec $ first three letters of the input leading term 
$\T=\3\2\A\m\1$. However, this is not a $\3$-led reduction, nor a $\3\bar{\2}$-dominated one,
because $\3$ (or $\3\bar{\2}$) is exactly the first letter (or first two letters) of $\T$.

We compare the difference between subsequence-led reduction and subsequence-dominated reduction.
Let $\L$ be a monomial satisfying $|\L|<|\T|$, where $\T$ is the
leading monomial of $f$. 
\bi
\item In an $\L$-led reduction of $f$ with respect to $G$, 
$\L\prec $ first $|\L|$ letters of $\T$, 
an element $g\in \langle G\rangle$ is chosen, then $f-\L g$ is the reduction result.

\item In an $\L$-dominated reduction of $f$ with respect to $G$, 
$\L=\A\L_1\prec $ first $|\L|$ letters of $\T$, where submonomial 
$\A=$ first $|\A|$ letters of $\T$, an element $g\in G$ is chosen such that
$\L=$ first $|\L|$ letters of the leading monomial of $\A g$, then
$f-\A g \B$ for some monomial $\B$ is the reduction result.
\ei

\vskip .2cm
{\bf Technique 3.} Substitutional reduction.
 
In reducing by $G$ a polynomial $f$ involving {\it monomials of symbolic length}, e.g., monomial $\A$ 
in $\bf BGm$, often new letters are needed to represent the monomial, together with
a new order among the new letters. After this rewriting, $G$ and $f$ are changed into 
$G'$ and $f'$ in new variables. The order among the new letters often disagrees with the
order of monomials in old letters, e.g., in the extreme case the symbolic length of 
a monomial is allowed to take value 0, so the order among the new letters always disagree with
the old order. The reduction procedure of $f'$ with respect to $G'$, 
when translated back to old letters, need to be checked stepwise whether or not it remains
a correct reduction. This is the {\it first instance of substitutional reduction}: introduce
new letters to substitute monomials, define an order among the new letters, make reduction in 
new letters, and finally check the whole reduction procedure when translated back to old letters. 

The {\it second instance of substitutional reduction} involves introducing new letters to
replace old letters, usually in the case where an old letter has {\it varying order}, and a single
reduction procedure is needed for all possible orders of the letter. In such a case, introducing
new letters and then defining an order among the new letters, is equivalent to first considering
only a specific case of the varying order and making reduction only for this specific case, then
checking the reduction procedure to make sure that it is valid for all other possibilities of the
varying order.

Section \ref{sect:proof} is full of instances of substitutional reduction. 
For example in the proof of Proposition \ref{prop:h1}, it is proved 
that in a single reduction procedure, the following polynomial:
\be\ba{lll}
h_1 &=& (\5[\3\A_6\m_1\1]-[\3\A_6\m_1\1]\5)\m_2
-\5\3\A_6([\m_1\1]\m_2-\m_2[\m_1\1]) \\

&=& 
\5\bar{\1}\overline{\m_1}\overline{\A_6}\bar{\3}\m_2
-\5\3\A_6\bar{\1}\overline{\m_1}\m_2
-[\3\A_6\m_1\1]\5\m_2
+\m_2\5\3\A_6[\m_1\1],
\ea\ee
where 
(i) $\1\prec \3\prec \4\in \cal Q$, 
(ii) $\m_1\in {\cal Q}\cup {\cal E}$, 
(iii) $\m_2\in {\cal A}$, and $\m_2\prec \m_1$,
(iv) $\A_6\succeq_M \5$ is a non-decreasing monomial of length $>0$, and $\A_6\prec_M \m_1$,
can be reduced to zero by {\bf BG}, disregard of the following possibilities
of the varying order of $\m_2$:

(1) $\m_2\prec \1$, e.g. $\m_2=\0$ or $\bar{\0}$, where $\0\in {\cal Q}$;\\
\indent (2) $\m_2\in \{\1, \bar{\1}\}$;\\
\indent (3) $\bar{\1}\prec \m_2\prec \3$, e.g. $\m_2=\2$ or $\bar{\2}$, where $\2\in {\cal Q}$;\\
\indent (4) $\m_2\in \{\3, \bar{\3}\}$;\\
\indent (5) $\bar{\3}\prec \m_2\prec \5$, e.g. $\m_2=\4$ or $\bar{\4}$, where $\4\in {\cal Q}$;\\
\indent (6) $\m_2\in \{\5, \bar{\5}\}$;\\
\indent (7) $\bar{\5}\prec \m_2\prec \m_1$, e.g. $\m_2=\6$, where 
$\6\in {\cal A}$.

The proving method is to first introduce new variables $\1$ to $\6$:
\be
(\1, \3, \5, \A_6, \m_2, \m_1) \to (\1, \2, \3, \4, \5, \6), \hbox{ with }
\1\prec \2\prec \3\prec \4\prec \5\prec \6, 
\ee
so that $h_1$ is rewritten as 
$h_1'$, then reduce $h_1'$ to zero by the corresponding
${\bf BG}[\1, \ldots, \6]$ in the new letters, and finally check the reduction
procedure of $h_1'$ after replacing the new letters with old ones to see if
it remains a correct reduction for all the possible orders of $\m_2$, and for the role of
$\4=\A_6$ as an arbitrary non-decreasing monomial satisfying 
$\m_1\succ_M \A_6\succeq_M \5$ in old letters. 

\vskip .2cm
{\bf Technique 4.} Bottom-letter controlled reduction.

This is an integration of the three previous techniques. Let $f$ be a polynomial
whose leading monomial is led by letter $\t$, and let $\cal B$ be a set of letters
with property that every letter $\b\in \cal B$ satisfies $\b\prec_M \t\bar{\t}$. 
If after some reductions with respect to $G$, $f$ is reduced to a polynomial $f'$
whose terms have the property that the monic monomial of the term is led by a letter
in $\cal B$. Then disregard of the original order among the letters, 
new letters can be introduced for the monomials and old letters in $f'$, together with an
order among the new letters that may disagree with the order in old letters, 
and all kinds of $\b$-led reductions, $\b$-dominated reductions can be done to
$f'$ in new letters, for all letters $\b\in {\cal B}$,
as long as in the new order of letters, 
{\it those letters of $\cal B$ are the lowest}, 
although they may not be the lowest in old order. The substitutional
reductions need to be checked for correctness.

In quaternionic polynomial algebra, if either no letter of $\cal B$ is in $\cal E$,
or the multiplication table of basis letters is not used in making reduction to $f'$,
then after all kinds of reductions, $f'$ remains to have the letters of $\cal B$ as 
the leading letter of every term of it. In this case, the bottom-letter controlled reduction 
is composed of ${\cal B}$-led reductions and ${\cal B}$-dominated reductions in new letters.
When translated back to old letters, the ${\cal B}$-led reductions remain to be
${\cal B}$-led reductions, so only the ${\cal B}$-dominated reductions in new letters need
to be checked for correctness in old letters. The following proposition claims that if
$\cal B$ contains only two letters and their conjugates, then there is no need to check
the correctness of any reduction with respect to the {\bf BG} relations in new letters.

\bp
\label{prop:bottom-reduce}
Let $f$ be a polynomial
whose leading monomial is led by letter $\c$, 
and after some reductions with respect to ${\bf BG}[{\cal A}]$,
each term in the reduction result is led by a letter of 
${\cal B}=\{\a,\bar{\a}, \b, \bar{\b}\}\subset {\cal A}$, with property
$\c\succ_M \a\bar{\a}\b\bar{\b}$. If introducing new quaternionic letters
$\1\prec \2\prec \ldots\prec \r$ for the letters and monomials in $f$ by setting 
$(\a,\b)\to (\1,\2)$, and making reductions to $f$ by ${\bf BG}[\1, \2, \ldots, \r]$
in new letters, then all the reductions when translated back to old letters, are
reductions of $f$ with respect to ${\bf BG}[{\cal A}]$.
\ep

\begin{proof}
In new letters, the
new alphabet is ${\cal A}'={\cal Q}'\cup \overline{{\cal Q}'}\cup {\cal E}$, where
${\cal Q}'=\{\1, \2, \ldots, \r\}$, and $f'$ has all its terms led by letters of 
${\cal B}'=\{\1,\bar{\1},\2, \bar{\2}\}\subseteq {\cal Q}'\cup \overline{{\cal Q}'}$, 
which are the four lowest ordered ones in ${\cal A}'$. Every reduction with respect to
${\bf BG}[{\cal A}']$ is either a ${\cal B}'$-led reduction or a ${\cal B}'$-dominated one. A 
${\cal B}'$-led reduction is naturally a ${\cal B}$-led reduction in old letters.

As $\c\succ_M \a\bar{\a}\b\bar{\b}$ in old letters, and furthermore 
if $\c$ becomes $\s$ in new letter,
then $\s\succ_M \1\bar{\1}\2\bar{\2}$,
any ${\cal B}'$-dominated reduction $h$ of $f$ in new letters
by applying a relation $g\in {\bf BG}[{\cal A}']$ has the property
that $f'=h+g \R$, where $\R$ is a monomial. So $g$ has its leading monomial led by
a letter of ${\cal B}'$. In ${\bf BG}[{\cal A}']$, all elements $g$ with this 
property are the following: 
\bi
\item ${\bf BG2}[{\cal A}']$: \
$\bar{\1}\1-\1\bar{\1}$,\
$\bar{\2}\2-\2\bar{\2}$,\
$[\2]\bar{\1}-\bar{\1}[\2]$,\ $[\2]\1-\1[\2]$,\
$\2[\1]-[\1]\2$;

\item ${\bf BG3}[{\cal A}']$:\
$[\2\1]\bar{\1}-\bar{\1}[\2\1]$, \
$[\2\1]\1-\1[\2\1]$,\
$\2[\2\1]-[\2\1]\2$.
\ei
These elements when written in old letters, are all elements of ${\bf BG2}[{\cal A}]$
and ${\bf BG3}[{\cal A}]$.
So any ${\cal B}'$-dominated reduction by applying ${\bf BG}[{\cal A}']$
is a ${\cal B}$-dominated reduction by applying ${\bf BG}[{\cal A}]$.
\end{proof}

For example, in the proof of Proposition \ref{prop:overline:BG} below, for input polynomial
$f=\3[\bar{\2}\A\m\1]-[\bar{\2}\A\m\1]\3$, after some reductions by {\bf BG}, 
$f$ is reduced to 
\be
f'=[\2](\3\A\m\1-\A\m\1\3)-\bar{\1}[\2]\bar{\m}\bar{\A}\3
-\1\bar{\3}[\2]\bar{\m}\bar{\A}+[\2][\3\bar{\1}]\bar{\m}\bar{\A},
\label{bgbar:ii:2}
\ee
where each term is led by a letter of
${\cal B}=\{\1, \bar{\1}, \2, \bar{\2}\}$. Set new letters
\be
(\1,\2,\3,\A\m)\to (\1,\2,\3,\4), \hbox{ where } \1\prec \2 \prec \3 \prec \4, 
\ee
and make substitutional reduction to $f'$ with respect to ${\bf BG}[\1,\ldots,\4]$
in the new letters:
\be
\ba{lcl}
f' &=& [\2](\3\4\1-\4\1\3)-\underbrace{\bar{\1}[\2]}\bar{\4}\3
-\underbrace{\1\bar{\3}[\2]}\bar{\4}+[\2][\3\bar{\1}]\bar{\4} \\

& \overset{([\2]):,\, \1{\cal I}}{=}&
[\2](\3\4\1-\underbrace{\4\1\3}-\bar{\1}\bar{\4}\3-\1\bar{\3}\bar{\4}+[\3\bar{\1}]\bar{\4}) \\

& \overset{[\2]{\cal I}}{=}& 0,

\ea
\label{tech4:ex1}
\ee
where in the first step, there are $[\2]$-dominated reduction 
$\bar{\1}[\2]=[\2]\bar{\1}$, and a combination of 
$\1$-led reduction $\1\bar{\3}[\2]=\1[\2]\bar{\3}$ with the $[\2]$-dominated reduction,
leading to reduction $\1\bar{\3}[\2]=[\2]\1\bar{\3}$. 
In the last step, the $[\2]$-led reduction is by
$-\4\1\3=\bar{\1}\bar{\4}\3-\3[\4\1]$, which is not 
a ${\bf BG}\3$-relation in old letters, because $\4$ is a monomial instead of a letter.
Furthermore, all substitutional
reductions do not need to be checked for correctness in old letters.

In the following, we make a preliminary extension of {\bf BG} by allowing any letter
in any element of {\bf BG} to be replaced by its conjugate. 

\bp
\label{prop:overline:BG}
Let $\overline{\bf BG}$ denote the union of {\bf BG} with the
polynomials obtained from {\bf BG}
by replacing any number of letters in
$\bf BG$ with their conjugates. Then 
any element of $\overline{\bf BG}$ can be reduced to zero by $\bf BG$.
\ep

\begin{proof} It is easy to verify that in any element of
{\bf BG2} to {\bf BG4}, if one or more than one letter is replaced by its
conjugate, the resulting polynomial can be reduced to zero by 
{\bf BG2} to {\bf BG4}. The corresponding extensions are denoted by 
$\overline{\bf BG}\2$ to $\overline{\bf BG}\4$ respectively.

As to {\bf BGm} where $m>4$, 
by similarity we only show the proof that
$\3[\2\A\m\1]-[\2\A\m\1]\3$ can be reduced to zero by {\bf BG},
where $\A=\x_4\x_5\cdots \x_{m-1}$.
If some $\x_i$ is replaced by $\overline{\x_i}$, the result is still in
{\bf BGm}, so we only need to consider the conjugations of $\1,\2,\3,\m$
respectively, then consider various combinations of their conjugations.

(i) $\1$ is replaced by $\bar{\1}$.
\be\ba{cl}
& \3[\2\A\m\bar{\1}]-[\2\A\m\bar{\1}]\3 \\[2mm]

=& \underbrace{\3\2\A\m[\1]}-\underbrace{\3\2\A\m\1}+\3\1\bar{\m}\bar{\A}\bar{\2}
-[\2\A\m\bar{\1}]\3 \\

\overset{\overline{\bf BG}\2, {\bf BGm}}{=}&
[\1]\3\2\A\m+\underbrace{\3[\1]}\bar{\m}\bar{\A}\bar{\2}
-\underbrace{[\2\A\m[\1]]}\3 \\

\overset{\overline{\bf BG}\2}{=}&
[\1](\3[\2\A\m]-[\2\A\m]\3) \\

\stackrel{\scriptscriptstyle [\1]{\cal I}}{=}& 0.
\ea\ee

(ii) $\2$ is replaced by $\bar{\2}$.
\be\ba{cl}
& \3[\bar{\2}\A\m\1]-[\bar{\2}\A\m\1]\3 \\[2mm]

=& \underbrace{\3[\2]}\A\m\1-\underbrace{\3\2\A\m\1}
+\3\bar{\1}\bar{\m}\bar{\A}\2-[\bar{\2}\A\m\1]\3 \\

\overset{\overline{\bf BG}\2, {\bf BGm}}{=}&
[\2]\3\A\m\1+\3\bar{\1}\underbrace{\bar{\m}\bar{\A}[\2]}
-[\2]\A\m\1\3-\bar{\1}\underbrace{\bar{\m}\bar{\A}[\2]}\3\\

\overset{\overline{\bf BG}\2}{=}&
[\2](\3\A\m\1-\A\m\1\3)-\underbrace{\bar{\1}[\2]}\bar{\m}\bar{\A}\3
+\underbrace{\3\bar{\1}[\2]}\bar{\m}\bar{\A}\\

\overset{([\2]):\overline{\bf BG}\2,\,
(\3[\2]\bar{\1}):\overline{\bf BG}\2}{=}&
[\2](\3\A\m\1-\A\m\1\3-\bar{\1}\bar{\m}\bar{\A}\3
+\3\bar{\1}\bar{\m}\bar{\A})\\[2mm]

=&[\2](\3[\A\m\1]-[\A\m\1]\3)\\

\overset{[\2]{\cal I}}{=}& 0.
\ea\ee
In the fourth line, applying $g=\bar{\1}[\2]-[\2]\bar{\1}$ of $\overline{\bf BG}\2$
increases the order of term $\bar{\1}[\2]$, but still the leading term
of $g$ is $\prec$ the input leading term $\3\bar{\2}\A\m\1$, so this is
a $[\2]$-dominated reduction. Similarly, $\3\bar{\1}[\2]=[\2]\3\bar{\1}$
is a combination of applying $\bar{\1}[\2]-[\2]\bar{\1}$ and
$\3[\2]\bar{\1}-[\2]\3\bar{\1}$, a procedure in which the leading term is led by
subsequence $\3\bar{\2}\bar{\1}\prec $ first three letters of 
input leading term $\3\bar{\2}\A\m\1$, so this is a 
$\3[\2]\bar{\1}$-dominated reduction.

Alternatively, in the fifth line we can also apply $\overline{\bf BG}\3$ to
$\3\bar{\1}[\2]$ to get 
$-\1\bar{\3}[\2]+[\2][\3\bar{\1}]$, so that the input polynomial is reduced to the polynomial
$f'$ in (\ref{bgbar:ii:2}), then make bottom-letter controlled reduction (\ref{tech4:ex1})
with bottom letters $\{\1, \bar{\1}, \2, \bar{\2}\}$.

(iii) $\3$ is replaced by $\bar{\3}$.
\be\ba{cl}
& \bar{\3}[\2\A\m\1]-[\2\A\m\1]\bar{\3} \\[2mm]

=& \underbrace{[\3](\2\A\m\1+\bar{\1}\bar{\m}\bar{\A}\bar{\2})}
-\underbrace{\3[\2\A\m\1]}-[\2\A\m\1]\bar{\3}\\

\overset{\overline{\bf BG}\2, {\bf BGm}}{=}&
\2\underbrace{[\3]\A\m\1}+\bar{\1}\underbrace{[\3]\bar{\m}\bar{\A}\bar{\2}}
-[\2\A\m\1][\3]\\

\overset{\2{\cal I}, \bar{\1}{\cal I}}{=}&
\2\A\m\1[\3]+\bar{\1}\bar{\m}\bar{\A}\bar{\2}[\3]
-[\2\A\m\1][\3] \\

=& 0.
\ea\ee

(iv) $\m$ replaced by $\bar{\m}$.
\be\ba{cl}
& \3[\2\A\bar{\m}\1]-[\2\A\bar{\m}\1]\3 \\[2mm]

=& \3\2\A\underbrace{[\m]\1}+\3\bar{\1}\underbrace{[\m]\bar{\A}\bar{\2}}
-\underbrace{\3[\2\A\m\1]}
-\2\A\underbrace{[\m]\1\3}-\bar{\1}\underbrace{[\m]\bar{\A}\bar{\2}\3}
+[\2\A\m\1]\3 \\

\overset{\overline{\bf BG}\2,
\3\bar{\1}{\cal I}, {\bf BGm}, \2{\cal I}, \bar{\1}{\cal I}}{=}&
(\3[\2\A\1]-[\2\A\1]\3)[\m]\\

\overset{{\bf \overline{BG}(m-1)}}{=}& 0.
\ea\ee
In the last step, if the last letter $\x_{m-1}$ of
$\A$ belongs to $\overline{\cal Q}$, then induction on
$m\geq 5$ is needed so that $\bf \overline{BG}(m-1)$ can be
called for, else $\bf BG(m-1)$ is used directly.

In the next to the last step, 
\bi
\item $\3\2\A[\m]\1=\3\2\A\1[\m]$ is by $\overline{\bf BG}\2$-relation
$[\m]\1-\1[\m]$;

\item $\3\bar{\1}[\m]\bar{\A}\bar{\2}=\3\bar{\1}\bar{\A}\bar{\2}[\m]$ is 
a $\3\bar{\1}$-led reduction;

\item $\2\A[\m]\1\3=\2\A\1\3[\m]$ is a $\2$-led reduction;

\item $\bar{\1}[\m]\bar{\A}\bar{\2}\3
=\bar{\1}\bar{\A}\bar{\2}\3[\m]$ is a $\bar{\1}$-led reduction.
\ei

(v) More than one of $\1,\2,\3,\m$ is replaced by its conjugate.

For $i=1,2,3,4$, denote by $\overline{\bf BGm}_i$ the set of polynomials obtained from
${\bf BGm}$ by replacing $i$ letters in $\{\1,\2,\3,\m\}$ to their conjugates
respectively. We have proved that all elements of $\overline{\bf BGm}_1$ can be reduced to
zero by $\bf BG$. We need to prove that for $i=2,3,4$,
$\overline{\bf BGm}_i$ can be reduced to zero by ${\bf BG}$ and
$\cup_{j<i} \overline{\bf BGm}_j$.
The reduction procedure is exactly 
the same as in (i) to (iv), the only difference
is the rule ${\bf BGm}$ used, which should be updated to 
$\cup_{j<i} \overline{\bf BGm}_j$.
\end{proof}

\section{Further extension of BG}
\setcounter{equation}{0}
\label{sect:ext}

The following is a powerful commutativity.

\bp
For any $\1\in \cal Q$, any monomial $\S$, $\S[\1]-[\1]\S$ can be reduced to zero
by $\bf BG$. The conclusion also holds for $\S\1\bar{\1}-\1\bar{\1}\S$.
\label{prop:11}
\ep

\begin{proof} The proofs of the two conclusions are similar, so we only present
the first one.
When $|\S|\leq 2$, the conclusion 
can be easily established by $\overline{\bf BG}\2, \overline{\bf BG}\3$. 
Assume that the conclusion holds for
all monomials of length $<|\S|$, where $|\S|\geq 3$. 
We prove the conclusion for $\S$.

Let $\S=\x_1\x_2\x_3\Z$ where the $\x_i$ are letters, and $\Z$ is a monomial of length
$\geq 0$. When any of $\x_2, \x_3\succ_M\1\bar{\1}$, by
induction hypothesis, the conclusion can be easily verified.
So we only consider the case where both $\x_2, \x_3\preceq_M\1\bar{\1}$.

When $\x_1\succ_M \1\bar{\1}$, by $\overline{\bf BG}\3$,
$
\x_1\x_2\x_3=
-\overline{\x_2}\,\overline{\x_1}\x_3+\x_3[\x_1\x_2].
$
The rest is by induction hypothesis and $\overline{\bf BG}$.
When $\x_1\in \{\1,\bar{\1}\}$, the proof is also trivial by
induction hypothesis and reduction $\bar{\1}\1=\1\bar{\1}$. When 
$\x_1\prec_M \1\bar{\1}$, by induction hypothesis and reduction $[\1]\x_1=\x_1[\1]$, 
the conclusion is true.
\end{proof}

The following two propositions further extend the above commutativity.

\bp
\label{special:f8}
For any $r,s>0$, any $\2\in \cal Q$, 
and any $\x\in \cal A$,
$f'=\x\2^r\bar{\2}^r-\2^r\bar{\2}^r\x$ and 
$f''=\x[\2^s]-[\2^s]\x$ can both be reduced to zero by {\bf BG}.
\ep

\begin{proof} 
When $\x\in \{\2, \bar{\2}\}$, the conclusion is trivial. So we assume 
$\x\notin \{\2, \bar{\2}\}$.
When $r=1$, $f'$ can be easily reduced to zero by $\overline{\bf BG}\3$.
The induction hypothesis ({\sf IH}) is that
the conclusion holds for $r$ taking values $<r$ in $f'$. For value $r$, 
when $\x\prec_M \2$,
\be\ba{lcl}
\2^r\bar{\2}^{r-2}\underbrace{\bar{\2}\bar{\2}\x} &\overset{\overline{\bf BG}\3}{=}& 
\2^r\bar{\2}^{r-2}(-\underbrace{\bar{\2}\bar{\x}}\2
+\bar{\2}\x\bar{\2}+\bar{\x}\2\bar{\2}) \\

&\overset{\overline{\bf BG}\2}{=}& 
\2(\2^{r-1}\bar{\2}^{r-1}\x)[\2]
-\2^2(\2^{r-2}\bar{\2}^{r-2}\x)\2\bar{\2}
 \\

&\overset{{\sf IH},\overline{\bf BG}\2}{=}& 
\2\x\2^{r-1}\bar{\2}^{r-1}[\2]
-\underbrace{\2\2\x}\2^{r-1}\bar{\2}^{r-1}
 \\

&\overset{\overline{\bf BG}\3, \overline{\bf BG}\2}{=}& 
\underbrace{\2[\x]}\2^{r-1}\bar{\2}^{r}-\bar{\x}\2^{r}\bar{\2}^{r} \\

&\overset{\overline{\bf BG}\2}{=}& \x\2^{r}\bar{\2}^{r}.
\ea\ee
When $\x\succ_M \bar{\2}$, the deduction of $\x\2^r\bar{\2}^r$ is similar.

As to $f''$, when $s=1$, $f''$ can be easily reduced to zero 
by $\overline{\bf BG}\2$. 
Let the induction hypothesis be that
the conclusion holds for $s$ taking values $<s$ in $f''$. For value $s$, 
by 
\be
\2^s+\bar{\2}^s\ \overset{\overline{\bf BG}\2}{=}\ (\2+\bar{\2})(\2^{s-1}+\bar{\2}^{s-1})
-\2\bar{\2}(\2^{s-2}+\bar{\2}^{s-2})
\ee
and the induction hypothesis, 
$f''$ can be easily reduced to zero.
\end{proof}

\bp
\label{special:f9}
For any $0\leq r<s$, any letter $\2\in \cal Q$, 
and any letter $\x\in \cal A$,
$f=\x[\2^r\bar{\2}^s]-[\2^r\bar{\2}^s]\x$ can be reduced to zero by {\bf BG}.
\ep

\begin{proof} 
When $\x\in \{\2,\bar{\2}\}$, the conclusion is trivial.
In the following, we assume $\x\notin\{\2,\bar{\2}\}$.

When $\x\prec \2$, the leading term of $f$ is $-\2^r\bar{\2}^s\x$. 
The case $r=0$ is already covered by Proposition \ref{special:f8}. 
Fix $r\geq 1$. The induction hypothesis ({\sf IH}) is that 
for every $0\leq r'<r$, every $s'>0$, 
$\x[\2^{r'}\bar{\2}^{s'}]-[\2^{r'}\bar{\2}^{s'}]\x$ is 
reduced to zero by {\bf BG}.
We prove the conclusion for $r'=r$ and all $s'>r$.  

When $s'=r+1$, 
\be\ba{lcl}
{[} \2^r\bar{\2}^{r+1}]\x &=& 
\2^r\bar{\2}^{r-1}\underbrace{\bar{\2}\bar{\2}\x}+\2\underbrace{\2^{r}\bar{\2}^r\x} \\

&\overset{\overline{\bf BG}\3, \, {\rm Prop.}\, \ref{special:f8}}{=}& 
\underbrace{\2^r\bar{\2}^{r}\x}\bar{\2}+
\2\underbrace{\2^{r-1}\bar{\2}^{r-1}\bar{\x}}\2\bar{\2}-
\underbrace{\2^r\bar{\2}^{r}\bar{\x}}\2
+\2\x\2^{r}\bar{\2}^r
 \\
 
&\overset{{\rm Prop.}\, \ref{special:f8},\,\overline{\bf BG}\2}{=}& 
\x\2^r\bar{\2}^{r+1}
+\underbrace{\2[\x]}\2^{r}\bar{\2}^{r}
-\bar{\x}\2^{r+1}\bar{\2}^{r}
 \\

&\overset{\overline{\bf BG}\2}{=}& 
\x[\2^r\bar{\2}^{r+1}].
\ea\ee

For general $s\geq r+2$, the induction hypothesis is that the conclusion holds for all values
of $s'$ from $r+1$ up to $s-1$. For $s'=s$, the leading term of ${[} \2^r\bar{\2}^{s}]\x$ 
is $\2^r\bar{\2}^{s}\x$, so
\be\ba{lcl}
{[} \2^r\bar{\2}^{s}]\x &=& 
\2^r\bar{\2}^{s-2}\underbrace{\bar{\2}\bar{\2}\x}+\2^{s-r}\underbrace{\2^{r}\bar{\2}^r\x} \\

&\overset{\overline{\bf BG}\3, \,{\rm Prop.}\, \ref{special:f8}}{=}& 
\2^r\bar{\2}^{s-1}\x\bar{\2}
+\underbrace{\2^r\bar{\2}^{s-2}\bar{\x}}\2\bar{\2}
-\underbrace{\2^r\bar{\2}^{s-1}\bar{\x}}\2
+\2^{s-r-2}\underbrace{\2\2\x}\2^{r}\bar{\2}^r
 \\

&\overset{\overline{\bf BG}\2, \overline{\bf BG}\3}{=}& 
\underbrace{\2^r\bar{\2}^{s-1}\x}[\2]
-\underbrace{\2^{r}\bar{\2}^{s-2}\x}\2\bar{\2}
+\2^{s-r-2}\underbrace{(\2\x\2+\bar{\x}\bar{\2}\2-\2\bar{\x}\bar{\2})\2^{r}\bar{\2}^r}
 \\

&\overset{{\sf IH},\,\overline{\bf BG}\2, 
(\2^{s-2}\bar{\2}^{r}\2):{\rm Prop.}\, \ref{special:f8}}{=}& 
\x[\2^r\bar{\2}^{s-1}][\2]-\2^{s-1}\bar{\2}^r\x[\2]
-\x[\2^{r+1}\bar{\2}^{s-1}]+\2^{s-2}\bar{\2}^r\x\2\bar{\2}\\ 

&& \hfill
+\2^{s-1}\bar{\2}^r\x\2
+\2^{s-2}\bar{\2}^r\bar{\x}\2\bar{\2}
-\2^{s-1}\bar{\2}^r\bar{\x}\bar{\2} \\[2mm]

&=&
\x[\2^r\bar{\2}^{s-1}][\2]-\underbrace{\2^{s-1}\bar{\2}^r[\x]}\bar{\2}
-\x[\2^{r+1}\bar{\2}^{s-1}]+\underbrace{\2^{s-2}\bar{\2}^r[\x]}\2\bar{\2}\\ 

&\overset{\overline{\bf BG}\2}{=}&
\x[\2^r\bar{\2}^s].
\ea\ee

When $\x\succ \bar{\2}$, the deduction of $\x[\2^r\bar{\2}^s]$ is similar. 
\end{proof}

The following is the main theorem of this section.

\bt
\label{ext:thm}
For any letter $\a$, any monic monomial $\T$,
$f=\a[\T]-[\T]\a$ can be reduced to zero by {\bf BG}.
\et

By the multiplication table
of basis letters, we assume that in $\T$, 
no two basis letters occur as neighbors. We also assume that
in $\T$, if a sequence of letters of
$\{\q,\bq\}$ for some $\q\in \cal Q$ occur as neighbors, the letters are
sorted as either $\q^i\bq^j$ when $\q\in \cal Q$, or 
$\bq^j\q^i$ when $\q\in \overline{\cal Q}$,
by commutativity $\bq\q=\q\bq$.

Now we start to prove the theorem. When $|\T|\leq 3$,
$f$ can be reduced to zero by $\overline{\bf BG}\2$ to
$\overline{\bf BG}\4$ directly. So we assume $|\T|\geq 4$.
Write $\T=\b\T'\c$, where letters $\b,\c\in {\cal A}$,
and $|\T'|\geq 2$. Let $\1,\2,\3$ be three arbitrary letters in $\cal A$ 
such that 
\be
\3\succ_M \2\bar{\2}\succ_M \1\bar{\1}. 
\ee
All together, $f=\a[\b\T\c]-[\b\T\c]\a$ where $|\T|\geq 2$, has the following 
different types:
\be\ba{lll}
f_1 &=& \1[\1\T\1]-[\1\T\1]\1, \\
f_2 &=& \1[\1\T\bar{\1}]-[\1\T\bar{\1}]\1, \\
f_3 &=& \1[\bar{\1}\T\1]-[\bar{\1}\T\1]\1, \\

f_4 &=& \1[\2\T\1]-[\2\T\1]\1, \\
f_5 &=& \1[\2\T\bar{\1}]-[\2\T\bar{\1}]\1, \\
f_6 &=& \2[\1\T\1]-[\1\T\1]\2, \\
f_7 &=& \2[\1\T\bar{\1}]-[\1\T\bar{\1}]\2, \\

f_8 &=& \1[\2\T\2]-[\2\T\2]\1, \\
f_9 &=& \1[\2\T\bar{\2}]-[\2\T\bar{\2}]\1, \\
f_{10} &=& \2[\2\T\1]-[\2\T\1]\2, \\
f_{11} &=& \2[\bar{\2}\T\1]-[\bar{\2}\T\1]\2, \\

f_{12} &=& \1[\3\T\2]-[\3\T\2]\1, \\
f_{13} &=& \2[\3\T\1]-[\3\T\1]\2,\\
f_{14} &=& \3[\2\T\1]-[\2\T\1]\3,
\ea
\label{list:f1to13}
\ee
where if $\1, \bar{\1}$ occur simultaneously in some $f_i$, 
then $\1\notin \cal E$ is assumed; same for
$\2, \bar{\2}$.

Fix $\1,\2,\3$ and $\T$. Fix the following induction hypothesis ``$\sf IH$" :
\begin{center}
Assume that for all sequences $\T'$ where
$|\T'|<|\T|$, and for all letters $\a',\b',\c'\in {\cal A}$,
$\a'[\b'\T'\c']-[\b'\T'\c']\a'$ can be reduced to zero by 
$\bf BG$. 
\end{center}
We prove that the polynomials in list (\ref{list:f1to13})
can all be reduced to zero by $\bf BG$. 
It turns out that 
even for the most degenerate cases
$f_1$ to $f_3$, the reduction is non-trivial. 

\bp
\label{prop:f1}
$f_1=\1[\1\T\1]-[\1\T\1]\1$ can be 
reduced to zero by $\bf BG$.
\ep

\begin{proof}
If all the letters of $\T$ are in $\{\1,\bar{\1}\}$, the reduction can be done by
element $\1\bar{\1}-\bar{\1}\1$ of $\bf BG2$ directly. So we assume that
there are letters of $\T$ that $\notin \{\1,\bar{\1}\}$.

Case 1. $\1\in \cal E$. 
Then $\1^2=-1$, and $[\1\T\1]=\1[\T]\1$. So
\be
f_1=\1^2[\T]\1-\1[\T]\1^2
\overset{\bf BG2}{=} -[\T]\1+\1[\T]
\overset{\sf IH}{=} 0.
\ee

Case 2. $\1\prec \bar{\1}$. 
The leading term of $f_1$ is
$-\bar{\1}\bar{\T}\bar{\1}\1$.
\be\ba{lcl}
f_1 &=& \1^2\T\1+\1\bar{\1}\bar{\T}\bar{\1}-\1\T\1^2-\underbrace{\bar{\1}\bar{\T}\bar{\1}\1} \\

&\overset{{\bf BG2},\, {\rm Prop.}\, \ref{prop:11}}{=}& \1
(\underbrace{\1\T\1}
+\underbrace{\bar{\1}\bar{\T}\bar{\1}}-\underbrace{\T\1^2}-\bar{\1}^2\bar{\T}) \\

&\overset{\1{\cal I}}{=}& \1
([\1]\T\1-\bar{\1}\T\1+\underbrace{\bar{\1}\bar{\T}[\1]}-\bar{\1}\bar{\T}\1
-\underbrace{\T\1[\1]}+\underbrace{\T\1\bar{\1}}-\bar{\1}^2\bar{\T}) \\

&\overset{\1{\cal I}}{=}& \1
([\1][\T\1]-\underbrace{\bar{\1}[\T]\1}
-[\1]\T\1+\1\bar{\1}\T-\bar{\1}^2\bar{\T}) \\

&\overset{\1{\cal I}}{=}& 0.
\ea
\ee

Case 3. $\1\succ \bar{\1}$. Let $\T=\1^r\y\Z$, where letter 
$\y\neq \1$, $r\geq 0$, and
$\Z$ is a subsequence that may be empty. 
The leading term of $f_1$ is one of 
$\1^{r+2}\y\Z\1, -\1^{r+1}\y\Z\1^2$.

Subcase 3.1. $\y\succ \1$. The leading monomial is
\be\ba{lcl}
\1^{r+1}\underbrace{\y\Z\1\1} &\overset{\sf IH}{=}& 
-\1^{r+1}\bar{\1}\bar{\Z}\bar{\y}\1+\1^{r+2}\y\Z\1
+\1^{r+2}\bar{\1}\bar{\Z}\bar{\y} \\

&=& \1^2\T\1+\bar{\1}\1^{r+1}(\1\bar{\Z}\bar{\y}-\bar{\Z}\bar{\y}\1).
\ea
\label{f1:31}
\ee
So
\be\ba{lcl}
f_1 &\overset{(\ref{f1:31})}{=}& \bar{\1}\{
\1\bar{\Z}\bar{\y}\bar{\1}^{r+1}
-\1^{r+1}(\underbrace{\1\bar{\Z}\bar{\y}-\bar{\Z}\bar{\y}\1})
-\bar{\Z}\bar{\y}\bar{\1}^{r+1}\1\} \\

&\overset{\bar{\1}{\cal I}}{=}& \bar{\1}\{
\1[\1^{r+1}\y\Z]-[\1^{r+1}\y\Z]\1
\}\\

&\overset{\sf IH}{=}& 0,
\ea
\ee
where in the second line, 
$\1\bar{\Z}\bar{\y}-\bar{\Z}\bar{\y}\1=-\1\y\Z+\y\Z\1$ is used.

Subcase 3.2. $\y\prec \1$. Then $\y\preceq \bar{\1}$.
The leading monomial of $f_1$ is
\be\ba{lcl}
\1^{r+1}\underbrace{\1\y}\Z\1 &\overset{\bf BG2}{=}& 
-\underbrace{\1^{r+1}\bar{\1}}\y\Z\1+\1^{r+1}\y[\1]\Z\1 \\

&\overset{{\bf BG2}, \1^{r+1}\y{\cal I}}{=}& 
-\bar{\1}\1^{r+1}\y\Z\1+\1^{r+1}\y\Z[\1]\1 \\

&=& -\bar{\1}\1\T\1+\1\T[\1]\1.
\ea
\label{f1:32}
\ee
So 
\be\ba{lcl}
f_1 &\overset{(\ref{f1:32})}{=}& 
(-\bar{\1}\1\T\1+\underbrace{\1\bar{\1}}\bar{\T}\bar{\1})
+(\1\T[\1]\1-\1\T\1^2)-\bar{\1}\bar{\T}\bar{\1}\1
\\

&\overset{\bf BG2}{=}& -\bar{\1}\1[\T]\1
+\bar{\1}\1\bar{\T}[\1]+\underbrace{\1\T\bar{\1}\1}
-\bar{\1}\bar{\T}\bar{\1}\1\\

&\overset{{\rm Prop}. \ref{prop:11}}{=}& \bar{\1}\{
-\1[\T]\1+\1\bar{\T}[\1]+\1^2\T-\bar{\T}\bar{\1}\1\}\\

&\overset{\bar{\1}{\cal I}}{=}& 0.
\ea\ee

\end{proof}

\bp
\label{prop:f2}
$f_{2}=\1[\1\T\bar{\1}]-[\1\T\bar{\1}]\1$
can be reduced to zero by {\bf BG}.
\ep

\begin{proof}
When $\T$ is composed of letters $\1,\bar{\1}$, $f_{2}$ can
be reduced to zero by $\bar{\1}\1-\1\bar{\1}$ in
$\overline{\bf BG}\2$. So we assume that $\T$ contains
letters $\notin \{\1, \bar{\1}\}$. 
Notice that $[\1\T\bar{\1}]=\1[\T]\bar{\1}$.
Let $[\T]=\T_1+\T_2$ be the two terms of the bracket.
There are three cases.

Case 1. $\T_i\bar{\1}\succeq \bar{\1}\T_i$ for $i=1,2$. Let 
$\T_i=\bar{\1}^{r_i}\x_i\Y_i$ for $i=1,2$, where $r_i\geq 0$, letter
$\x_i\neq \bar{\1}$, and $\Y_i$ is a submonomial, then $\x_i\succ \bar{\1}$. 
When $\1\succ \bar{\1}$, 
\be
f_2 \overset{\sf IH}{=} 
\underbrace{\1\1\bar{\1}}[\T]-\underbrace{\1\bar{\1}}[\T]\1
\overset{\overline{\bf BG}\2}{=}  \bar{\1}\1(\1[\T]-[\T]\1)
\overset{\bar{\1}\1{\cal I}}{=}  0.
\label{f2:case1.1}
\ee
When $\1\prec \bar{\1}$, then $f_2$ has some
$\1\T_i\bar{\1}\1=\1\bar{\1}^{r_i}\x_i\Y_i\bar{\1}\1$ as its
leading term, say $i=1$. Then 
$\1\bar{\1}(\T_1+\T_2)\1$ has leading term 
$\1\bar{\1}\T_1\1=\1\bar{\1}^{r_1+1}\x_1\Y_1\1\prec \1\T_1\bar{\1}\1$, so
\be
f_2 \overset{\sf IH}{=} 
\1\1\bar{\1}[\T]-\1\bar{\1}\underbrace{[\T]\1}
\overset{(\1\bar{\1}^{r_1+1}):{\sf IH}}{=}
\1(\1\bar{\1}-\bar{\1}\1)[\T]
\overset{\overline{\bf BG}\2}{=}  0.
\label{f2:case1.2}
\ee

Case 2. $\1\T_i\succeq \T_i\1$ for $i=1,2$.  Let 
$\T_i=\1^{r_i}\x_i\Y_i$ for $i=1,2$, where $r_i\geq 0$, letter
$\x_i\neq \1$, and $\Y_i$ is a submonomial, then $\x_i\prec \1$. 
Because $\1^2\T_i=\1^{r_i+2}\x_i\Y_i\succ \1^{r_i+1}\x_i\Y_i\bar{\1}=\1\T_i\bar{\1}$,
$f_2$ must have some $\1^2\T_i\bar{\1}$ as
its leading term, say $i=1$. Then 
$\1[\T]\preceq \1\T_1=\1^{r_1+1}\x_1\Y_1$, whose first 
$r_1+2$ letters $\prec \1^{r_1+2}=$ first $r_1+2$ letters of $\1^2\T_1\bar{\1}$.
So
\be
f_2 \overset{\sf IH}{=} \1[\T]\underbrace{\1\bar{\1}}-[\T]\1\bar{\1}\1
\overset{(\1^{r_1+1}\x_1):\overline{\bf BG}\2}{=}  (\1[\T]-[\T]\1)\bar{\1}\1
\overset{\sf IH}{=} 0.
\ee

Case 3. By symmetry, assume
$\T_1\bar{\1}\succeq \bar{\1}\T_1$ but $\T_2\bar{\1}\prec \bar{\1}\T_2$.
Then $\T_1=\bar{\1}^r\x\Y$, where $r\geq 0$, letter $\x\succ \bar{\1}$, and $\Y$ is a submonomial.
Similarly, $\T_2=\bar{\1}^s\y\Z$, where $s\geq 0$, letter $\y\prec \bar{\1}$, and 
$\Z$ is a submonomial.
By 
$
\T_1\bar{\1}\succeq \bar{\1}\T_1\succeq \bar{\1}\T_2\succ 
\T_2\bar{\1}, 
$ we get that
$\T_1\bar{\1}$ is the leading term of $[\T]\bar{\1}-\bar{\1}[\T]$.
By induction hypothesis, $\1[\T]\bar{\1}=\1\bar{\1}[\T]$. Then the reduction of
$f_2$ is the same as in (\ref{f2:case1.1}) and (\ref{f2:case1.2}).
\end{proof}

\bp
\label{prop:f3}
$f_{3}=\1[\bar{\1}\T\1]-[\bar{\1}\T\1]\1$
can be reduced to zero by {\bf BG}.
\ep

\begin{proof}
When $\T$ is composed of letters $\1,\bar{\1}$, $f_{3}$ can
be reduced to zero by $\bar{\1}\1-\1\bar{\1}$ in
$\overline{\bf BG}\2$. So we assume that $\T$ contains
letters $\notin \{\1, \bar{\1}\}$. 
Notice that $[\bar{\1}\T\1]=\bar{\1}[\T]\1$. So
$
f_{3}=\1\bar{\1}[\T]\1-\bar{\1}[\T]\1^2.
$
When $\1\succ \bar{\1}$, then
\be
f_{3} \overset{\overline{\bf BG}\2}{=} \bar{\1}(\1[\T]-[\T]\1)\1
\overset{\bar{\1}{\cal I}}{=} 0.
\ee

In the following, we assume $\1\prec \bar{\1}$.
Let $[\T]=\T_1+\T_2$ be the two terms of the bracket.
There are three cases.

Case 1. $\T_i\1\succeq \1\T_i$ for $i=1,2$. 
\be
f_3 \overset{\sf IH}{=} 
\1\bar{\1}\1[\T]-\underbrace{\bar{\1}\1}[\T]\1
\overset{\overline{\bf BG}\2}{=}  \1\bar{\1}(\1[\T]-[\T]\1)
\overset{\1\bar{\1}{\cal I}}{=}  0.
\label{f3:case1.2}
\ee

Case 2. $\bar{\1}\T_i\succeq \T_i\bar{\1}$ for $i=1,2$.  
\be
f_3 \overset{\sf IH}{=} 
\1[\T]\underbrace{\bar{\1}\1}-[\T]\underbrace{\bar{\1}\1^2}
\overset{\overline{\bf BG}\2}{=} 
(\1[\T]-[\T]\1)\1\bar{\1}
\overset{{\sf IH}}{=} 0.
\label{f3:case2.2}
\ee

Case 3. By symmetry, assume
$\T_1\1\succeq \1\T_1$ but $\T_2\1\prec \1\T_2$.
Then $\T_1=\1^r\x\Y$, where $r\geq 0$, letter $\x\succ \1$, and $\Y$ is a submonomial.
Similarly, $\T_2=\1^s\y\Z$, where $s\geq 0$, letter $\y\prec \1$, and 
$\Z$ is a submonomial. So $\T_1\succ \T_2$ is the leading term of $[\T]$, and
by $\T_1\1\succeq \1\T_1\succ \1\T_2\succ \T_2\1$, the leading term of $[\T]\1-\1[\T]$
is $\T_1\1$.
By induction hypothesis, $\bar{\1}[\T]\1=\bar{\1}\1[\T]$. Then the reduction of
$f_3$ is the same as in (\ref{f3:case1.2}).
\end{proof}

The reduction of $f_4, f_5$ to zero by {\bf BG} will be given as a corollary of
the reduction of $f_{13}$ to zero by {\bf BG} in Proposition \ref{lem:f10,f11}
later in this section.

\bp
\label{prop:f6}
$f_6=\2[\1\T\1]-[\1\T\1]\2$ can be reduced to zero by {\bf BG}.
\ep

\begin{proof}
If $\1\in \cal E$, then the leftmost letter and the 
rightmost letter of $\T$ are not in $\cal E$, so they
have lower order than $\1$. Then
\be
[\1\T\1]=\1[\T]\1\overset{\sf IH}{=}[\T]\1\1=-[\T].
\ee
Again by {\sf IH}, $\2[\T]-[\T]\2$ is reduced to zero. 
So we assume $\1\notin {\cal E}$ below.

By symmetry, let $\1\prec \bar{\1}$.
\be
\ba{lcl}
f_6 &{=}&\underbrace{\2\bar{\1}}\bar{\T}\bar{\1}+\2\1\T\1-[\1\T\1]\2\\

&\overset{\overline{\bf BG}\2}{=}& \2\1(-\bar{\T}\bar{\1}+\T\1)
+[\1]\2\bar{\T}\bar{\1}-[\1\T\1]\2\\

&\overset{\2\1{\cal I}}{=} &\underbrace{\2\1(-\bar{\1}[\T]+[\1]\T)}
+[\1]\2\bar{\T}\bar{\1}-[\1\T\1]\2\\

&\overset{\overline{\bf BG}\2, \overline{\bf BG}\3}{=}&
-\1\bar{\1}\underbrace{\2[\T]}+[\1]\underbrace{\2[\1\T]}-[\1\T\1]\2 \\

&\overset{\1\bar{\1}{\cal I}, [\1]{\cal I}}{=}&
(-\1\bar{\1}[\T]+[\1][\1\T]-[\1\T\1])\2 \\

&\overset{\overline{\bf BG}\2}{=}&
\1([\1\T]-[\T\1])\2 \\

&\overset{\1{\cal I}}{=}& 0.
\ea\ee 
\end{proof}

\bp
\label{prop:f7}
$f_7=\2[\1\T\bar{\1}]-[\1\T\bar{\1}]\2$ can be reduced to zero by {\bf BG}.
\ep

\begin{proof}
If $\T$ is composed of letters $\1, \bar{\1}$, this special case is already covered by 
Proposition \ref{special:f9}. So we assume that $\T$ contains letters $\notin \{\1, \bar{\1}\}$.
Notice that $[\1\T\bar{\1}]=\1[\T]\bar{\1}$.

When $\1\succ \bar{\1}$, 
\be
\ba{lcl}
f_7 &\overset{\overline{\bf BG}\2}{=}&  
-\2\bar{\1}\underbrace{[\T]\bar{\1}}
+[\1]\2\underbrace{[\T]\bar{\1}}-\1\underbrace{[\T]\bar{\1}\2}\\

&\overset{\2\bar{\1}{\cal I}, [\1]{\cal I}, \1{\cal I}}{=}& 
\underbrace{\2\bar{\1}\1}[\T]-\underbrace{\2\bar{\1}[\1]}[\T]
+[\1]\2\bar{\1}[\T]-\bar{\1}\1\2[\T]\\

&\overset{\overline{\bf BG}\2, \overline{\bf BG}\3}{=}& 0.
\ea\ee 

When $\1\prec \bar{\1}$, let $[\T]=\T_1+\T_2$ be the two terms of
$[\T]$. 

Case 1. $\T_i\bar{\1}\succeq \bar{\1}\T_i$ for $i=1,2$. 
Then $\1[\T]\bar{\1}$ is reduced to
$\1\bar{\1}[\T]$ by the induction hypothesis, so 
\be
f_7 \overset{\sf IH}{=}\underbrace{\2\1\bar{\1}}[\T]-\1\bar{\1}[\T]\2 
\overset{\overline{\bf BG}\3}{=}  \1\bar{\1}(\2[\T]-[\T]\2) 
\overset{\1\bar{\1}{\cal I}}{=}  0.
\ee

Case 2. $\1\T_i\succeq \T_i\1$ for $i=1,2$.
Then $\T_i=\1^{r_i}\x_i\Y_i$ where $r_i\geq 0$, 
letter $\x_i\prec \1\prec \2$, and $\Y_i$ is a submonomial. 
By induction hypothesis, $\2[\T]=[\T]\2$. So
\be
f_7 \overset{\sf IH}{=}\underbrace{\2[\T]}\1\bar{\1}-[\T]\1\bar{\1}\2 
\overset{\sf IH}{=}  [\T]\underbrace{\2\1\bar{\1}}-[\T]\1\bar{\1}\2 
\overset{\overline{\bf BG}\3}{=} 0.
\ee

Case 3. 
By symmetry, let 
$\T_1\bar{\1}\succeq \bar{\1}\T_1$ but $\T_2\bar{\1}\prec \bar{\1}\T_2$.
Then $\T_i=\bar{\1}^{r_i}\x_i\Y_i$ for $i=1,2$, where $r_i\geq 0$,
$\Y_i$ is a submonomial, and letter $\x_1\succ \bar{\1}$, letter 
$\x_2\prec \bar{\1}$. So $\T_1\succ \T_2$ is the leading term of $[\T]$. 
By $\T_1\bar{\1}\succeq \bar{\1}\T_1\succ \bar{\1}\T_2\succ \T_2\bar{\1}$,
in $[\T]\bar{\1}-\bar{\1}[\T]$, the leading term is $\T_1\bar{\1}$. So 
by induction hypothesis {\sf IH}, 
$f_7=\2\1\bar{\1}[\T]-\1\bar{\1}[\T]\2$. The rest is the same as in Case 1.
\end{proof}

\bp
$f_{8}=\1[\2\T\2]-[\2\T\2]\1$ can be reduced to zero by {\bf BG}.
\ep

\begin{proof} When $\T$ is composed of letter $\2,\bar{\2}$,
Proposition \ref{special:f8} has covered such special cases. 
We assume that $\T$ contains letters $\notin\{\2,\bar{\2}\}$ below.

If $\2\in \cal E$, then the leftmost letter and rightmost letter of $\T$ 
are not in $\cal E$, so they
have lower order than $\2$. Then
\be
[\2\T\2]=\2[\T]\2\overset{\sf IH}{=}[\T]\2\2=-[\T],
\ee
again by {\sf IH}, $\1[\T]-[\T]\1$ is reduced to zero. 
So we further assume $\2\notin {\cal E}$.

By symmetry, assume
$\2\prec \bar{\2}$, and $\2\T\2=\2^r\bar{\2}^s\U\x\2^t$, where $s, |\U|\geq 0$, and $r,t>0$,
and letter $\x\notin\{\2,\bar{\2}\}$. The leading term of 
$f_8$ is $-\bar{\2}\bar{\T}\bar{\2}\1=-\bar{\2}^{t}\bar{\x}\bar{\U}\2^s\bar{\2}^{r}\1$.

Case 1. $\x\prec\2$.

When $r\geq t$, 
\be\ba{lcl}
[\2\T\2]\1 
&=& \2^r\bar{\2}^s\U\x\2^t\1
+\underbrace{\bar{\2}^t\bar{\x}\bar{\U}\2^s\bar{\2}^r}\\

& \overset{{\rm Prop.}\,\ref{special:f8},\, \bar{\x}{\cal I}}{=} & 
\2^t(\2^{r-t}\bar{\2}^s\U\x\underbrace{\2^t\1}
-\bar{\x}\bar{\U}\2^s\bar{\2}^{r-t}\bar{\2}^{t}\1)+\bar{\x}\bar{\U}\2^s\bar{\2}^r\1[\2^t] \\

& \overset{\2^t{\cal I}}{=} 
& \underbrace{\2^r\bar{\2}^s\U\x\1}[\2^t]
-\underbrace{\2^t[\2^{r-t}\bar{\2}^s\U\x]\bar{\2}^{t}\1}
+\bar{\x}\bar{\U}\2^s\bar{\2}^r\1[\2^t] \\

& \overset{{\sf IH}, \2^t{\cal I}, {\rm Prop.}\, \ref{special:f8},
\1{\cal I}}{=} & 
\1[\2^r\bar{\2}^s\U\x][\2^t]-\1[\2^{r}\bar{\2}^s\bar{\2}^{t}\U\x] \\

& \overset{\1{\cal I}}{=} & 
\1[\2^{r+t}\bar{\2}^s\U\x],
\ea
\ee
where in the next to the last step, reduction $[\2^r\bar{\2}^s\U\x]\1=\1[\2^r\bar{\2}^s\U\x]$ 
is by the induction hypothesis
{\sf IH}, and 
\be
\2^t[\2^{r-t}\bar{\2}^s\U\x]\bar{\2}^{t}\1=\2^t\bar{\2}^{t}\1[\2^{r-t}\bar{\2}^s\U\x]
=\1\2^t\bar{\2}^{t}[\2^{r-t}\bar{\2}^s\U\x]
=\1[\2^{r}\bar{\2}^s\bar{\2}^{t}\U\x]
\ee
is a combination of applying Proposition \ref{special:f8} and $\1{\cal I}$.

When $r\leq t$, the reduction is similar.

Case 2. $\x\succ\bar{\2}$. When $r\geq t$, 
\be\ba{lcl}
{[}\2\T\2]\1 &=& \2^r\bar{\2}^s\U\x\2^t\1
+\bar{\2}^t\underbrace{\bar{\x}\bar{\U}\2^s\bar{\2}^r\1}\\

& \overset{{\sf IH}, \overline{\bf BG}\2}{=} & 
\2^t\underbrace{\2^{r-t}\bar{\2}^s\U\x\2^t\1}
+\underbrace{\bar{\2}^t\1}[\bar{\x}\bar{\U}\2^s\bar{\2}^r]
-\2^r\bar{\2}^t\bar{\2}^s\U\x\1\\

&\overset{\2^t{\cal I}, {\rm Prop.}\, \ref{special:f8}}{=}& 
\2^t(\1[\2^{r-t}\bar{\2}^s\U\x\2^t]-\bar{\2}^t\bar{\x}\bar{\U}\2^s\bar{\2}^{r-t}\1)
+\1[\2^t][\2^r\bar{\2}^s\U\x]\\

&& \hfill
-\2^t\1[\2^r\bar{\2}^s\U\x]
-\2^r\bar{\2}^t\bar{\2}^s\U\x\1\\

&\overset{\2^t{\cal I}, {\rm Prop.}\, \ref{special:f8}}{=}& 
-\1\2^t\bar{\2}^t[\2^{r-t}\bar{\2}^s\U\x]+\1[\2^t][\2^r\bar{\2}^s\U\x]\\

&\overset{\1{\cal I}}{=}& \1[\2^{r+t}\bar{\2}^s\U\x],
\ea
\ee
where in the next to the last step, 
$\2^t\1[\2^{r-t}\bar{\2}^s\U\x\2^t]=\2^t\1[\2^r\bar{\2}^s\U\x]$, and
\be
-\2^t\bar{\2}^t\bar{\x}\bar{\U}\2^s\bar{\2}^{r-t}\1
-\2^t\2^{r-t}\bar{\2}^t\bar{\2}^s\U\x\1
=-\2^t\bar{\2}^t[\2^{r-t}\bar{\2}^s\U\x]\1
=-\1\2^t\bar{\2}^t[\2^{r-t}\bar{\2}^s\U\x]
\ee
is a combination of $\2^{r-t}\bar{\2}^t=\bar{\2}^t\2^{r-t}$,
commutativity $[\2^{r-t}\bar{\2}^s\U\x]\1=\1[\2^{r-t}\bar{\2}^s\U\x]$ from 
induction hypothesis {\sf IH}, and
$\2^t\bar{\2}^t\1=\1\2^t\bar{\2}^t$ from Proposition \ref{special:f8}. 

When $r\leq t$, the reduction is similar.
\end{proof}

\bp
$f_{9}=\1[\2\T\bar{\2}]-[\2\T\bar{\2}]\1$
can be reduced to zero by {\bf BG}.
\ep

\begin{proof}
When $\T$ is composed of letters $\2,\bar{\2}$,
Propositions \ref{special:f8} and \ref{special:f9} have covered such special cases
of $f_9$. So we assume that $\T$ contains letters $\notin \{\2, \bar{\2}\}$. 
Notice that $[\2\T\bar{\2}]=\2[\T]\bar{\2}$.
Let $[\T]=\T_1+\T_2$ be the two terms of the bracket.
There are three cases.

Case 1. $\T_i\bar{\2}\succeq \bar{\2}\T_i$ for $i=1,2$. Then
the leading letter of $\T_i$ is $\succeq \bar{\2}\succ_M \1\bar{\1}$.
\be
f_9 \overset{\sf IH}{=}\1\2\bar{\2}[\T]-\2\bar{\2}\underbrace{[\T]\1} 
\overset{\sf IH}{=}  \1\2\bar{\2}[\T]-\underbrace{\2\bar{\2}\1}[\T] 
\overset{\overline{\bf BG}\3}{=}  0.
\label{f9:case1}
\ee

Case 2. $\2\T_i\succeq \T_i\2$ for $i=1,2$. Then
\be
f_9 \overset{\sf IH}{=} \1[\T]\2\bar{\2}-[\T]\underbrace{\2\bar{\2}\1}
\overset{\overline{\bf BG}\3}{=}  (\1[\T]-[\T]\1)\2\bar{\2}
\overset{\sf IH}{=} 0.
\ee

Case 3. By symmetry, assume
$\T_1\bar{\2}\succeq \bar{\2}\T_1$ but $\T_2\bar{\2}\prec \bar{\2}\T_2$.
Then for $i=1,2$, 
$\T_i=\bar{\2}^{r_i}\x_i\Y_i$, where $r_i\geq 0$, $\Y_i$ is a submonomial,
letter $\x_1\succ \bar{\2}$, and letter $\x_2\prec \bar{\2}$. So
$\T_1\succ \T_2$ is the leading term of $[\T]$. 
By $\T_1\bar{\2}\succeq \bar{\2}\T_1\succ \bar{\2}\T_2\succ \T_2\bar{\2}$, 
$\T_1\bar{\2}$ is the leading term of $[\T]\bar{\2}-\bar{\2}[\T]$.
By induction hypothesis, $\2[\T]\bar{\2}=\2\bar{\2}[\T]$. The rest
is the same as in (\ref{f9:case1}).
\end{proof}

\bp
$f_{10}=\2[\2\T\1]-[\2\T\1]\2$
can be reduced to zero by {\bf BG}.
\ep

\begin{proof} In
$f_{10} = \2\2\T\1+\2\bar{\1}\bar{\T}\bar{\2}
-\2\T\1\2-\bar{\1}\bar{\T}\bar{\2}\2$, 
when $\2\in \cal E$, then $\T$ has its leading letter $\prec \2$, so
\be
f_{10}=-\T\1-\2\bar{\1}\bar{\T}\2-\2\T\1\2-\bar{\1}\bar{\T}
=-\underbrace{\2[\T\1]}\2-[\T\1]
\overset{{\sf IH}, {\bf BG}\2}{=} 0.
\ee
In the following, we assume $\2\notin \cal E$.

When $\T=\2^s$, where $s\geq 2$, 
\be
\ba{lcl}
f_{10} &=& \2\underbrace{\2\2^s\1}
+\2\bar{\1}\bar{\2}^s\bar{\2}-\2\2^s\1\2
-\bar{\1}\bar{\2}^s\bar{\2}\2\\

&\overset{{\sf IH}}{=}&
-\underbrace{\2\2\bar{\1}}\bar{\2}^s
+\2\bar{\1}\underbrace{\bar{\2}^s[\2]}-\bar{\1}\underbrace{\bar{\2}^s\bar{\2}\2}\\

&\overset{\overline{\bf BG}\3,
\2\bar{\1}{\cal I},\bar{\1}{\cal I}}{=}&
(-[\2\bar{\1}]\2
+\2\1\bar{\2}
+\2\bar{\1}[\2]-\bar{\1}\2\bar{\2})\bar{\2}^s\\

&\overset{\overline{\bf BG}\2}{=} &
(\underbrace{\2[\1]}\bar{\2}-[\1]\2\bar{\2})\bar{\2}^s\\

&\overset{\overline{\bf BG}\2}{=}&0.
\ea\ee
In the following, we further assume that $\T$ is not a power of $\2$.

Case 1. $\T=\2^s\x\Y$, where $s\geq 0$, letter $\x\prec \2$, and $\Y$ is 
a submonomial. 
Then $\2\2\T\1\succ \2\T\1\2$ is the leading term of $f_{10}$. 
If $s=0$ and $\x\prec_M\2\bar{\2}$, then 
\be
\ba{lcl}
f_{10} &=& 
\underbrace{\2\2\x}\Y\1+\underbrace{\2\bar{\1}\bar{\Y}\bar{\x}}
\bar{\2}-\2\x\Y\1\2-\bar{\1}\bar{\Y}\bar{\x}\bar{\2}\2\\

&\overset{\overline{\bf BG}\3,(\2\bar{\1},\2\x):{\sf IH}}{=}&
-\underbrace{\2\bar{\x}}\bar{\2}\Y\1+[\2\x]\2\Y\1
-\2\x\Y\1\bar{\2}+[\x\Y\1]\2\bar{\2}
-\2\x\Y\1\2-\bar{\1}\bar{\Y}\bar{\x}\underbrace{\bar{\2}\2}\\

&\overset{\overline{\bf BG}\2, \bar{\1}{\cal I}}{=}& 
\2\x(\underbrace{\bar{\2}\Y\1+\2\Y\1-\Y\1\bar{\2}-\Y\1\2})-\x(\underbrace{\2\bar{\2}\Y\1-\Y\1\2\bar{\2}})\\

&\overset{\2\x{\cal I},\x{\cal I}}{=}&0.
\label{f10:case2}
\ea\ee
If either $s>0$, or $s=0$ and $\x=\bar{\2}\prec\2$, then
\be
\ba{lcl}
f_{10} &=& \2\underbrace{\2\2^s\x\Y\1}+\2\bar{\1}\bar{\Y}\bar{\x}\bar{\2}^s\bar{\2}
-\2\2^s\x\Y\1\2-\bar{\1}\bar{\Y}\bar{\x}\bar{\2}^s\bar{\2}\2\\

&\overset{{\sf IH}}{=}&
-\underbrace{\2\2\bar{\1}}\bar{\Y}\bar{\x}\bar{\2}^s+\2\bar{\1}\underbrace{\bar{\Y}\bar{\x}\bar{\2}^s[\bar{\2}]}
-\bar{\1}\underbrace{\bar{\Y}\bar{\x}\bar{\2}^s\bar{\2}\2}\\

&\overset{\overline{\bf BG}\3,\2\bar{\1}{\cal I},\bar{\1}{\cal I}}
{=}&
(
\2\1\bar{\2}-\2\bar{\1}\2-\1\bar{\2}\2
+\2\bar{\1}[\bar{\2}]-\bar{\1}\bar{\2}\2)\bar{\Y}\bar{\x}\bar{\2}^s\\[2mm]

&=& 
(\underbrace{\2[\1]}\bar{\2}-[\1]\underbrace{\bar{\2}\2})\bar{\Y}\bar{\x}\bar{\2}^s\\

&\overset{{\overline{\bf BG}\2},[\1]{\cal I}}{=}&0.
\ea\ee

Case 2. $\T=\2^s\x\Y$, where $s\geq 0$, letter $\x\succ \2$, and $\Y$ is
a submonomial. Then $\2\T\1\2\succ \2\2\T\1$ is the leading term of $f_{10}$.
\be
\ba{lcl}
f_{10} &=& \2\2\2^s\x\Y\1+\2\bar{\1}\bar{\Y}\bar{\x}\bar{\2}^s\bar{\2}
-\2\underbrace{\2^s\x\Y\1\2}-\bar{\1}\bar{\Y}\bar{\x}\bar{\2}^s\bar{\2}\2\\

&\overset{\sf IH}{=}& 
\2\2\2^s\x\Y\1+\2\bar{\1}\bar{\Y}\bar{\x}\bar{\2}^s\bar{\2}
-\2\2[\2^s\x\Y\1]
+\2\bar{\1}\bar{\Y}\bar{\x}\bar{\2}^s\2
-\bar{\1}\underbrace{\bar{\Y}\bar{\x}\bar{\2}^s\bar{\2}\2}\\

&\overset{\2\bar{\1}{\cal I},\bar{\1}{\cal I}}{=}&
(\2\bar{\1}[\bar{\2}]-\underbrace{\2\2\bar{\1}}-\bar{\1}\bar{\2}\2)\bar{\Y}\bar{\x}\bar{\2}^s\\

&\overset{\overline{\bf BG}\3}{=}& 
(\underbrace{\2[\1]}\bar{\2}-[\1]\underbrace{\bar{\2}\2})\bar{\Y}\bar{\x}\bar{\2}^s\\

&\overset{{\overline{\bf BG}\2},[\1]{\cal I}}{=}&0.
\ea\ee

\end{proof}

\bp
$f_{11}=\2[\bar{\2}\T\1]-[\bar{\2}\T\1]\2$ can be reduced to zero by {\bf BG}.
\ep

\begin{proof} When $\bar{\2}\prec \2$, $f_{11}$ is led by $\2\bar{\2}\T\1$. We have
\be
\ba{lcl}
f_{11} &=& \underbrace{\2\bar{\2}}\T\1
+\underbrace{\2\bar{\1}}\bar{\T}\2-\bar{\2}\T\1\2-\bar{\1}\bar{\T}\2^2\\

&\overset{\overline{\bf BG}\2}{=}&
\bar{\2}(\2\T\1-[\bar{\1}\bar{\T}]\2)
+\bar{\1}([\2]\bar{\T}\2-\bar{\T}\2^2)\\

&\overset{\bar{\2}{\cal I}, \bar{\1}{\cal I}}{=}& 
-\underbrace{\bar{\2}\2\bar{\1}}\bar{\T}+\bar{\1}\bar{\T}\bar{\2}\2\\

&\overset{\overline{\bf BG}\3}{=}& \bar{\1}(-\bar{\2}\2\bar{\T}+\bar{\T}\bar{\2}\2)\\

&\overset{\bar{\1}{\cal I}}{=}& 0.
\ea\ee

In the following, let $\bar{\2}\succ \2$. Then $f_{11}$ is led by $-\bar{\2}\T\1\2$. 
Let $\T=\x\Y$. When $\x\preceq \2$, then 
\be
\ba{lcl}
f_{11} &=& \underbrace{\2\bar{\2}\x}\Y\1
+\2\bar{\1}\bar{\Y}\bar{\x}\2-\underbrace{\bar{\2}\x}\Y\1\2
-\bar{\1}\bar{\Y}\bar{\x}\2^2\\

&\overset{\overline{\bf BG}\3,\overline{\bf BG}\2}{=}& 
\x(\2\bar{\2}\Y\1-[\2]\Y\1\2)
+\2[\x\Y\1]\2-\bar{\1}\bar{\Y}\bar{\x}\2^2\\

&\overset{\x{\cal I}, \2{\cal I}}{=}& 0.
\ea\ee
When $\x\succeq \bar{\2}$, then
\be
\ba{lcl}
f_{11} &=& \2\bar{\2}\x\Y\1
+\2\bar{\1}\bar{\Y}\bar{\x}\2
-\bar{\2}\underbrace{\x\Y\1\2}-\bar{\1}\bar{\Y}\bar{\x}\2^2\\

&\overset{\sf IH}{=}& \2\bar{\2}\x\Y\1
+\underbrace{[\2]\bar{\1}}\bar{\Y}\bar{\x}\2
-\underbrace{\bar{\2}\2}[\x\Y\1]
-\bar{\1}\bar{\Y}\bar{\x}\2^2\\

&\overset{\overline{\bf BG}\2, \bar{\1}{\cal I}}{=}& 
-\2\bar{\2}(\bar{\1}\bar{\Y}\bar{\x})+(\bar{\1}\bar{\Y}\bar{\x})\2\bar{\2} \\

&\overset{{\rm Prop.}\ \ref{prop:11}}{=}& 0.
\ea\ee
\end{proof}

In the following computations involving letters $\1,\2,\3$, 
those terms led by letters $\1,\2$ or their conjugates
are usually not displayed, and are collected in the set called ``lower".

\bp
\label{lem:f10,f11}
$f_{12}=\1[\3\T\2]-[\3\T\2]\1$ and 
$f_{13}=\2[\3\T\1]-[\3\T\1]\2$
can both be reduced to zero by {\bf BG}.
\ep

\begin{proof}
We first make reduction to $f_{13}$, and then show that the reduction procedure also
works for $f_{12}$. In $f_{13}$, the only term led by letter $\3$
or its conjugate is $-\3\T\1\2$. We first reduce this term to a polynomial
where each term is led by one of the letters in 
${\cal B}=\{\1, \bar{\1},
\2, \bar{\2}\}$, then use the bottom-letter-$\cal B$ controlled reduction 
technique to finish the proof. 

Let $\T=\X\y$, where $\y$ is the last letter of $\T$, and $|\X|>0$. 
When $\y\prec_M\3\bar{\3}$, 
\be\ba{lcl}
-\underbrace{\3\X\y\1}\2 &\overset{\sf IH}{=}& 
\bar{\y}\bar{\X}\underbrace{\bar{\3}\1\2}+\hbox{ lower}\\

&\overset{\overline{\bf BG}\3}{=}& 
\bar{\y}\bar{\X}(\2\bar{\3}\1+\2\bar{\1}\3-\bar{\1}\3\2)
+\hbox{ lower}\\[2mm]

&=& 
\bar{\y}\bar{\X}([\2\bar{\3}\1]-\bar{\1}\3[\2])
+\underbrace{\bar{\y}\bar{\X}\2\bar{\1}}\3
+\hbox{ lower}\\

&\overset{\bar{\y}{\cal I}, (\bar{\y},\bar{\2}):{\sf IH}}{=}& 
\bar{\y}[\2\bar{\3}\1]\bar{\X}
-\bar{\y}[\2]\bar{\X}\bar{\1}\3
+\bar{\1}[\bar{\y}\bar{\X}\2]\3
-\bar{\2}\X\y\bar{\1}\3
+\hbox{ lower}\\

&\overset{(\bar{\y}, [\2]):}{=}& 
[\2\bar{\3}\1]\bar{\y}\bar{\X}
-[\2]\bar{\y}\bar{\X}\bar{\1}\3
+\hbox{ lower}.
\ea
\label{lem:3:step1}
\ee

When $\y\succeq \3$ or $\bar{\3}$, then $\y\succ_M\2\bar{\2}$.
\be\ba{lcl}
-\3\X\underbrace{\y\1\2} &\overset{\bf BG3}{=}& 
-\3\X\2\y\1-\3\X\2\bar{\1}\bar{\y}+\3\X\bar{\1}\bar{\y}\2\\

&=& 
-\3\X[\2\y\1]
-\underbrace{\3\X\2\bar{\1}}\bar{\y}+\underbrace{\3\X\bar{\1}\bar{\y}[\2]}\\

&\overset{{\sf IH}, {\rm Prop.}\ \ref{prop:11}}{=}& 
-\3\X[\2\y\1]+\hbox{ lower}.
\ea
\label{lem:3:step2:1}
\ee

Below we prove that $f=-\3\X[\2\y\1]+[\2\y\1]\3\X$ can be reduced to zero by 
$\bf BG$, where $\y\succeq \3$ or $\bar{\3}$. When $|\X|=0$, then 
$f$ can be easily reduced to zero by $\overline{\bf BG}\bf 2$ to
$\overline{\bf BG}\bf 4$. So we assume $|\X|>0$.
Fix $\X$. Suppose that the conclusion holds for all 
monomials of length $<|\X|$ when $\X$ is replaced by the monomial. In monomial
$\X$, if there is any letter $\succ_M\2\bar{\2}$, say $\X=\U_1\z\U_2$,
where letter $\z\succ_M\2\bar{\2}$, and $\U_1, \U_2$ are submonomials
of length $<|\X|$, then applying the induction hypothesis
twice, we get 
\be
-\3\X[\2\y\1] = -\3\U_1\underbrace{\z\U_2[\2\y\1]} 
=-\underbrace{\3\U_1[\2\y\1]}\z\U_2
=-[\2\y\1]\3\U_1\z\U_2.
\label{3x2y1:0}
\ee 
If all letter of $\X$ are 
$\preceq \2$ or $\bar{\2}$, then
\be\ba{lcl}
-\3\X[\2\y\1] &=& 
-\underbrace{\3\X\2}\y\1-\underbrace{\3\X\bar{\1}}\bar{\y}\bar{\2} \\

&\overset{\sf IH}{=}& 
\bar{\X}\bar{\3}\2\y\1+\bar{\X}\bar{\3}\bar{\1}\bar{\y}\bar{\2} 
-\2\underbrace{[\3\X]\y\1}-\bar{\1}\underbrace{[\3\X]\bar{\y}\bar{\2}}\\

&\overset{\bar{\X}{\cal I}, \2{\cal I}, \bar{\1}{\cal I}}{=}& 
\underbrace{\bar{\X}[\2\y\1]}\bar{\3} 
-(\2\y\1+\bar{\1}\bar{\y}\bar{\2})[\3\X]
\\

&\overset{(\2,\bar{\2}):{\sf IH}}{=}& -[\2\y\1]\3\X.
\ea
\label{3x2y1:1}
\ee
This proves that the result of (\ref{lem:3:step2:1}) can be 
reduced to ``lower" by {\bf BG}.

The reduction results of (\ref{lem:3:step1}), (\ref{lem:3:step2:1}) are denoted by
$f', f''$ respectively. They are both polynomials where each term is led by
a letter in $\cal B$. Now in each polynomial,
replace $\X$ by $\4$, and replace $\y$ by $\5$. In the 
new variables with order $\1\prec\2\prec \3\prec \4\prec \5$, do 
reductions to $f', f''$ by the corresponding $\overline{\bf BG}$ in the new variables.
The reduction results become zero in 363 steps and 180 
steps respectively. 

Next consider $f_{12}$, which is obtained from $f_{13}$ by switching $\1,\2$.
Under this switch, the procedure (\ref{lem:3:step1}) is still a correct
reduction of $f_{12}$ when $\y\prec_M\3\bar{\3}$, so is (\ref{lem:3:step2:1})
when $\y\succeq \3$ or $\bar{\3}$. Since $[\2\y\1]=[\bar{\1}\bar{\y}\bar{\2}]$,
(\ref{3x2y1:0}) and (\ref{3x2y1:1}) naturally works for $f_{12}$. 
The final bottom-letter-${\cal B}$ controlled reduction remains a correct one 
under the switch of $\1, \2$ So 
the whole reduction procedure of $f_{13}$ works for $f_{12}$.
\end{proof}

\bc
\label{prop:f4&5}
$f_4 = \1[\2\T\1]-[\2\T\1]\1$ and $f_5 = \1[\2\T\bar{\1}]-[\2\T\bar{\1}]\1$ 
can both be reduced to zero by {\bf BG}.
\ec

\begin{proof} 
To reduce $f_4$, set the letter $\1$ outside bracket $[\2\T\1]$ to be new letter $\2$, and set
the old letter $\2$ to be new letter $\3$. Then $f_4$ is rewritten as $f_{13}$. 
We check the reduction procedure of $f_{13}$ in the proof of Proposition 
\ref{lem:f10,f11}
to verify that each step works for $f_4$ in old letters.

The reduction procedure (\ref{lem:3:step1}) remains a correct
reduction of $f_4$ when $\y\prec_M\2\bar{\2}$ and $\2=\1$ in old letters, 
so does (\ref{lem:3:step2:1})
when $\y\succeq \2$ or $\bar{\2}$ in old letters. In (\ref{3x2y1:0}), where 
letter $\z\succ_M\1\bar{\1}$ in old letters, the reduction by induction still works, so does
(\ref{3x2y1:1}) where all letters of $\X$ are 
$\preceq \1$ or $\bar{\1}$ in old letters. In the bottom-letter controlled reduction of
$f', f''$ from $f_{13}$, when new letters $\1,\2$ are replaced by the corresponding old ones
$\1,\1$ respectively, the reduction remains to be correct. So 
the whole reduction procedure of $f_{13}$ works for $f_4$.

Similarly, to reduce $f_5$, set the letter $\1$ outside bracket $[\2\T\bar{\1}]$
to be new letter $\2$, set the letter $\bar{\1}$ inside the bracket to be new $\1$,
and set the old letter $\2$ to be new letter $\3$. Then $f_5$ is rewritten as $f_{13}$. 
Checking the validity of the reduction procedure of $f_{13}$ in the proof of Proposition 
\ref{lem:f10,f11} for the reduction of $f_5$ in old letters is much the same as above.
\end{proof}

\bp \label{prop:f12}
$f_{14}=\3[\2\T\1]-[\2\T\1]\3$ can be reduced to zero by {\bf BG}.
\ep

\begin{proof}
There are three cases.

Case 1. In $\T$, there exist two neighboring letters $\y, \z$, such that
$\y\succ_M\z\bar{\z}\succ_M \1\bar{\1}$. Let $\T=\X\y\Z$, where submonomial
$\X$ may have degree 0, and submonomial $\Z$ has $\z\succ_M \1\bar{\1}$
as its first letter.
Then
\be
\label{reduction:32T1_0}
\ba{lcl}
f_{14} &=& \3[\2\X\y\Z\1]-[\2\X\y\Z\1]\3 \\

&=&\3\2\X\underbrace{\y\Z\1}
+\3\bar{\1}\bar{\Z}\bar{\y}\overline{\X}\bar{\2}
+\hbox{ lower} \bigstrut\\

&\overset{{\sf IH}}{=}& 
-\underbrace{\3\2\X\y\bar{\1}}\bar{\Z}
+\underbrace{\3\2\X\bar{\1}}\bar{\Z}\y
+\underbrace{\3\2\X\Z\1}\y 
+\3\bar{\1}\bar{\Z}\bar{\y}\overline{\X}\bar{\2}
+\hbox{ lower} \bigstrut\\

&\overset{{\sf IH}}{=}& 
\3\1\underbrace{\bar{\y}\overline{\X}\bar{\2}}\bar{\Z}
-\3\1\underbrace{\overline{\X}\bar{\2}}\bar{\Z}\y
-\3\bar{\1}\underbrace{\bar{\Z}\overline{\X}\bar{\2}}\y
+\3\bar{\1}\underbrace{\bar{\Z}\bar{\y}\overline{\X}\bar{\2}}
+\hbox{ lower} \bigstrut\\

&\overset{\3\1{\cal I}, \3\bar{\1}{\cal I}}{=}& 
(-\3\1\X\y\bar{\2}\bar{\Z}
+\underbrace{\3\1\bar{\2}}[\X\y]\bar{\Z})
+(
\3\1\X\bar{\2}\bar{\Z}\y
-\underbrace{\3\1\bar{\2}}[\X]\bar{\Z}\y)
\bigstrut
\\

&& +(
\underbrace{\3\bar{\1}}\X\Z\bar{\2}\y
-\underbrace{\3\bar{\1}\bar{\2}}[\X\Z]\y)
+(\underbrace{\3\bar{\1}\bar{\2}}[\X\y\Z]-\underbrace{\3\bar{\1}}\X\y\Z\bar{\2})
+\hbox{ lower} \bigstrut\\

&\overset{\overline{\bf BG}{\bf 2}, \overline{\bf BG}{\bf 3}}{=}& 
\3\1\X(-\underbrace{\y\bar{\2}\bar{\Z}}
+\bar{\2}\bar{\Z}\y
-\Z\bar{\2}\y
+\y\Z\bar{\2})
+\hbox{ lower} \bigstrut\\

&\overset{\3\1\X{\cal I}}{=}& 
\3\1\X(\y\Z[\2]-\Z[\2]\y) 
+\hbox{ lower} \bigstrut\\

&\overset{\3\1{\cal I}}{=}& \underbrace{\3\1[\2]}
\X(\y\Z-\Z\y)
+\hbox{ lower} \bigstrut\\

&\overset{\overline{\bf BG}{\bf 3}}{=}& \hbox{ lower}.
\ea
\ee

Denote the result of (\ref{reduction:32T1_0}) by $f'$. It
contains 48 ``lower terms". 
Now replace $\X,\y,\Z$ by new letters $\4,\5,\6$ respectively, and set
new order
$\1\prec \2\prec \3\prec
\4\prec \5\prec \6.$
In the new letters, do reduction to $f'$
with the corresponding {\bf BG} in new letters. In
947 steps, $f'$ is reduced to zero by this bottom-letter controlled reduction.

Case 2. In $\T$, there exist two neighboring letters $\y, \z$, such that
$\y\succ_M\z\bar{\z}$, but $\z\preceq \1$ or $\bar{\1}$. Let $\T=\Y\z\U$, where 
the last letter of submonomial
$\Y$ is $\y\succ_M\z\bar{\z}$, and $\U$ is a submonomial of degree $\geq 0$. 

First, we show that $h=\3\z[\2\Y\bar{\1}]-[\2\Y\bar{\1}]\3\z$ 
can be reduced to zero by {\bf BG}.
\be\ba{lcl}
h &=& \underbrace{\3\z\2}\Y\bar{\1}+\underbrace{\3\z\1}\bar{\Y}\bar{\2}
-[\2\Y\bar{\1}]\3\z \\

& \overset{\overline{\bf BG}\3}{=}&
-\bar{\z}\underbrace{\bar{\3}[\2\Y\bar{\1}]}+\2\underbrace{[\3\z]\Y\bar{\1}}
+\1\underbrace{[\3\z]\bar{\Y}\bar{\2}}-[\2\Y\bar{\1}]\3\z \\

& \overset{{\sf IH}, \2{\cal I}, \1{\cal I}}{=}&
-\underbrace{\bar{\z}[\2\Y\bar{\1}]}\bar{\3}+[\2\Y\bar{\1}][\3\z]
-[\2\Y\bar{\1}]\3\z \\

& \overset{(\2):{\sf IH}}{=}& 0,
\ea
\label{f14:repair}
\ee
where in the last step, $\bar{\z}[\2\Y\bar{\1}]=[\2\Y\bar{\1}]\bar{\z}$ 
is a reduction dominated by $\2$.

Next, 
\be
\ba{lcl}
f_{14} 
&=& \underbrace{\3\2\Y\z}\U\1
+\underbrace{\3\bar{\1}\overline{\U}\overline{\z}}\overline{\Y}\bar{\2}
+\hbox{ lower} \bigstrut\\

&\overset{{\sf IH}, (\3\bar{\1},\3\z):}{=}&
-\3\overline{\z}\underbrace{\overline{\Y}\bar{\2}\U\1}
+\overline{\z}\overline{\Y}\bar{\2}\3\U\1
-\3\z\U\1\overline{\Y}\bar{\2}
+\z\U\1\3\overline{\Y}\bar{\2}
+\hbox{ lower} \bigstrut\\

&\overset{\3\bar{\z}{\cal I}, (\3[\z]):}{=}&
-\underbrace{\3\overline{\z}\1}[\overline{\Y}\bar{\2}\U]
+\underbrace{\3[\z]}\overline{\U}\2\Y\1
-\3\z(\bar{\U}\2\Y\1+\U\1\bar{\Y}\bar{\2})
+\overline{\z}\overline{\Y}\bar{\2}\3\U\1
+\z\U\1\3\overline{\Y}\bar{\2}\\

&& \hfill
+\hbox{ lower}.
\ea
\label{case2:f12}
\ee

We have
\be\ba{cl}
& \3\z(\bar{\U}\2\Y\1+\U\1\bar{\Y}\bar{\2})\\[2mm]

= & \3\z([\U]\2\Y\1-\U\2\Y[\1]+\U[\2\Y\bar{\1}])\\ 

\overset{\3\z{\cal I}}{=}&
\underbrace{\3\z\2}\Y\1[\U]
-\underbrace{\3\z[\1]}\U\2\Y+\underbrace{\3\z[\2\Y\bar{\1}]}\U\\ 

\overset{\overline{\bf BG}\3, {\rm Prop.}\ \ref{prop:11},
(\ref{f14:repair})}{=}&
-\bar{\z}\bar{\3}\2\Y\1[\U]
+\hbox{ lower}.
\ea
\label{special:f12}
\ee
Substituting (\ref{special:f12}) into (\ref{case2:f12}), we get
\be
\ba{lcl}
f_{14} 
&\overset{\overline{\bf BG}, (\ref{special:f12})}{=}&
\z\bar{\3}\1[\bar{\Y}\bar{\2}\U]
+[\z]\3\overline{\U}\2\Y\1
+\bar{\z}\bar{\3}\2\Y\1[\U]
+\overline{\z}\overline{\Y}\bar{\2}\3\U\1
+\z\U\1\3\overline{\Y}\bar{\2}
+\hbox{ lower}.\\
\ea
\label{case2:f12:2}
\ee
Denote the result by $f'$. 

Introduce new variables to $f'$: 
\be
(\1, \2, \z, \3, \Y, \U) \to (\1,\2, \3, \4, \5, \6).
\ee
In the new variables, make reduction to $f'$ by $\overline{\bf BG}\2$ to 
$\overline{\bf BG}\6$.
In 1268 steps, the result is reduced to zero. 
Below we verify the correctness of this substitutional reduction 
in old letters.

In new letters, $f'$ becomes
\be
\3\bar{\4}\1[\bar{\5}\bar{\2}\6]
+[\3]\4\bar{\6}\2\5\1
+\bar{\3}\bar{\4}\2\5\1[\6]
+\bar{\3}\bar{\5}\bar{\2}\4\6\1
+\3\6\1\4\bar{\5}\bar{\2}
+\hbox{ lower}.
\label{expr:fp}
\ee
Among the elements of ${\bf BG}[\1, \ldots, \6]$ in new letters
that are applicable to make reduction to the above polynomial,
\bi
\item all polynomials of the form $\a{\cal I}$, 
$\a\in \{\1, \bar{\1}, \2, \bar{\2}, \z, \bar{\z}\}$,
are allowed in the reduction of $f_{14}$, because
in old letters, $\3\succ_M \z\bar{\z}\1\bar{\1}\2\bar{\2}$.
This certifies all the elements of the form 
$\b{\cal I}$, where $\b\in \{\1, \bar{\1}, \2, \bar{\2}, \3, \bar{\3}\}$
in new letters, in the substitutional reduction. 

\item The elements of ${\bf BG}[\1, \ldots, \6]$ 
led by letters of
$\{\1, \bar{\1}, \2, \bar{\2}\}$ are also allowed in the reduction
of $f_{14}$, by
Proposition \ref{prop:bottom-reduce}. 

\item $\bar{\3}\3-\3\bar{\3}$ in new letters, is allowed in old letters, 
which is just $\bar{\z}\z-\z\bar{\z}$.

\item The elements of ${\bf BG}[\1, \ldots, \6]$ changing the
leading letter of a monomial
from $\3$ or $\bar{\3}$ to a letter of
$\{\1, \bar{\1}, \2, \bar{\2}\}$, are the following:
\bi
\item ${\bf BG}\2$: 
$[\3]\1-\1[\3],\ [\3]\bar{\1}-\bar{\1}[\3], \
[\3]\2-\2[\3],\ [\3]\bar{\2}-\bar{\2}[\3], \
\3[\1]-[\1]\3,\ \3[\2]-[\2]\3$;

\item ${\bf BG}\3$: $[\3\2]\1-\1[\3\2],\ [\3\1]\2-\2[\3\1], 
\ [\3\1]\bar{\2}-\bar{\2}[\3\1]$;

\item ${\bf BG}\4$ to ${\bf BG}\6$: 
$\3[\2\S\1]-[\2\S\1]\3$ for any increasing 
sequence $\S\succ_M\bar{\3}$ of length ranging from 1 to 3.
\ei
\ei
In old letters, the above polynomials of ${\bf BG}\2$ to ${\bf BG}\3$
are now in $\overline{\bf BG}\2$ to $\overline{\bf BG}\3$. The above
polynomials in ${\bf BG}\4$ to ${\bf BG}\6$ are now
of the form $\z[\2\S\1]-[\2\S\1]\z$, where $|\S|\leq |\T|$; these polynomials
are of one of the types $f_4, f_5, f_{12}$, 
because $\z\preceq \1$ or $\bar{\1}$; these types are already
proved to be reducible to zero by {\bf BG}.

Case 3. $\T$ is non-decreasing. Let $\T=\x\Y$, where $\x$ is the first letter, and
$\Y\succeq_M \x$ is a non-decreasing submonomial.
If all letters of $\T$ are $\succeq \3$ or $\bar{\3}$,
then already $f_{14}\in \overline{\bf BG}$. So we assume $\x\prec_M \3\bar{\3}$.
\be
\ba{lcl}
f_{14} &=& \underbrace{\3\2\x}\Y\1
+\3\bar{\1}\underbrace{\overline{\Y}\overline{\x}\bar{\2}}
+\hbox{ lower}\\

&\overset{\overline{\bf BG}{\bf 3}, \3\bar{\1}{\cal I}}{=}&
\x[\3\2]\Y\1+\underbrace{\3\bar{\1}\bar{\2}}[\x\Y]
-\underbrace{\3\bar{\1}\x}\Y\bar{\2}
+\hbox{ lower} \bigstrut\\

&\overset{\overline{\bf BG}{\bf 3}}{=}&
\x([\3\2]\Y\1-\x[\3\bar{\1}]\Y\bar{\2})
+\hbox{ lower} \bigstrut\\

&=&
\x(\underbrace{\3[\2\Y\1]}-\3\bar{\1}\underbrace{[\Y]\bar{\2}}
+[\bar{\2}\bar{\3}\Y\1]-\underbrace{\bar{\1}}\bar{\Y}\3\2
-\1\bar{\3}\Y\bar{\2})
+\hbox{ lower} \bigstrut\\

&\overset{\x{\cal I}}{=}&
\x([\2\Y\1]\3-\3\bar{\1}\bar{\2}[\Y]
+[\bar{\2}\bar{\3}\Y\1]-[\1]\bar{\Y}\3\2
+\1\underbrace{[\bar{\Y}\3]\2}
-\underbrace{\1\bar{\3}\Y[\2]})
+\hbox{ lower} \bigstrut\\

&\overset{\x{\cal I},(\x,[\1],[\2]):}{=}&
[\2\Y\1]\x\3
-\underbrace{\x\3\bar{\1}\bar{\2}}[\Y]
+[\bar{\2}\bar{\3}\Y\1]\x-[\1]\x\bar{\Y}\3\2
+\underbrace{\x\1\2}[\bar{\Y}\3]
-[\2]\x\1\bar{\3}\Y

+\hbox{ lower} \bigstrut\\

&\overset{(\x,\2,\bar{\2}):}{=}& \hbox{ lower},\bigstrut
\ea
\label{reduction:32T1_6}
\ee
where in the next to the last step,
\bi
\item
$\x[\2\Y\1]=[\2\Y\1]\x$ is by the reduction hypothesis {\sf IH},
the leading term during the reduction
is led by letter $\x$ or $\2$.

\item $\x[\bar{\2}\bar{\3}\Y\1]=[\bar{\2}\bar{\3}\Y\1]\x$ 
is dominated by letter $\x$ or $\bar{\2}$: 
\bi
\item when $\x\succ \2\bar{\2}$, this reduction is of type $\overline{\bf BG}$;
\item when $\x\in \{\2, \bar{\2}\}$, this reduction is of type $f_{10}$ or $f_{11}$;
\item when $\1\bar{\1}\prec_M \x\prec_M \2\bar{\2}$, this reduction is of type $f_{13}$; 
\item when $\x\in \{\1, \bar{\1}\}$, this reduction is of type $f_4$ or $f_5$;
\item when $\x\prec_M \1\bar{\1}$, this reduction is of type $f_{12}$.
\ei
\ei

Denote the result of (\ref{reduction:32T1_6}) by $f'$. 
Now introduce new letters: replace $\x, \Y$ by $\4, \5$ respectively.
Set the following order: $\1\prec\2\prec \3\prec \4\prec \5$. 
In the new letters, use the corresponding ${\bf BG}$
to make bottom-letter controlled reduction to $f'$. 
The result becomes zero in 207 steps. 
\end{proof}

By now we have proved Theorem \ref{ext:thm}. Denote the extension of $\overline{\bf BG}$ by
this theorem by ${\bf BG}^{ext}$:
\be
{\bf BG}^{ext}:=\overline{\bf BG}\cup \{f_1, f_2, \ldots, f_{14}\}.
\ee

\section{Gr\"obner basis certification by reductions of clear S-polynomials}
\setcounter{equation}{0}
\label{sect:gb}

In this section,
we first review the basics of S-polynomials in
free-associative algebra \cite{mora1994introduction,mora1985grobner},
then propose a theorem that further reduces the number of S-polynomials
needed for Gr\"obner basis certification.
The following lemma is trivial but very useful.

\bl
\label{lem:free}
In a free associative algebra, for any monomials
$\A_1, \A_2, \B_1, \B_2$
such that $\A_1\A_2=\B_1\B_2$, if the degrees $|\A_1|<|\B_1|$, then
$\B_1=\A_1\B_1',\ \A_2=\B_1'\B_2$,
where $\B_1'$ is a submonomial of degree $>0$; if $|\A_1|=|\B_1|$, then
$\B_1=\A_1,\ \A_2=\B_2$.
\el

Let $G$ be a set of polynomials in free-associative algebra
$\mathbb{K}\langle \mathcal{A}\rangle$.
Fix the degree lexicographic order in the algebra.
For any two elements $f,g\in G$ that are irreducible with respect to each other,
an {\it S-polynomial} they generate is of the form
$\A_1f\B_1-\A_2g\B_2$, where the $\A_i, \B_i$ are monomials satisfying
${\rm T}(\A_1f\B_1)={\rm T}(\A_2g\B_2)$, with operator ``{\rm T}" extracting the leading term of
a polynomials. 

For an ideal ${\cal I} \subseteq \mathbb{K}\langle \mathcal{A}\rangle$,
a subset $G \subseteq {\cal I}$ is a Gr\"obner basis of ${\cal I}$
if and only if $G$ generates ${\cal I}$, and any S-polynomial of $G$
can be reduced to zero by $G$.

If $h$ is an S-polynomial, then for any monomials
$\A$ and $\B$, $\A h\B$ is also an S-polynomial. 
So only those S-polynomials of minimal degree are non-trivial.
By Lemma \ref{lem:free}, after removing redundant common left submonomials and common
right submonomials, one can always get
either $|\A_1|=0$ or $|\B_2|=0$, where if $|\A_1|=0$ then one further gets $|\B_2|=0$, and if
$|\B_2|=0$ then one further gets $|\A_1|=0$. So only the following {\it unilateral S-polynomials}
need to be considered:

\bdf
For any two polynomials $f,g$ in ${\mathbb K}\langle {\cal A}\rangle$
that are irreducible with respect to each other,
if there exist monomials $\L, \R$ with properties $0<|\L|<|f|$ and $0<|\R|<|g|$, such that
${\rm T}(f)\R=\L{\rm T}(g)$, then $(f,\R;g,\L)$
is called an {\it S-quadruplet}, where
${\rm T}(f)\R$ is called the {\it leader} of the S-quadruplet,
$f\R$ is called the {\it left polynomial}, and $\L g$ is called the {\it right polynomial}.
Polynomial $f\R-\L g$ is the corresponding
{\it S-polynomial} of the S-quadruplet.
\edf

Notice that the leading term of the S-polynomial of an S-quadruplet is strictly lower
than the leader, because the leader is the leading term of both $f\R$ and $\L g$, and is
canceled in the subtraction $f\R-\L g$. Also notice that $|f\R-\L g|<|f|+|g|$, so the tail of
monomial ${\rm T}(f)$ overlaps with the head of monomial ${\rm T}(g)$ in any S-quadruplet.

The reason why $0<|\L|<|f|$ is as follows. If $|\L|=0$, then
${\rm T}(f)\,|\, {\rm T}(g)$, contradicting with the irreducibility assumption.
If $|\L|\geq |f|$, then by lemma \ref{lem:free},
$\L={\rm T}(f)\L'$ for some submonomial $\L'$, and $\R=\L'{\rm T}(g)$.
Let $f={\rm T}(f)+\Delta f$, $g={\rm T}(g)+\Delta g$, where
$\Delta f, \Delta g$ stand for the remaining terms
of $f,g$ respectively other than the leading term. Then
the S-polynomial becomes $f\R-\L g=(\Delta f)\L'{\rm T}(g)-{\rm T}(f)\L'(\Delta g)$.
Direct reductions of $(\Delta f)\L'{\rm T}(g)$ by $g$, and
${\rm T}(f)\L'(\Delta g)$ by $f$, yield
$-(\Delta f)\L'(\Delta g)+(\Delta f)\L'(\Delta g)=0$, so the S-polynomial is trivially
reduced to zero by $f,g$, and is not useful in Gr\"obner basis certification.

In the following, we further reduce the number of S-polynomials
in the certification of Gr\"obner basis by proposing a more general technique
than S-polynomial reduction. 


\bdf
Let $G$ be a set of polynomials in ${\mathbb K}\langle {\cal A}\rangle$.
For any S-quadruplet $(f,\R;g,\L)$ where $f,g\in G$,
if after removing both the leading letter and the ending letter
from the leader ${\rm T}(f)\R$, the leftover of the leader 
is irreducible with respect to $G$,
then the S-quadruplet is said to be {\it clear}, and the corresponding
S-polynomial is called a {\it clear S-polynomial}.
\edf

\begin{prop}
\label{prop:k-decomp equiv}
Let $G$ be a set of polynomials in ${\mathbb K}\langle {\cal A}\rangle$,
where every element $g\in G$ is of degree $>1$ and is irreducible with respect to
$G\backslash \{g\}$. Let $t>2$ be an integer.
Then the following two statements are equivalent:

1. For any $f,g\in G$, if there exists an S-quadruplet
$(f,\R; g,\L)$ where $|f \R|\leq t$,
then there exist finitely many monomials $\L_i, \R_i$, and elements $g_i \in G$, such that
$f \R = \sum_i \L_i g_i \R_i$,
where for each $i$, (a) $\L_i g_i \R_i\preceq {\rm T}(f) \R$, (b)
if $|\L_i|=0$, then $g_i \R_i \prec {\rm T}(f) \R$.

2. For any $f,g\in G$, if there exists a clear S-quadruplet
$(f,\R; g,\L)$, where $|f \R| \leq t$,
then there exist finitely many monomials $\L_i, \R_i$, elements $g_i \in G$, such that
$f \R = \sum_i \L_i g_i \R_i$,
where for each $i$, (a) $\L_i g_i \R_i\preceq {\rm T}(f) \R$, (b)
if $|\L_i|=0$, then $g_i \R_i \prec {\rm T}(f) \R$.
\end{prop}

\begin{proof}
Statement 1 trivially implies statement 2. We prove the converse implication.
Assume that $(f,\R;g,\L)$ is an S-quadruplet but not a clear one, namely, if 
${\rm T}(f) \R=\m_1\m_2\cdots\m_u$, where the $\m_i$ are letters, then 
submonomial $\H=\m_2\cdots \m_{u-1}$ is reducible by $G$. So $u\geq 4$ by the degree assumption
on elements of $G$. 
If $\H$ is reducible by more than one element of
$G$, say $\H=\L_1 {\rm T}(f_1) \R_1=\L_2 {\rm T}(f_2) \R_2$ for two different elements
$f_1, f_2\in G$, then $|\L_1|\neq |\L_2|$ by
Lemma \ref{lem:free} and the irreducibility assumption on $G$. 
Because of this, there exist monomials $\L_s, \R_s$ and an element
$f_s\in G$, such that $\H=\L_s {\rm T}(f_s) \R_s$, and $|\L_s|$ is the smallest 
among all the $\L_i$ in triplets $(\L_i, \R_i, f_i)$ satisfying $\H=\L_i {\rm T}(f_i) \R_i$.

Now ${\rm T}(f) \R=\m_1 \L_s {\rm T}(f_s) \R_s\m_u$. Again by Lemma \ref{lem:free}
and the irreducibility property on $G$,
$\R_s\m_u\,|\,\R$, and $|\R|>|\R_s\m_u|>0$. Let $\R = \R'\R_s\m_u$, with $\R'$ being a
submonomial. Then $|\R|>|\R'|>0$, and
${\rm T}(f) \R'=\m_1 \L_s {\rm T}(f_s)$. Set $\L'=\m_1 \L_s$. Then $|\L'|>0$, and
${\rm T}(f) \R'=\L' {\rm T}(f_s)$, and
$|f\R'|<|f\R|\leq t$. Once it is proved that (1) $|\L'|<|f|$,
(2) S-quadruplet $(f, \R'; f_s, \L')$ is clear,
then all the assumptions in statement 2 on the S-quadruplet
are satisfied, so a decomposition of $f \R'$ is obtained, and the required
decomposition of $f\R=f\R'\R_s\m_u$ follows.

Proof of (1).
Since ${\rm T}(f) \R = \m_1 \H \m_u =\L {\rm T}(g)$, where $0<|\L|<|f|$, 
$\L$ can be decomposed into
$\L=\m_1 \L_0$, with $\L_0$ being a submonomial, and $\H\m_u=\L_0 {\rm T}(g)$. 
By the minimality of $|\L_s|$, among 
all the $\L_i$ in triplets $(\L_i, \R_i, f_i)$ satisfying $\H\m_u=\L_i {\rm T}(f_i) \R_i\m_u$,
$|\L_s|$ is the smallest. So $|\L_0|\geq |\L_s|$, and
$|\L'|=1+|\L_s|\leq 1+|\L_0|=|\L|<|f|$ follows.

Proof of (2).
For ${\rm T}(f) \R'=\m_1 \L_s {\rm T}(f_s)=\m_1\m_2\cdots\m_j$, obviously
$j= u-|\R_s|-1$. Since $|f|\geq 2$ and $|\R'|\geq 1$, we have $j=|f|+|\R'|\geq 3$. Let
$\H'=\m_2\cdots\m_{j-1}$, then $|\H'|>0$. If $\H'$
is reducible by some $g'\in G$, say $\H'=\A{\rm T}(g')\B$ for some monomials $\A,\B$,
then 
\be
\H=\L_s {\rm T}(f_s) \R_s
=\m_2\cdots\m_j \R_s
=\H' \m_j \R_s=\A{\rm T}(g')\B\m_j \R_s.
\ee
By the minimality of $|\L_s|$, we get $|\A|\geq |\L_s|$. Then
from $\L_s {\rm T}(f_s)=\A{\rm T}(g')\B\m_j$, we get
${\rm T}(g')\mid {\rm T}(f_s)$,
contradicting
with the irreducibility property on $G$. So $\H'$ is irreducible with respect to $G$, and
$(f, \R'; f_s, \L')$ is clear.
\end{proof}

\bt
\label{thm:clear}
Let $G$ be a set of polynomials in ${\mathbb K}\langle {\cal A}\rangle$, where
every element $g\in G$ is of degree $>1$ and irreducible with respect to
$G\backslash \{g\}$. Then $G$ is a Gr\"obner basis of $\langle G\rangle$ if and only if
all the clear S-polynomials of $G$ can be reduced to zero by $G$.
\et

\begin{proof}
We only need to prove the ``if" part. For any clear S-quadruplet
$(f,\R; g,\L)$ of $G$, 
there exist finitely many monomials $\L_i, \R_i$, elements $g_i \in G$, such that
$f \R = \L g+\sum_i \L_i g_i \R_i$,
where for each $i$, $\L_i g_i \R_i\preceq {\rm T}(f \R - \L g)\prec {\rm T}(f) \R$.
So statement 2 of Proposition \ref{prop:k-decomp equiv} is true.
By this proposition, statement 1 must be true, namely,
for arbitrary S-quadruplet $(f,\R; g,\L)$ of $G$, denote $t=|f\R|\geq 3$, then
there exist finitely many monomials $\L'_i, \R'_i$, nonzero scalars $\lambda_i$,
and elements $g'_i \in G$, such that the
corresponding S-polynomial $f'=f\R-\L g$ satisfies 
\be
f'=\sum_i \lambda_i \L'_i g'_i \R'_i, 
\label{proof:thm1}
\ee
where for each $i$, $|\L'_i g'_i \R'_i|\leq t$. 

Claim 1. For any nonzero polynomial
$f'\in \langle G\rangle$ satisfying (\ref{proof:thm1}) and the constraint
that $|\L'_i g'_i \R'_i|\leq t$ for each $i$, 
there exists an element $g'\in G$ such that ${\rm T}(g') \mid {\rm T}(f')$. 

We prove this claim below. 
On the right side of (\ref{proof:thm1}), every summand is a scaled product of
a triplet $(\L'_i, g'_i, \R'_i)$, and the triplets are pairwise different.
Let $d$ be the maximal degree of these summands. Then $d\leq t$.

Define the following finite set and some subsets of it,
where $\L, \R, \A, \B$ are monic monomials, and $f,g\in G$:
\be\ba{lll}
\mathcal{T} &:=& \{(\L,f,\R)\ \mid \
|\L f \R| \leq d\}; \\

\mathcal{T}_1 &:=& \bigstrut\{(\L,f,\R)\in \mathcal{T}\ \mid\ \exists g,
\A,\B, \hbox{ s.t. }  
\L {\rm T}(f) \R = \A {\rm T}(g) \B,\ |\L| < |\A| < |\L f|\}; \\

\mathcal{T}_2 &:=& \strut\{(\L,f,\R)\in \mathcal{T}\backslash  \mathcal{T}_1\ \mid \ \exists g,
\A,\B, \hbox{ s.t. }  
\L {\rm T}(f) \R = \A {\rm T}(g) \B,\ |\A| \geq |\L f|\}; \\

\mathcal{T}_3 &:=& \strut\{(\L,f,\R)\in \mathcal{T}\backslash  \mathcal{T}_1\ \mid\ \exists g,
\A,\B, \hbox{ s.t. }
  \L {\rm T}(f) \R = \A {\rm T}(g) \B,\ |\A| \leq |\L|\}.
\ea\ee
Obviously $\mathcal{T}$ is a finite set,
$\mathcal{T}=\mathcal{T}_1\cup \mathcal{T}_2\cup \mathcal{T}_3$, and
the sets $\mathcal{T}_i$ do not overlap with each other. All the triplets on the right side
of (\ref{proof:thm1}) belong to $\mathcal{T}$.

First investigate $\mathcal{T}_3$. For any two different triplets
$(\L,f,\R), (\A,g,\B)\in \mathcal{T}_3$,
if $\L {\rm T}(f) \R = \A {\rm T}(g) \B$, then $|\A|=|\L|$ by the definition of $\mathcal{T}_3$, 
which leads to
$\L=\A$, hence one of ${\rm T}(f), {\rm T}(g)$ is a factor of the other,
resulting in $f=g$, and then $(\L,f,\R)=(\A,g,\B)$, contradiction.
So any two different triplets of $\mathcal{T}_3$
have different leading terms, if each triplet is multiplied and then linearly expanded
to form a polynomial called the {\it multiplied triplet}.
In (\ref{proof:thm1}), if the triplets on the right side are all in $\mathcal{T}_3$, then
${\rm T}(f')= \lambda_j \L'_j {\rm T}(g'_j) \R'_j$ for some $j$, 
so ${\rm T}(g'_j)\mid {\rm T}(f')$. Indeed, the idea behind the proof 
is to change every multiplied 
triplet of $\mathcal{T}_1\cup \mathcal{T}_2$ 
into the sum of some scaled multiplied triplets of $\mathcal{T}_3$.

Next investigate $\mathcal{T}_2$. For any $(\L,f,\R)\in \mathcal{T}_2$, there exists a triplet
$(\A,g,\B)$ such that $\L {\rm T}(f) \R = \A {\rm T}(g) \B$ and $\A= \L {\rm T}(f) \L'$,
where $\L'$ is a submonomial.
Obviously $(\A,g,\B)\in \mathcal{T}$.
Then 
\be
\L {\rm T}(f) \R = \A {\rm T}(g) \B
= \L {\rm T}(f) \L' {\rm T}(g) \B, 
\ee
so $\R=\L' {\rm T}(g) \B$.
Let $g=\sum_{j=0}^k g_{[j]}$, where ${\rm T}(g)=g_{[0]}\succ g_{[1]}\succ \cdots \succ
g_{[k]}$ are the terms of $g$ in descending order. Similarly, let
$f=\sum_{i=0}^l f_{[i]}$.
Make the following decomposition of polynomial $\L f \R$:
\be
\L f \R = \sum_{i=0}^l \L f_{[i]} \L' g \B
-\sum_{j=1}^k \L f \L' g_{[j]} \B.
\label{trans:T2}
\ee
On the right side, the summand
$\L f_{[i]} \L' g \B$ has the property that $(\L f_{[i]} \L', g, \B)\in \mathcal{T}$,
$|\L f_{[i]} \L' g \B| \leq |\L f \R|$, and either (a) $|\L f_{[i]} \L'|>|\L|$,
or (b) $|\L f_{[i]} \L'|=|\L|$, and $f_{[i]}$ is a scalar, and 
$\L f_{[i]} \L' {\rm T}(g) \B\prec \L {\rm T}(f) \R$. The summand
$\L f \L' g_{[j]} \B$ where $j>0$, has the property that $(\L, f, \L' g_{[j]} \B)\in \mathcal{T}$,
$|\L f \L' g_{[j]} \B|\leq |\L f \R|$, and (c)
$\L {\rm T}(f) \L' g_{[j]} \B \prec \L {\rm T}(f) \L' {\rm T}(g) \B=\L {\rm T}(f) \R$.

Finally, 
we show that every multiplied triplet of $\mathcal{T}_1\cup \mathcal{T}_2$
can be decomposed iteratively into a sum of scaled multiplied triplets in
$\mathcal{T}_3$. Initially the output set $\mathcal O$ collecting the summands is set to be empty.
The following procedure starts from a summand on the right side of (\ref{proof:thm1})
that has the maximal order, say $\L' g' \R'$.

Step 1. If $(\L', g', \R')\in \mathcal{T}_3$, then directly
move $\L' g' \R'$ to the set $\mathcal O$.

Step 2. If $(\L', g', \R')\in \mathcal{T}_1$, then
statement 1 of Proposition \ref{prop:k-decomp equiv} provides a decomposition of
$ \L' g' \R'$ into a sum of scaled multiplied triplets
$\L_j g_j \R_j$, where each $(\L_j, g_j, \R_j)\in \mathcal{T}$, with the {\it monotonous property}
that $|\L_j g_j \R_j| \leq |\L' g' \R'|$,
and either $|\L_j|> |\L'|$, or $|\L_j|=|\L'|$ and $\L_j {\rm T}(g_j) \R_j \prec \L' {\rm T}(g') \R'$.
For each summand in the decomposition, do Steps 1 to 3.

Step 3. If $(\L', g', \R')\in \mathcal{T}_2$, then (\ref{trans:T2}) gives a
decomposition of $\L' g' \R'$ into a sum of scaled multiplied
triplets, where each triplet is in $\mathcal{T}$, and has the same
monotonous property as in Step 2, according to
the investigation of $\mathcal{T}_2$.
For each summand in the decomposition, do Steps 1 to 3.

By the monotonous property of the decompositions in Steps 1 and 2, the
iteration of decompositions finishes in finite steps. The outcome is a sum of
scaled multiplied triplets in $\mathcal{T}_3$. By the property that in $\mathcal{T}_3$, no two
scaled multiplied triplets have the same leading monomial, 
we get that $f'\neq 0$ must have its leading term
agree with that of a scaled multiplied triplet in $\mathcal{T}_3$, and in particular,
have its leading term agree with that of an element $g'\in G$.
This proves Claim 1.

Claim 2. S-polynomial $f'=f\R-\L g$ can be reduced to zero by $G$.

We prove this claim below. By Claim 1, 
let ${\rm T}(f')=\lambda' \L' {\rm T}(g') \R'$ for some
nonzero scalar $\lambda'$ and monic monomials $\L', \R'$. Then by (\ref{proof:thm1}),
\be
f''=f'-\lambda'\L' g'\R'=-\lambda'\L' g'\R'+\sum_i \lambda_i \L'_i g'_i \R'_i,
\ee
where on the right side, $\L' g' \R'\prec {\rm T}(f) \R$. 
Obviously ${\rm T}(f'')\prec {\rm T}(f')$.
Again by Claim 1, as long as $f''\neq 0$,
there exists an element $g''\in G$ such that ${\rm T}(g'') \mid {\rm T}(f'')$. 
Do the reduction of $f''$ by $g''$ to get $f'''$.
As long as the reduction result is nonzero, one can continue to apply Claim 1 to make 
reduction with respect to $G$. In the end, the reduction result has to be zero.
\end{proof}

\section{Proof of the Main Theorem}
\setcounter{equation}{0}
\label{sect:proof}

With all the preparations made in the previous sections, we are now ready to 
prove the Main Theorem \ref{mainthm} of this paper. The theorem
can be proved by verifying that all clear S-polynomials of $\bf BG$ can be reduced to zero by
${\bf BG}^{ext}$. First we make a classification of the clear S-polynomials.

The leading terms of the elements of {\bf BG} are depicted in Fig. \ref{fig:BG}.
They can be classified into five types:
\bi
\item {\bf Type I.} As shown in Fig. \ref{fig:BG}(1), this type of leading term 
is of the form $\m_2\m_1$, where letter $\m_2\succ \m_1$. The leading terms are all
from {\bf BG2}(b,c), and the converse is also true.

\item {\bf Type V.} As shown in Fig. \ref{fig:BG}(2.1)-(2.2), this type of leading term 
is of the form $\m_3\m_1\m_2$, where in the letters, $\m_3\succeq \m_2\succ \m_1$, and
$\m_3=\m_2$ if and only if they are both $\k$.
The leading terms are all
from {\bf BG3}, and include
$\k\1\k$ of {\bf BG3}(a), and
$\3\1\bar{\2}, \3\1\2, \2\1\bar{\1}, \e\1\2, \e\1\bar{\1}, 
\e\1\e'$ of {\bf BG3}(b), where $\e\succ \e'$ and both are in $\cal E$.

\item {\bf Type U.} As shown in Fig. \ref{fig:BG}(3.1)-(3.3), this type of leading term 
is of the form $\m_3\m_2\m_1$, where $\m_3\succeq \m_2\succeq \m_1$ and
at most one equality holds. The leading terms are all
from {\bf BG3}, and include $\2\2\1$ of {\bf BG3}(c), and
$\3\2\1, 
\2\1\1, 
\e\2\1, 
\e\1\1$ of {\bf BG3}(b).

\item {\bf Type N4.} As shown in Fig. 
Fig. \ref{fig:BG}(4.1)-(4.2), this type of leading term 
is of the form $\m_3\m_2\m_4\m_1$, where 
$\m_4\succeq \m_3\succ \m_2\succ \m_1$. The leading terms are all
from {\bf BG4}, and the converse is also true.

\item {\bf Type-N.} As shown in Fig. \ref{fig:BG}(5),
where in the letters, this type of leading term 
is of the form $\m_3\m_2\A_4\m_5\m_1$, where letters
$\m_5\succ \m_3\succ \m_2\succ \m_1$, and
$\A_4$ is a non-decreasing monomial satisfying $\m_3\preceq_M \A_4\prec_M \m_5$.
The leading terms are all
from {\bf BGm} $(m>4)$, and the converse is also true.
\ei

\begin{figure}[htbp]
\centering\includegraphics[width=4.5in]{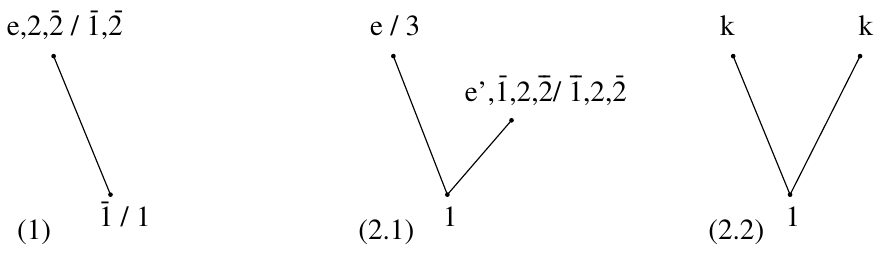}
\centering\includegraphics[width=4.7in]{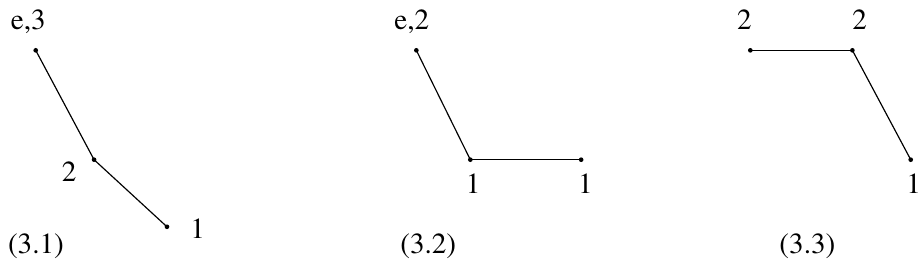}
\centering\includegraphics[width=5.3in]{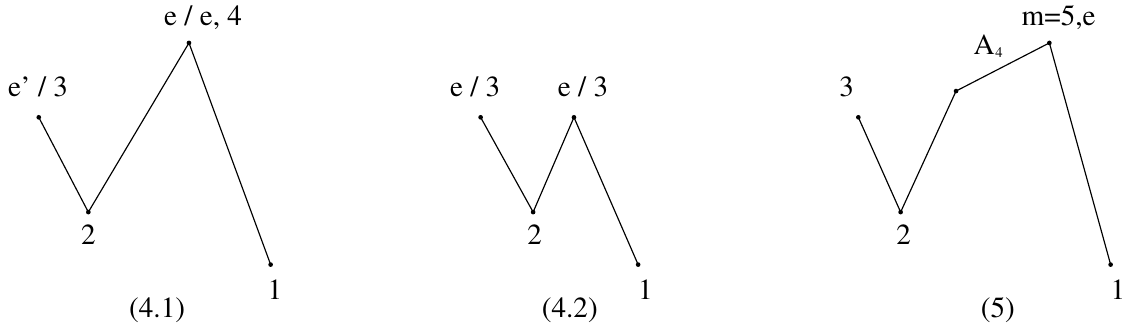}
\caption{Leading terms of {\bf BG}. 
Here $\e\succ \e'$ are letters in $\cal E$, $\1\prec \2\prec \ldots \prec \5$ are 
letters in $\cal Q$,
and $\A_4$ is a non-decreasing monomial satisfying $\5\succ_M \A_4\succeq \3$.
The notation $\e,\2$ at a vertex
stands for two possible evaluations at the vertex: $\e$ or $\2$.
The notation $\e / \3$ stands for two different combinations: those letters
on the left the slash symbol at each vertex form the first group, and those
on the right of the slash form the second group.
}
\label{fig:BG}
\end{figure}

The clear S-polynomials generated by elements of types {\bf I}, {\bf U}, {\bf V},
$\bf N4$ in {\bf BG} have degree ranging from 3 to 5. They can be easily reduced to zero by
${\bf BG}^{ext}$ with the help of a computer.
In the following, we consider only the
clear S-polynomials generated by an element of type {\bf N} and another element
of arbitrary type in {\bf BG}.

\begin{figure}[htbp]
\hskip -.3cm
\centering\includegraphics[width=2.67in]{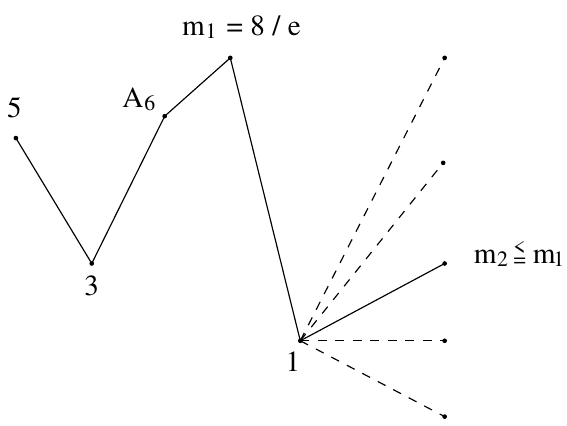}{(a)}
\includegraphics[width=2.35in]{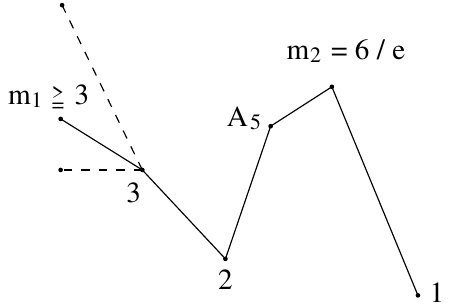}{(b)} \\
\vskip .2cm
\centering\includegraphics[width=2.50in]{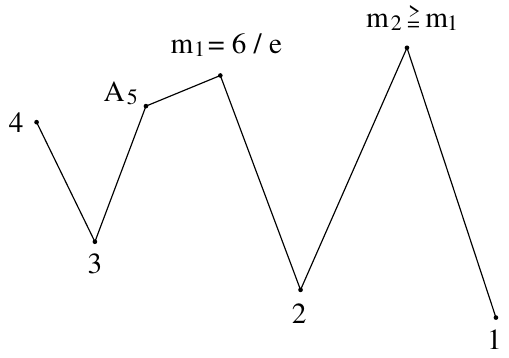}{(c)}
\includegraphics[width=2.54in]{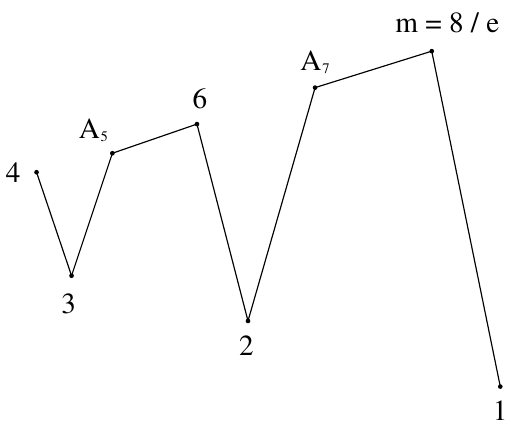}{(d)}\\
\vskip .2cm
\centering\includegraphics[width=2.85in]{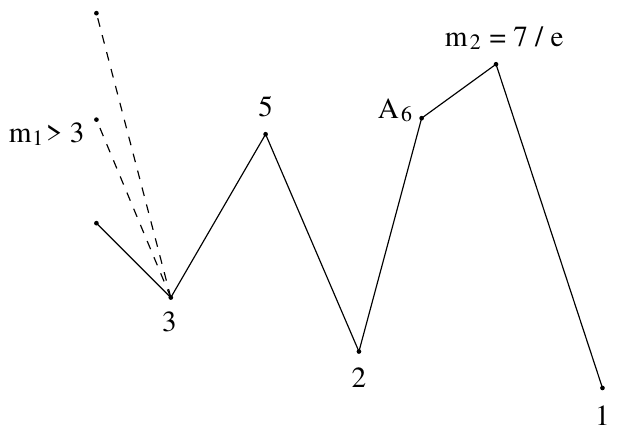}{(e)}
\caption{\rm Leaders of clear S-quadruplets with a generator of type $\N$.
Here $\e\in \cal Q$, 
$\1\prec \2\prec \cdots \prec \8$ are letters in
${\cal Q}$, and for $l=5,6,7$, 
$\A_l$ is a non-decreasing monomial satisfying 
${\bf (l-1)} \preceq_M \A_l\prec_M \l$.
}
\label{fig:Spoly}
\end{figure}

The leaders of clear S-quadruplets $(f,\R; g,\L)$ where one of
$f,g\in \bf BG$ has degree $\geq 5$, are depicted in Fig. \ref{fig:Spoly}.
They can be classified into five types:
\bi
\item {\bf Type N+I/U/V.} As shown in Figure \ref{fig:Spoly}(a), this type of 
clear S-polynomial is generated by $f$ of type $\bf N$ and $g$ of one of the types
$\I, \U, \V$. The leader is of the form $\5\3\A_6\m_1\1\m_2$, where 
$\m_1\succ_M \A_6\succeq_M \5\succ \3\succ \1$, 
$\m_2$ can be any letter $\preceq \m_1$,
and the equality holds if and only if $\m_1=\m_2=\k$.

\item {\bf Type I/U+N.} As shown in Figure \ref{fig:Spoly}(b), this type of 
clear S-polynomial is generated by $f$ of one of the types $\I, \U$, and 
$g$ of type $\N$. The leader is of the form $\m_1\3\2\A_5\m_2\1$, where 
$\m_2\succ_M \A_5\succeq_M \3\succ \2\succ \1$, and
$\m_1$ can be any letter $\succeq \3$.

\item {\bf Type N+N4.} As shown in Figure \ref{fig:Spoly}(c), this type of 
clear S-polynomial is generated by $f$ of type $\N$ and 
$g$ of type $\bf N4$. The leader is of the form 
$\4\3\A_5\m_1\2\m_2\1$, where 
$\m_2\succeq \m_1\succ_M \A_5\succeq_M \4\succ \3\succ \2\succ \1$.

\item {\bf Type N+N.} As shown in Figure \ref{fig:Spoly}(d), this type of 
clear S-polynomial is generated by $f,g$, both of which are of type $\N$. 
The leader is of the form $\4\3\A_5\6\2\A_7\m\1$, where 
$\m\succ_M \A_7\succeq_M \6\succ_M \A_5\succeq_M \4\succ \3\succ \2\succ \1$.

\item {\bf Type V/N4+N.} As shown in Figure \ref{fig:Spoly}(e), this type of 
clear S-polynomial is generated by $f$ of one of the types $\V, \N\4$, and 
$g$ of type $\N$. The leader is of the form $\m_1\3\5\2\A_6\m_2\1$, where 
$\m_2\succ_M \A_6\succeq_M \5\succ \3\succ \2\succ \1$, 
and letter $\m_1\succ \3$.
\ei

\bt
The clear S-polynomials generated by $f,g\in BG$, where
one of $f,g$ is in {\bf BGm} ($m\geq 5$), are all reduced to zero by 
{\bf BG}.
\label{ext:thm2}
\et

In the following, we prove the theorem by establishing a sequence of propositions
on the reduction of clear S-polynomials whose leaders are shown in Figure \ref{fig:Spoly}.

The first case is Figure \ref{fig:Spoly}(a), where $\m_2\prec \m_1$.
In this case, it may be that $\m_1$ is in $\cal E$, or both
$\m_1, \m_2$ are in $\cal E$. Since the conclusion is true no matter if
these letters are in $\cal E$ or not, in making reduction to the corresponding clear S-polynomial,
all relations involving basis letters should be avoided. 

Although there are many different cases depending on the different
order relations between $\m_2$ and the letters in \{$\1, \bar{\1},
\3, \bar{\3}, \5, \bar{\5}\}$, by choosing the order of $\m_1$ to be as high as possible,
namely, $\bar{\5}\prec \m_2\prec \m_1$, the reduction procedure
of the corresponding S-polynomial will be shown to be valid for all other cases. 
This is a typical example of substitutional reduction.

By introducing new variables $\1$ to $\6$:
\be
(\1, \3, \5, \A_6, \m_2, \m_1) \to (\1, \2, \3, \4, \5, \6),\ \ 
\hbox{ with } \1\prec \2\prec \3\prec \4\prec \5\prec \6,
\ee
the corresponding S-polynomial becomes 
\be\ba{lll}
h_1 &=& (\3[\2\4\6\1]-[\2\4\6\1]\3)\5-\3\2\4([\6\1]\5-\5[\6\1]) \\

&=& \3\bar{\1}\bar{\6}\bar{\4}\bar{\2}\5-\3\2\4\bar{\1}\bar{\6}\5
+\3\2\4\5\6\1+\3\2\4\5\bar{\1}\bar{\6}-[\2\4\6\1]\3\5.
\ea
\label{expr:h1}
\ee
This polynomial can be easily reduced to zero by ${\bf BG}[\1,\ldots,\6]$ with computer help.
Indeed, this is exactly the S-polynomial generated by $f=\3[\2\4\6\1]-[\2\4\6\1]\3$ of
{\bf BG}\5, and $g=[\6\1]\5-\5[\6\1]$ of {\bf BG}\3. 
However, since $\4$ now substitutes monomial $\A_6$, and the order
between $\m_2$ and the other letters except for $\m_1$ is arbitrary, the substitutional
reduction procedure of $h_1$ needs to be checked for correctness in old letters.

In checking the reduction procedure, it needs to be warned that 
by now we have only proved that for monomials $\R, \S$ in old letters, 
$\R[\S]-[\S]\R$ can be reduced to zero by 
{\bf BG} only in three special cases: $|\R|=1$, $|\S|=1$, $\S=\q\bq$ respectively, cf. 
Theorem \ref{ext:thm} and Proposition \ref{prop:11}. Because of this,
in making substitutional reduction to a polynomial with leading term $\T$, 
if a new letter substitutes a submonomial, e.g.,
$\4$ replaces $\A_6$ in $h_1$ of (\ref{expr:h1}), 
then a reduction of the form $\4[\S]=[\S]\4$ is allowed only
in three special cases: 
(1) $|\S|=1$, (2) $\S=\q\bq$ for some $\q\in {\cal Q}\cup \overline{\cal Q}$, 
(3) $\M\4[\S]=\M[\S]\4$, 
where $\M$ is a non-empty subsequence $\prec$ first $|\M|$ letters of $\T$.

\bp
\label{prop:h1}
$h_1$ in old letters can be reduced to zero by {\bf BG}.
\ep

\begin{proof}
In the new letters, $\4$ represents a submonomial $\succeq_M \3$. 
In old letters, $h_1$ has its leading term between two terms:
$-\3\2\4\bar{\1}\bar{\6}\5$ and $\3\2\4\5\6\1$, depending on the actual order relation
between $\5=\m_2$ and $\bar{\1}$. In both cases, the leading subsequence of the leading term
of $h_1$ is $\3\2\4$.
\be
\ba{lcl}
h_1&=& \3\bar{\1}\bar{\6}\bar{\4}\bar{\2}\5-\underbrace{\3\2\4\bar{\1}}\bar{\6}\5
+\underbrace{\3\2\4\5\6\1}+\underbrace{\3\2\4\5\bar{\1}}\bar{\6}
+\hbox{ lower} \\

&\overset{f_{14}}{=}&
\underbrace{\3\bar{\1}}\bar{\6}\bar{\4}\bar{\2}\5+\3\1\bar{\4}\bar{\2}\bar{\6}\5
-\underbrace{\3\bar{\1}}\bar{\6}\bar{\5}\bar{\4}\bar{\2}
-\3\1\bar{\5}\bar{\4}\bar{\2}\bar{\6}
+\hbox{ lower} \\

&\overset{\overline{\bf BG}\2}{=}& \3\1(
-\underbrace{\bar{\6}\bar{\4}\bar{\2}\5}
+\bar{\4}\bar{\2}\bar{\6}\5
+\underbrace{\bar{\6}\bar{\5}\bar{\4}\bar{\2}}
-\underbrace{\bar{\5}\bar{\4}\bar{\2}}\bar{\6}
)+\hbox{ lower} \\

&\overset{\3\1{\cal I}}{=}& \3\1(
\4\underbrace{\6\bar{\2}\5}
+\underbrace{\bar{\4}\bar{\2}}\bar{\6}\5
-\4\5\6\bar{\2}
+\4\5\bar{\2}\bar{\6}
)
-\underbrace{\3\1\bar{\2}}[\4\6]\5
+\underbrace{\3\1\bar{\2}}[\4\5\6]
-\underbrace{\3\1\bar{\2}}[\4\5]\bar{\6}
+\hbox{ lower} \\

&\overset{\3\1{\cal I}, \overline{\bf BG}\3}{=}& \3\1(
-\underbrace{[\4]\2}\bar{\6}\5+\underbrace{\4\5[\2]}\bar{\6}
)
+\underbrace{\3\1[\2]}\bar{\4}\bar{\6}\5
+\hbox{ lower} \\

&\overset{\3\1{\cal I},\overline{\bf BG}\3}{=}& 
-\underbrace{\3\1\2}[\4]\bar{\6}\5+\underbrace{\3\1[\2]}\4\5\bar{\6}
+\hbox{ lower} \\

&\overset{\overline{\bf BG}\3}{=}& \hbox{ lower}.
\ea
\label{reduction:h_1}
\ee
By computer,
the ``lower" terms are reduced to zero by {\bf BG}\2 to {\bf BG}\6 in new letters
in 995 steps.

We check the influence of the varying order of $\5$ and the submonomial represented by $\4$.
The order of $\5$ varies from below $\1$ to above $\4$, but still $\5\prec \6$.
\bi
\item In Step 1, $f_{14}$ is resorted 3 times, changing the leading subsequence of $h_1$
from
$\3\2\4$ to $\3\bar{\1}$. In each $f_{14}$-instance, letters $\4,\5$ are sandwiched between $\2$ and
$\1$ (or $\bar{\1}$), so they have no influence on the reduction.

\item Step 2 only replaces the leading subsequence to $\3\1$, using 
$\3\bar{\1}=\3[\1]-\3\1=[\1]\3-\3\1$. 

\item The remaining steps are either $\3\1$-led reduction, or based on
$\overline{\bf BG}\3$ relations $[\3\1]\2-\2[\3\1]$ and
$[\3\1]\bar{\2}-\bar{\2}[\3\1]$, so $\4,\5$ have no influence on them.
\ei
Summing up, the reductions in (\ref{reduction:h_1}) are valid for all allowed orders of
$\5$ and for submonomial $\4$. The final stage is a bottom-letter controlled reduction, which is
independent of the change of order of $\5$ and the role of submonomial of $\4$.
\end{proof}

The second case is Figure \ref{fig:Spoly}(a) with $\m_2=\m_1=\k$.
By introducing new variables $\1$ to $\4$:
\be
(\1, \3, \5, \A_6) \to (\1, \2, \3, \4), \hbox{ with }\ \1\prec \2\prec \3\prec \4,
\ee
the corresponding S-polynomial becomes 
\be\ba{lcl}
h_2 &=& (\3[\2\4\k\1]-[\2\4\k\1]\3)\k-\3\2\4(\k\1\k+\j\1\j+\i\1\i+2\times\bar{\1}+\1) \\

&=& -\3\bar{\1}\underbrace{\k\bar{\4}\bar{\2}}\k
-\underbrace{\3\2\4\j\1}\j
-\underbrace{\3\2\4\i\1}\i
-2\times \underbrace{\3\2\4\bar{\1}}
-\3\2\4\1+\hbox{ lower}.
\ea
\ee

\bp
$h_2$ in old letters can be reduced to zero by {\bf BG}.
\ep

\begin{proof} The only specialty in the new letters is 
that $\4$ represents submonomial $\A_6$. In $h_2$, the leading term is
$\3\2\4\j\1\j$.
\be\ba{lcl}
h_2 &\overset{\3\bar{\1}{\cal I},
f_{14}, {\rm Prop.}\,\ref{prop:11}}{=} & 
-\underbrace{\3\bar{\1}\bar{\2}}[\k\bar{\4}]\k
-\3\bar{\1}\4\underbrace{\k\bar{\2}\k}
-\3\bar{\1}\underbrace{\j\bar{\4}\bar{\2}}\j
-\3\bar{\1}\underbrace{\i\bar{\4}\bar{\2}}\i
+\underbrace{\3\2\4\1}+\hbox{ lower}\\

&\overset{\overline{\bf BG}\3,\3\bar{\1}{\cal I}, f_{14}}{=} & 
\3\bar{\1}\4(\j\bar{\2}\j+\i\bar{\2}\i)+2\times 
\underbrace{\3\bar{\1}\4\2}+\3\bar{\1}\4\bar{\2}
-\underbrace{\3\bar{\1}\bar{\2}}[\j\bar{\4}]\j
-\3\bar{\1}\4\j\bar{\2}\j
\\

&& \hfill
-\underbrace{\3\bar{\1}\bar{\2}}[\i\bar{\4}]\i
-\3\bar{\1}\4\i\bar{\2}\i
-\3\bar{\1}\bar{\4}\bar{\2}
+\hbox{ lower}\\

&\overset{\3\bar{\1}{\cal I}, {\rm Prop.}\,\ref{prop:11},\overline{\bf BG}\3}{=} & 
-\3\bar{\1}\underbrace{[\4]\bar{\2}}
+\hbox{ lower}\\

&\overset{\3\bar{\1}{\cal I}}{=} & 
-\underbrace{\3\bar{\1}\bar{\2}}[\4]
+\hbox{ lower}\\

&\overset{\overline{\bf BG}\3}{=} & 
\hbox{ lower}.
\ea
\ee
The lower terms are reduced to zero by {\bf BG2} to {\bf BG6} in new letters
in 99 steps.

Influence of the role of submonomial of $\4$:
\bi
\item In Step 1, the $\3\bar{\1}$-led reduction 
$\3\bar{\1}[\k\bar{\4}]\bar{\2}=\3\bar{\1}\bar{\2}[\k\bar{\4}]$ is irrelevant to
$\4$ as a submonomial. In the two
$f_{14}$-instances, $\4$ is sandwiches between $\2,\1$. The relation 
$\3\2\4\bar{\1}=\3\2\4[\1]-\3\2\4\1=[\1]\3\2\4-\3\2\4\1$ is by Proposition \ref{prop:11}.

\item In Step 2, 
$[\3\bar{\1}]\bar{\2}=\bar{\2}[\3\bar{\1}]$ is a $\overline{\bf BG}\3$-relation, 
$\3[\2\4\1]=[\2\4\1]\3$ is an $f_{14}$-relation, 
and all others are $\3\bar{\1}$-led reductions.

\item In Step 3, $\3\bar{\1}\4\2=\3\bar{\1}\4[\2]-\3\bar{\1}\4\bar{\2}=[\2]\3\bar{\1}\4
-\3\bar{\1}\4\bar{\2}$ is a combination of $\3\bar{\1}$-led reduction and application of
Proposition \ref{prop:11}.
 
\item In the next to the last step, 
$\3\bar{\1}[\4]\bar{\2}=\3\bar{\1}\bar{\2}[\4]$ is a $\3\bar{\1}$-led reduction.
\ei
Summing up, the submonomial role of $\4$
has no influence on the substitutional reduction of $h_2$. 
\end{proof}

The third case is Figure \ref{fig:Spoly}(b), where $\m_1\succ_M \3\bar{\3}$.
The order relation between $\m_1, \m_2$ is arbitrary.
In substitutional reduction, 
$\A_5$ and $\m_2$ can be concatenated to a longer sequence.
Introduce new letters:
\be
(\1,\2,\3,\m_1,\A_5\m_2)\to (\1,\2,\3,\4,\5), \hbox{ with }\
\1\prec\2\prec\3\prec\4\prec\5.
\ee
The S-polynomial is
\be\ba{lcl}
h_3 &=& ([\4\3]\2-\2[\4\3])\5\1-\4(\3[\2\5\1]-[\2\5\1]\3)\\

&=& \bar{\3}\bar{\4}\2\5\1-\2[\4\3]\5\1
-\underbrace{\4\3\bar{\1}}\bar{\5}\bar{\2}+\underbrace{\4[\2\5\1]}\3.
\ea
\ee
After the reduction of $h_3$, it must be checked 
whether the varying order between $\4$ and $\5$, and the
submonomial role of $\5$, have any influence on the substitutional reduction.

\bp
$h_3$ in old letters can be reduced to zero by {\bf BG}.
\ep

\begin{proof} In $h_3$, the leading term is always $\4\3\bar{\1}\bar{\5}\bar{\2}$, 
irrelevant to the varying order relation between $\4$ and $\5$.
\be
\label{reduction:h_2}
h_3\ \overset{\overline{\bf BG}\3, f_{14}}{=} \
\bar{\3}\bar{\4}\2\5\1
+\bar{\3}\bar{\4}\bar{\1}\bar{\5}\bar{\2}
+\hbox{ lower}\
\overset{f_{14}}{=}\
 [\2\5\1]\bar{\3}\bar{\4}+\hbox{ lower}.
\ee
The lower terms are reduced to zero by {\bf BG2} to {\bf BG5} in new letters
in 54 steps.

We check the influence of the varying order between $\4$ and $\5$, and the submonomial
role of $\5$.
\bi
\item In Step 1, the $f_{14}$-relations are 
$\bar{\4}[\2\5\1]=[\2\5\1]\bar{\4}$, \ \
$\bar{\3}[\2\5\1]=[\2\5\1]\bar{\3}$, \ \
$\4[\2\5\1]=[\2\5\1]\4$, where $\5$ is sandwiched between $\2,\1$ in each bracket,
and $\4\bar{\4}\succ_M \3\bar{\3}\succ_M \2$ is always true. The $\overline{\bf BG}\3$ 
relation $[\4\3]\bar{\1}=\bar{\1}[\4\3]$ is always valid.

\item In Step 2, the $f_{14}$-relations are still
$\bar{\4}[\2\5\1]=[\2\5\1]\bar{\4}$,\ \
$\bar{\3}[\2\5\1]=[\2\5\1]\bar{\3}$.
\ei
Summing up, the varying order between $\4$ and $\5$, and the submonomial
role of $\5$,
have no influence on the substitutional reduction.
\end{proof}

The fourth case is still Figure \ref{fig:Spoly}(b), but $\m_1=\bar{\3}$.
Introduce new letters:
\be
(\1,\2,\3,\A_5\m_2)\to (\1,\2,\3,\4), \hbox{ with }\
\1\prec\2\prec\3\prec\4.
\label{h45:letter}
\ee
The S-polynomial is
\be\ba{lcl}
h_4 &=& (\bar{\3}\3-\3\bar{\3})\2\4\1-\bar{\3}(\3[\2\4\1]-[\2\4\1]\3)\\

&=& -\3\bar{\3}\2\4\1-\underbrace{\bar{\3}\3}\bar{\1}\bar{\4}\bar{\2}
+\underbrace{\bar{\3}[\2\4\1]}\3.
\ea
\ee

\bp
$h_4$ in old letters can be reduced to zero by {\bf BG}.
\ep

\begin{proof} The leading term in $h_4$ is $-\bar{\3}\3\bar{\1}\bar{\4}\bar{\2}$.\
\be
h_4 \overset{{\bf BG2}, f_{14}}{=}  
-\underbrace{\3\bar{\3}[\2\4\1]}+[\2\4\1]\underbrace{\bar{\3}\3} 
\overset{{\rm Prop.} \, \ref{prop:11},\ {\bf BG2}}{=}  0.
\ee

Influence analysis of the submonomial role of $\4$:
\bi
\item In Step 1, the $f_{14}$-relation is $\bar{\3}[\2\4\1]=[\2\4\1]\bar{\3}$, where
$\4$ is sandwiched between $\2,\1$ in the bracket. The {\bf BG2}-relation is 
$\bar{\3}\3=\3\bar{\3}$.

\item In Step 2, $\3\bar{\3}[\2\4\1]=[\2\4\1]\3\bar{\3}$ by Proposition \ref{prop:11}. 
Alternatively, it can be taken as a combination of two $f_{14}$-relations of the same form
as in Step 1.
\ei
Summing up, the submonomial role of $\4$ has no influence on the substitutional reduction.
\end{proof}

The fifth case is once again Figure \ref{fig:Spoly}(b), but 
$\m_1=\3$. Introduce new letters as in (\ref{h45:letter}). 
The S-polynomial is
\be\ba{lcl}
h_5 &=& (\3[\3\2]-[\3\2]\3)\4\1-\3(\3[\2\4\1]-[\2\4\1]\3) \\

&=& \underbrace{\3\bar{\2}\bar{\3}\4\1}
-\underbrace{\3\2\3\4\1}-\bar{\2}\bar{\3}\3\4\1
-\underbrace{\3\3\bar{\1}}\bar{\4}\bar{\2}+\underbrace{\3[\2\4\1]}\3.
\ea
\ee

\bp
$h_5$ in old letters can be reduced to zero by {\bf BG}.
\ep

\begin{proof} The leading term of $h_5$ is $-\3\3\bar{\1}\bar{\4}\bar{\2}$.
\be\ba{lcl}
h_5 &\overset{f_{14}, \overline{\bf BG}\3}{=} & 
-\3\bar{\1}\underbrace{\bar{\4}\3\2}
+\3\bar{\1}\underbrace{\bar{\4}\bar{\3}\bar{\2}}
+\underbrace{\3\1}\bar{\3}\bar{\4}\bar{\2}
-\3\bar{\1}\3\bar{\4}\bar{\2}
+\hbox{ lower}\\

&\overset{
\3\bar{\1}{\cal I}, (\3\bar{\1}):\overline{\bf BG}\2}{=} & 
\3\bar{\1}(
\bar{\3}\4\underbrace{\2}
-\3\4\bar{\2}
-\bar{\3}\bar{\4}\bar{\2}
-\3\bar{\4}\bar{\2})
-\underbrace{\3\bar{\1}\2}[\bar{\4}\3]
+\underbrace{\3\bar{\1}\bar{\2}}[\3\4]
+\hbox{ lower}\\

&\overset{\3\bar{\1}{\cal I}, {\rm Prop.}\,\ref{prop:11}, \overline{\bf BG}\3}{=} & 
-\3\bar{\1}[\3][\4]\bar{\2}
+\hbox{ lower}\\

&\overset{\3\bar{\1}{\cal I}}{=} & 
-\underbrace{\3\bar{\1}\bar{\2}}[\3][\4]
+\hbox{ lower}\\

&\overset{\overline{\bf BG}\3}{=} & \hbox{ lower}.
\ea
\ee
The lower terms are reduced to zero by {\bf BG2} to {\bf BG5} in new letters
in 99 steps.

Influence analysis of the submonomial role of $\4$:
\bi
\item In Step 1, the $f_{14}$-relations include 
$\3[\bar{\2}\bar{\3}\4\1]=[\bar{\2}\bar{\3}\4\1]\3$,\ \
$\3[\2\3\4\1]=[\2\3\4\1]\3$,\ \
$\3[\2\4\1]=[\2\4\1]\3$, where $\4$ is sandwiched between $\2,\1$ in each bracket. 
The $\overline{\bf BG}\3$-relation is 
$\3[\3\bar{\1}]=[\3\bar{\1}]\3$.

\item In Step 2, $\3\1=\3[\1]-\3\bar{\1}=[\1]\3-\3\bar{\1}$ is used, which is a
$\3\bar{\1}$-dominated reduction of $h_5$. 
The rest are $\3\bar{\1}$-led reductions. 

\item In Step 3, first for
$\3\bar{\1}\bar{\3}\4\2=\3\bar{\1}\bar{\3}\4[\2]-\3\bar{\1}\bar{\3}\4\bar{\2}$,
the $\3\bar{\1}$-led reduction $\3\bar{\1}\bar{\3}\4[\2]=\3\bar{\1}[\2]\bar{\3}\4$ 
is made, then $\3\bar{\1}[\2]=[\2]\3\bar{\1}$ by Proposition \ref{prop:11}.

\item Step 4 is a $\3\bar{\1}$-led reduction. 
\ei
Summing up, the submonomial role of $\4$ has no influence on the substitutional reduction.
\end{proof}

The sixth case is Figure \ref{fig:Spoly}(c). 
Introduce new letters:
\be
(\1,\2,\3,\4, \A_5\m_1, \m_2)\to (\1,\2,\3,\4, \5, \6),
\hbox{ with } \ \1\prec \2 \prec \3\prec \4\prec \5\prec \6.
\ee
The S-polynomial is 
\be\ba{lcl}
h_6 &=& (\4[\3\5\2]-[\3\5\2]\4)\6\1-\4\3(\5[\2\6\1]-[\2\6\1]\5) \\

&=& \underbrace{\4\bar{\2}\bar{\5}\bar{\3}\6\1}
-[\3\5\2]\4\6\1
-\underbrace{\4\3\5\bar{\1}}\bar{\6}\bar{\2}
+\underbrace{\4\3[\2\6\1]}\5.
\ea
\ee
Similarly, in Figure \ref{fig:Spoly}(d), we can  
concatenate $\A_5, \6$ to form new letter $\5$, and
concatenate $\A_7, \m$ to form new letter $\6$.
The corresponding S-polynomial is also $h_6$.

\bp
The S-polynomials corresponding to Figure \ref{fig:Spoly}(c), \ref{fig:Spoly}(d) respectively
can be reduced to zero by {\bf BG}.
\label{prop:h6}
\ep

\begin{proof} In $h_6$, the leading term is $-\4\3\5\bar{\1}\bar{\6}\bar{\2}$.
In the new letters, $\5$ substitutes submonomial $\A_5\m_1$ in Figure \ref{fig:Spoly}(c),
or $\A_5\6$ in Figure \ref{fig:Spoly}(d); 
$\6$ substitutes submonomial $\A_7\m$ in Figure \ref{fig:Spoly}(d).
\be\ba{lcl}
h_6 &\overset{f_{14}}{=} &
-\underbrace{\4\bar{\1}}\underbrace{\bar{\6}\3\5\2}
-\3\5\2\4\6\1
+\4\1\bar{\5}\bar{\3}\bar{\6}\underbrace{\bar{\2}}
-\3\5\bar{\1}\4\bar{\6}\bar{\2}
+\hbox{ lower} \\

&\overset{\overline{\bf BG}\2, \4\1{\cal I},
{\rm Prop.}\, \ref{prop:11}}{=} &
\underbrace{\4\1\2}[\bar{\6}\3\5]-\4\1\bar{\5}\bar{\3}\underbrace{[\6]\2}
-\underbrace{\3\5[\2\4\6\1]}
+\3\5\bar{\1}\underbrace{[\bar{\6}\bar{\4}]\bar{\2}}
-\3\5\bar{\1}\4\underbrace{[\6]\bar{\2}}
+\hbox{ lower} \\

&\overset{\overline{\bf BG}\3, \4\1{\cal I}, \3{\cal I}, f_{14}}{=} & 
-\underbrace{[\4\1\bar{\5}\bar{\3}]\2}[\6]
+\underbrace{\3\5\bar{\1}\bar{\4}[\2]}[\6]
+\underbrace{\3\5\bar{\1}\bar{\2}}[\bar{\6}\bar{\4}]
+\underbrace{\3\5\bar{\1}[\4]\bar{\2}}[\6]
+\hbox{ lower} \\

&\overset{f_{12}, \,{\rm Prop.}\, \ref{prop:11}, f_{13}, \3{\cal I}}{=} &
\hbox{ lower}.
\ea
\ee
The lower terms are reduced to zero by {\bf BG2} to {\bf BG6} in the
new letters $\1$ to $\6$ and their conjugates
in 834 steps.

Influence analysis of the submonomial roles of $\5$ and $\6$:
\bi
\item In Step 1, the $f_{14}$-relations include 
$\4[\bar{\2}\bar{\5}\bar{\3}\6\1]=[\bar{\2}\bar{\5}\bar{\3}\6\1]\4$, 
$\4[\3\5\bar{\1}]=[\3\5\bar{\1}]\4$, $\3[\2\6\1]=[\2\6\1]\3$, and
$\4[\2\6\1]=[\2\6\1]\4$, where $\5, \6$ are sandwiched between $\1$ and 
one of $\bar{\2}, \3, \2$ in each bracket.

\item In Step 2, first $\4\bar{\1}=-\4\1+[\1]\4$ is used to change 
the first term to $\4\1\bar{\6}\3\5\2+$ lower, and then the 
$\4\1$-controlled reduction
$\4\1[\bar{\6}\3\5]\2=\4\1\2[\bar{\6}\3\5]$ is applied.
The third term is changed by Proposition \ref{prop:11} as follows:
$\4\1\bar{\5}\bar{\3}\bar{\6}\bar{\2}
=\4\1\bar{\5}\bar{\3}\bar{\6}[\2]-\4\1\bar{\5}\bar{\3}\bar{\6}\2
=[\2]\4\1\bar{\5}\bar{\3}\bar{\6}-\4\1\bar{\5}\bar{\3}\bar{\6}\2$.
The submonomial roles of $\5$ and $\6$ do not influence the reduction.

\item In Step 3, $[\4\1]\2=\2[\4\1]$ is a $\overline{\bf BG}\3$-relation, and
$\3\5[\2\4\6\1]=\3[\2\4\6\1]\5=[\2\4\6\1]\3\5$
is a combination of $\3$-controlled reduction and $f_{14}$ relation.

\item In the last step, 
$[\4\1\bar{\5}\bar{\3}]\2=\2[\4\1\bar{\5}\bar{\3}]$ is an $f_{12}$ relation,
$\3\5\bar{\1}\bar{\4}[\2]=[\2]\3\5\bar{\1}\bar{\4}$ is by 
Proposition \ref{prop:11}, 
$[\3\5\bar{\1}]\bar{\2}=\bar{\2}[\3\5\bar{\1}]$ is an $f_{13}$ relation,
and $\3\5\bar{\1}[\4]\bar{\2}=\3\5\bar{\1}\bar{\2}[\4]$ followed by
$[\3\5\bar{\1}]\bar{\2}=\bar{\2}[\3\5\bar{\1}]$ is a combination of
$\3$-led reduction and $f_{13}$ relation.
\ei
Summing up, the submonomial roles of $\5, \6$ have no influence on the 
substitutional reduction.
\end{proof}

The sixth case is 
Figure \ref{fig:Spoly}(e), where
$\m_1\succ_M \3\bar{\3}$. There are three subcases: either
$\3\prec\m_1\prec \5$ (e.g., $\m_1=\4$), or
$\5\preceq\m_1\prec \7$ (e.g., $\m_1=\6$), or 
$\m_1\succ \7$ (e.g., $\m_1=\8$). In the following, we prove that
if we set $\m_1=\4$, and concatenate $\A_6, \m_2$ to form new letter $\6$,  
then the reduction of the S-polynomial, which is still $h_6$ in new letters,
remains valid in all the three subcases.

\bp
The S-polynomial corresponding to Figure \ref{fig:Spoly}(e)
can be reduced to zero by {\bf BG}.
\ep

\begin{proof} The order of $\4=\m_1$ 
may vary from above $\bar{\3}$ to above $\m_2$.
We analyze the influence of the varying order of $\4=\m_1$ in the
substitutional reduction of $h_6$ in the proof of 
Proposition \ref{prop:h6}.
\bi
\item In Step 1, the $f_{14}$-relations
$\4[\bar{\2}\bar{\5}\bar{\3}\6\1]=[\bar{\2}\bar{\5}\bar{\3}\6\1]\4$, 
$\4[\3\5\bar{\1}]=[\3\5\bar{\1}]\4$, $\3[\2\6\1]=[\2\6\1]\3$, and
$\4[\2\6\1]=[\2\6\1]\4$, remain to be $f_{14}$-relation for all 
letters $\4\succ_M \3\bar{\3}$.

\item In Step 2, the reductions $\4\bar{\1}=-\4\1+[\1]\4$,\
$\4\1[\bar{\6}\3\5]\2=\4\1\2[\bar{\6}\3\5]$, and
$\4\1\bar{\5}\bar{\3}\bar{\6}\bar{\2}
=[\2]\4\1\bar{\5}\bar{\3}\bar{\6}-\4\1\bar{\5}\bar{\3}\bar{\6}\2$ are valid for
all letters $\4\succ_M \2\bar{\2}$.

\item In Step 3, $[\4\1]\2=\2[\4\1]$ is valid for
all letters $\4\succ_M \2\bar{\2}$, and reduction
$\3\5[\2\4\6\1]=\3[\2\4\6\1]\5=[\2\4\6\1]\3\5$
is independent of the order of $\4$. 

\item In the last step, 
$[\4\1\bar{\5}\bar{\3}]\2=\2[\4\1\bar{\5}\bar{\3}]$ is an $f_{12}$ relation
for all letters $\4\succ_M \3\bar{\3}$, $\3$-led reduction
$\3\5\bar{\1}[\4]\bar{\2}=\3\5\bar{\1}\bar{\2}[\4]$, and
$\3\5\bar{\1}\bar{\4}[\2]=[\2]\3\5\bar{\1}\bar{\4}$ by Proposition \ref{prop:11}, are both
independent of the order of $\4$. 
\ei
Summing up, the varying order of letter $\4$ has no influence on the 
substitutional reduction.
\end{proof}

\section{Conclusion}
\label{sect:conc}

Quaternionic polynomials occur naturally in science and engineering applications, and normalization
of quaternionic polynomials is a basic task. This paper proposes the first readable proof
of a conjectured reduced Gr\"obner basis of the defining ideal of coordinate-free quaternionic
polynomial algebra, by developing some novel reduction techniques for free associative algebras.

The order of letters used in the certified Gr\"obner basis is the conjugate-alternating order,
where each quaternionic variable immediately precedes its conjugate. This order is friendly
for employing the well known fact that the scalar part of any quaternionic monomial commutes
with any letter, as twice the scalar part is the sum of the monomial and its conjugate. 
From this viewpoint, the certified Gr\"obner basis under this order reflects the 
commutativity and symmetry of the scalar part as the foundations of normalization.

The following is another typical order among the letters:
\be
\q_1\prec \q_2\prec \ldots \prec \q_n\prec \i\prec \j\prec \k\prec
\bq_1\prec \bq_2\prec \ldots \prec \bq_n.
\ee
It is called the {\it conjugate-separating order}, and is used to shield the conjugates from the
quaternionic variables. Recall that in multi-variate complex analysis, a complex polynomial is in 
complex variables $z_1, \ldots, z_n$ and their complex conjugates
$\bar{z}_1, \ldots, \bar{z}_n$, the polynomial is analytic if and only if it is irrelevant to
the conjugates. In the quaternionic domain, a quaternionic polynomial is said to be {\it monogenic}
if it is irrelevant to the conjugates of the quaternionic variables. To normalize a
monogenic quaternionic polynomial, a Gr\"obner basis of the defining ideal of 
basis-free quaternionic polynomial algebra under the conjugate-separating order is desired.

The work on conjecturing such a Gr\"obner basis is an ongoing work. By now the elements of low
degree of a reduced Gr\"obner basis have been computed, and have been classified in the framework
of Grassmann-Cayley algebra and bracket algebra \cite{li2008} over ${\mathbb K}^3$. 
An interesting observation is that
opposite to the case of conjugate-alternating order, the elements of the
Gr\"obner basis under conjugate-separating order reflects the 
anti-commutativity and asymmetry of the vector part in Grassmann-Cayley algebra
as the foundations of normalization.

\vskip .2cm
Correspondence author Hongbo LI is supported by
National Key R\&D Project 2020YFA0712300.

\end{document}